\documentclass[final, 12pt,oneside]{class_diss}

\usepackage{amsmath} 
\usepackage{graphicx}

\usepackage[all,cmtip]{xy}

\usepackage{amsxtra}
\usepackage{amssymb}
\usepackage{amsthm}
\usepackage{latexsym}
\usepackage{appendix}
\theoremstyle{plain}
\newtheorem{thm}{Theorem}
\newtheorem{conj}{Conjecture}
\newtheorem{defn}{Definition}

\usepackage{makeidx} \makeindex

\begin{document}

\thispagestyle{empty}

\begin{center}

   \vspace{1cm}

   \Large OPTICAL BLACK HOLES AND SOLITONS \\

   \vspace{0.5cm}

   by\\

   \vspace{0.5cm}

   Shawn Michael Westmoreland\\

   \vspace{0.5cm}

   B.S., University of Texas, Austin, 2001\\
   M.A., University of Texas, Austin, 2004\\

   \vspace{0.65cm}
   \rule{2in}{0.5pt}\\
   \vspace{0.85cm}

   {\Large A DISSERTATION}\\

   \vspace{0.5cm}
   submitted in partial fulfillment of the\\
   requirements for the degree\\

   \vspace{0.5cm}

   {\Large DOCTOR OF PHILOSOPHY}\\
   \vspace{0.5cm}

   Department of Mathematics\\
   College of Arts and Sciences\\

   \vspace{0.5cm}
   {\Large KANSAS STATE UNIVERSITY}\\
   Manhattan, Kansas\\

   2010\\
   \vspace{1cm}

\end{center}

\begin{flushleft}
   \hspace{10cm}Approved by:\\
   \vspace{ 1cm}
   \hspace{10cm}Major Professor\\

   \hspace{10cm}Louis Crane\\
\end{flushleft}
\newpage
\addcontentsline{toc}{chapter}{Abstract}
\chapter*{Abstract}
We exhibit  a static, cylindrically symmetric, exact solution to the Euler-Heisenberg field equations (EHFE) and prove that its effective geometry contains (optical) black holes.  It is conjectured  that there are also soliton solutions to the EHFE which contain  black hole  geometries.
\tableofcontents
\listoffigures
\listoftables

\newpage
\addcontentsline{toc}{chapter}{Acknowledgements}
\chapter*{Acknowledgements} 

I thank my advisor Louis Crane for his guidance and for many exploratory conversations which, in particular, lead him to suggest the present topic as a thesis. Working with him has  taught me valuable lessons about  perseverance and creative thinking in research.

I also wish to thank the following people for their input: Lei Cao, Renaud Gauthier,  IkJae Lee, Dany Majard, Charles Moore, Virginia Naibo, James Neill, Larry Weaver, and  David Yetter. 

Special loving thanks to my  mother. 

\newpage
\addcontentsline{toc}{chapter}{Dedication}
\chapter*{Dedication} 
I  dedicate this thesis in loving memory to my father, Michael Eugene Westmoreland; my brother, John Kamin Stewart; and friends Josh Watson, Chad Hunter, Ricky Valenzuela.

\chapter{Introduction and overview}\label{chapter0}

The underlying motivation of the present thesis is the idea of a mathematical  connection  between solitons and black holes. Connections of this sort have been considered elsewhere, though not in the same  context as the present thesis   \cite{Gibbons1986}. We will study a particular system of   nonlinear PDEs, arising from the Euler-Heisenberg field equations (EHFE), which we conjecture has solutions  uniting solitonic and black hole-like properties. The EHFE derive from the Euler-Heisenberg effective Lagrangian for quantum electrodynamics;  the solutions that we are interested in are optical black holes.   

\section{Solitons}

The theory of solitons arises from the study of wave phenomena in nonlinear PDEs. \index{solitons} A soliton is a solitary traveling wave that maintains its shape through time.  Of particular importance for us is the nonlinear Schr\"odinger equation (NSE).\index{nonlinear Schr\"odinger equation}  In 1-dimension, for a complex wave amplitude $\psi(t,x)$ with coupling constant $\gamma$,   the NSE has the canonical form (e.g., Sulem and Sulem   \cite{Sulem} pp. 5, 20, Drazin and Johnson      \cite{Drazin} pp. 34 - 35):
 \begin{eqnarray}\label{NSE1}
i\partial_t \psi  +  \partial_x^2 \psi + \gamma |\psi|^2 \psi =0. 
 \end{eqnarray} 
 
 \section{General relativity}

A \emph{spacetime} is a 4-dimensional pseudo-Riemannian manifold.\index{spacetime} 
The key equation in general relativity\index{general relativity} 
 is the Einstein field  equation (EFE),\index{Einstein field equation} which unfolds to give a system of  nonlinear PDEs. Solving these PDEs allows one to express the  metric coefficients $g_{\mu\nu}$ of spacetime in  terms of the stress-energy tensor $T_{\mu\nu}$.  With the  cosmological constant $\Lambda$ included, the EFE reads  (Hawking and Ellis   \cite{HawkingEllis} p. 74):\index{field equations! Einstein}
\begin{eqnarray}\label{EFE1}
G_{\mu\nu}  + \Lambda g_{\mu\nu} =  8\pi\mathcal{G}  T_{\mu\nu},
\end{eqnarray}  where $\mathcal{G}$ is Newton's gravitational constant and the speed of light is   set equal to 1. The Einstein tensor $G_{\mu\nu}$ can be expressed in terms of  the Ricci tensor $R_{\mu\nu}$, the scalar $R=R_{\mu}^{\phantom{\mu}\mu}$, and the metric $g_{\mu\nu}$:
\begin{eqnarray}
G_{\mu\nu} = R_{\mu\nu} - \frac{R}{2}g_{\mu\nu}.
\end{eqnarray} 

Among the  known exact solutions to (\ref{EFE1})  are those describing  black holes. 
\begin{defn}
A black hole is a region of spacetime where future-directed outgoing null geodesics cannot escape.\index{black holes}
\end{defn}
\begin{defn}
A white hole is a region of spacetime where future-directed ingoing null geodesics  cannot enter.\index{white holes}
\end{defn} 
\begin{defn}The boundary of a black (or white) hole is called an event horizon.
\end{defn}\index{event horizon}We use these definitions even outside the context of general relativity. For us, any  pseudo-Riemannian manifold  will be called a spacetime whether it satisfies the EFE or not, and one can ask whether or not a given spacetime has black holes.

\section{The Euler-Heisenberg field equations}
Quantum electrodynamics\index{quantum electrodynamics} can be approximated as an effective field theory governed by the Euler-Heisenberg Lagrangian (cf. Euler and Heisenberg   \cite{EH}, Schwinger    \cite{Schwinger}, Novello   \cite{ABH} p. 292, Boer and van Holten   \cite{Boer}): \index{Lagrangian! Euler-Heisenberg} \index{Euler-Heisenberg! Lagrangian}
\begin{eqnarray}\label{EHL1}
L&=& -\frac{1}{4} F + \frac{\alpha^2}{90}\left(F^2 + \frac{7}{4}G^2\right),
\end{eqnarray} where $\alpha$ is the fine structure constant,  and $F$ and $G$ are the Poincar\'e invariants of the electromagnetic field.  The speed of light, the reduced Planck  constant, the mass of the electron, and the permittivity of free space are here set equal to 1. Applying the principle of least action to (\ref{EHL1})  leads to the Euler-Heisenberg field equation (EHFE):\index{Euler-Heisenberg! field equations}\index{field equations! Euler-Heisenberg}
\begin{eqnarray}\label{EHFE1}
\nabla_\mu F^{\mu\nu} &=& \frac{\alpha^2}{45}\nabla_\mu\left(4F F^{\mu\nu} + 7G{F^*}^{\mu\nu}\right).
\end{eqnarray} Here, $F^{\mu\nu}$ is the electromagnetic field tensor, ${F^*}^{\mu\nu}$ is its dual, and $\nabla_\mu$ represents the covariant derivative with respect to the coordinate $x^\mu$, using the connection determined  by the background spacetime metric. Note that (\ref{EHFE1}) is, in general, yet another  system of nonlinear PDEs.

According to the Euler-Heisenberg effective field theory, the  vacuum behaves like a nonlinear physical  medium.  Light rays passing through electromagnetic fields are bent as if they were passing  through water, thus affecting the apparent geometry of objects. This motivates the idea that the effective field theory can be interpreted geometrically. Indeed, Novello   \cite{ABH} has shown in a seminal work that light rays (small disturbances traveling through the field)  follow   null geodesics with respect to a spacetime metric $\widetilde{g}_{\mu\nu}$  distinct from the background metric $g_{\mu\nu}$. This is called the effective metric. The  inverse (or cometric) $\widetilde{g}^{\mu\nu}$ can be expressed in terms of  the stress-energy tensor of the electromagnetic field:
\begin{eqnarray}\label{Novellogeometricequation}
\widetilde{g}^{\mu\nu}&=& \mathcal{A} g^{\mu\nu} + \mathcal{B}T^{\mu\nu},
\end{eqnarray} where $\mathcal{A}$ and $\mathcal{B}$ are special functions of the Poincar\'e invariants $F$ and $G$ (see Equation (\ref{effcometricSET})). This geometrical interpretation of  effective field theory demonstrates an analogy between nonlinear optics and general relativity, with  Equation (\ref{Novellogeometricequation}) playing the role of  the  Einstein field equation  (\ref{EFE1})   \cite{ABH}.  Let us note two subtleties. (1) The effective metric is uniquely determined only up to a conformal factor. (2) Since light in the nonlinear vacuum experiences birefringence, a given electromagnetic field actually carries two distinct effective metrics; one for each polarization state.

\section{The idea}

Our main proposal is that the EHFE (\ref{EHFE1}) has soliton solutions with a  corresponding effective geometry containing a black hole.  In this sense, the EHFE would be  somewhat in between the NSE (\ref{NSE1}) and the EFE (\ref{EFE1}).  We have a theorem and a conjecture:
\begin{thm}\label{mainthm} There is an exact static solution to the Euler-Heisenberg field equations where the effective geometries of each polarization state have black holes.
\end{thm} 
\begin{conj}\label{mainconj}
There is  an exact solution to the Euler-Heisenberg field equations which is a soliton and whose effective geometries have  black holes. 
\end{conj}

The theorem is   proven in Section \ref{proofmaintheorem}.  The soliton of the conjecture is, we believe, an imploding solitonic wave. Evidence for this belief is discussed in the section below. The purist will note that the hypothetical ``imploding soliton" cannot strictly be a soliton as a soliton  does not change its shape. \index{solitons} As the soliton in the conjecture implodes, it will become more concentrated and lose its initial shape. However,  in cylindrical coordinates $(t,r,\theta,z)$, for  a cylindrically symmetric  wave approaching the axis $r=0$, we conjecture that  its radial cross section will keep its shape if multiplied by  $r$.

\section{Evidence for the conjecture}
 Evidence in support of the conjecture comes from  nonlinear optics   \cite{DelphenichNOA}. In particular, Solja\v ci\'c and Segev   \cite{SoljacicSegev} have examined the behavior of a beam resulting from the   perpendicular collision of two plane waves in the Euler-Heisenberg vacuum. Using approximations,  they determined that the amplitude of the resulting beam satisfies the NSE   (\ref{NSE1}). \index{nonlinear Schr\"odinger equation} Since an imploding wave can be thought of as a limiting case of infinitely many colliding plane waves, it seems reasonable to expect that the NSE should be obtained in the case of an imploding wave. 

In the work of Brodin \emph{et al.}   \cite{Brodin}, which nicely complements  Solja\v ci\'c and Segev's paper, a beam guided between  two parallel conducting planes is studied. It was found that the amplitude of this beam satisfies a 2-dimensional cylindrically symmetric  NSE. \index{nonlinear Schr\"odinger equation}  According to  Brodin \emph{et al.}   \cite{Brodin}, for a beam with a  certain critical intensity $I_c$, the dispersive and self-focusing effects exactly balance and the beam forms an optical soliton of constant width.\index{solitons} If the intensity $I$ of the beam is less than $I_c$, then the beam width diffracts without bound. If $I>I_c$, then  the beam width collapses to zero in a finite time.  These results from nonlinear optics show that there is an authentic link  between the EHFE and the NSE, which at least partly supports Conjecture \ref{mainconj}. 

Another piece of evidence for the conjecture comes from the work of  Section \ref{maxwellianapproximation}. The Maxwellian approximation, although it is only a first-order approximation to a solution to the EHFE, it gives information on the coordinate velocities of effective geodesics up to second-order (see Theorem \ref{justificationofMA}).  When we look at  the coordinate velocities of the outgoing geodesics to second-order, we find that they are trapped within a certain radius. (There is a black hole.) 

Since Conjecture \ref{mainconj}  concerns solutions of a nonlinear variational problem,  we suspect that   its proof will use tools from   Morse theory (i.e., the calculus of variations in the large).

\section{Organization}

This thesis is organized as follows.

Chapter \ref{chapter1} is a  pedestrian introduction to the required mathematical physics.  Chapter \ref{chapter2} is a self-contained review of  Novello's theory of effective geometry. In Chapter \ref{planewaves} (which can be omitted on a first reading), we  study the effective geometry of plane waves, and calculate the index of refraction through a plane wave confirming earlier approximations done by others using different methods. In Chapter \ref{chapter4},  we use  well-known solutions from Maxwell's theory (for imploding cylindrically symmetric waves) and investigate the corresponding effective geometries which resemble black holes. At the end of Chapter \ref{chapter4}, we prove Theorem \ref{mainthm} by explicitly finding an exact solution to the EHFE with the required properties. Although this exact solution is static, it shares some similarities with an ingoing cylindrical wave solution because its Poynting vector points radially inward.

\chapter{Preliminaries}\label{chapter1}
The present chapter is meant to be a  self-contained pedestrian introduction to the relevant mathematical physics. 

\section{Nonlinearity of the vacuum} 
According to quantum electrodynamics,\index{quantum electrodynamics} photons can scatter off of each other. This photon-photon scattering effect, also known as the  nonlinearity of the vacuum,\index{nonlinearity of the vacuum} was calculated by Euler and Heisenberg   \cite{EH} in the mid-1930s, but it is so subtle that no currently available  experiment   is yet sensitive enough to measure it. 

\begin{figure}[h]
\center
    \includegraphics[height=4.5in]{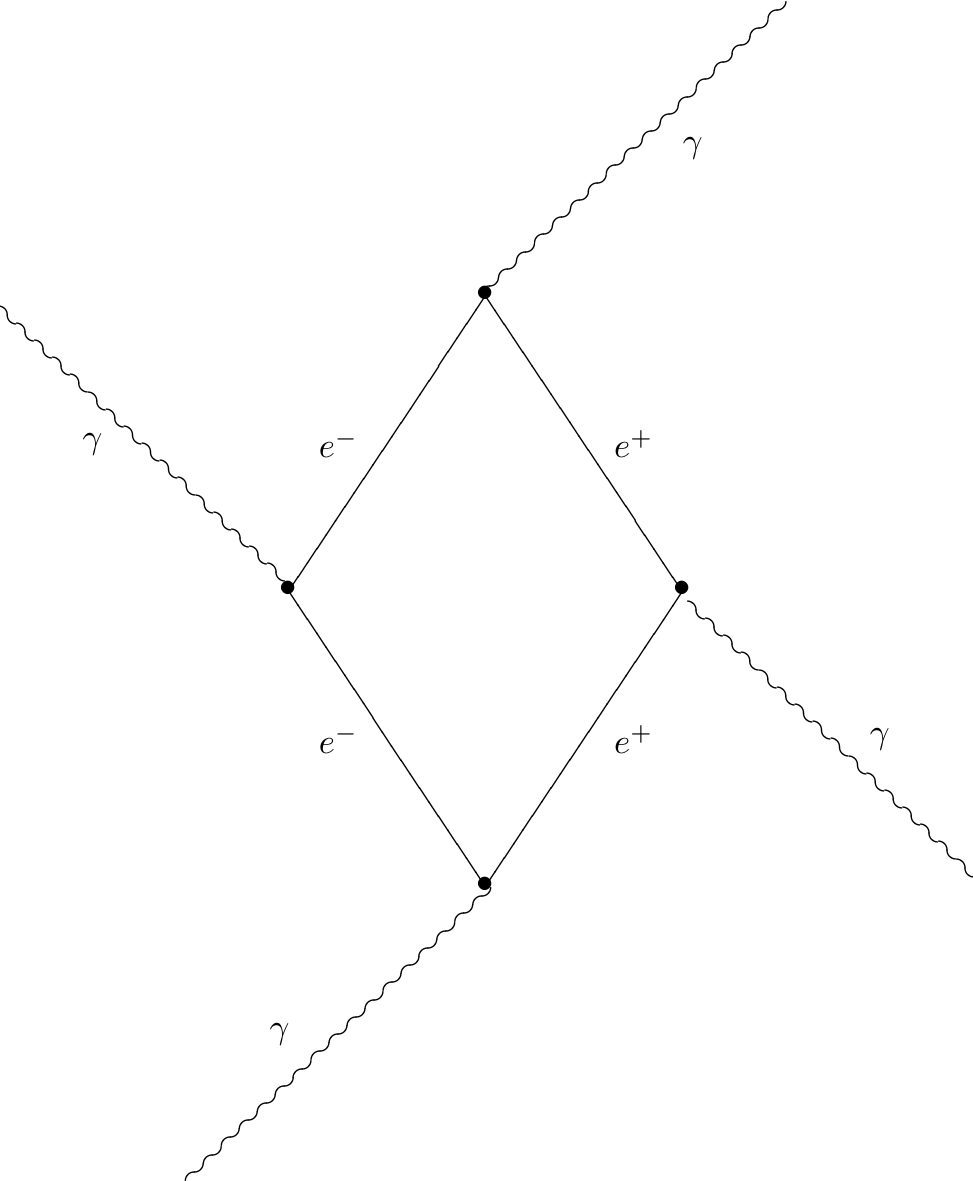}

    \caption[Photon-photon scattering]{Feynman diagram of photon-photon scattering.}

    \label{PhotonPhotonDiagram}
\end{figure}

Photon-photon scattering arises from processes that  involve virtual electron-positron pairs (see Figure \ref{PhotonPhotonDiagram}).  In the classical limit, the effect of these virtual particles on real  photons can be approximated  by introducing nonlinear terms to the Maxwellian Lagrangian. In this so-called effective field theory,   the photons do not necessarily follow null geodesics in the background metric. Instead, they follow null geodesics  with respect to a so-called effective metric, as will be explained  in Chapter \ref{chapter2}.  
\section{The background spacetime}

In a curved background, according to  Drummond and Hathrell   \cite{DrummondHathrell}, the physical Lagrangian for the effective field theory acquires  a nontrivial dependence on the     spacetime curvature. 
In the present thesis, we restrict ourselves  to a  Minkowskian background, so these curvature coupling effects can be ignored.

The background metric tensor \index{metric! background} is denoted by $g_{\mu\nu}$.\index{background metric}  Its inverse,\index{metric! inverse (cometric)} the so-called \emph{cometric,} \index{cometric! background} is denoted $g^{\mu\nu}$ and one has that  $g^{\mu\lambda}g_{\lambda\nu} = \delta^\mu_{\phantom{\mu}\nu}$, where $\delta^\mu_{\phantom{\mu}\nu}$ is the Kronecker delta (and the Einstein summation convention is  followed as usual). The present work uses the $+---$ signature convention.

Recall that a given metric determines a 
unique  torsion-free connection $\nabla$ \index{connection! derived from  background metric} by requiring  that the covariant derivative \index{covariant derivative! derived from  background metric} of the metric be zero (e.g. Hawking and Ellis   \cite{HawkingEllis} p. 40,  or  Spivak   \cite{Spivak} pp. 236 - 237):
\begin{eqnarray}
\nabla_\lambda g_{\mu\nu} = 0.
\end{eqnarray}The connection coefficients $\Gamma^\lambda_{\phantom{\lambda}\mu\nu}$ (Christoffel symbols) \index{Christoffel symbols (connection coefficients)! derived from background metric} thus  determined are given by:
\begin{eqnarray} \label{Christoffel}
\Gamma^\lambda_{\phantom{\lambda}\mu\nu}= \frac{1}{2}g^{\lambda\alpha}\left(\partial_\mu g_{\alpha\nu}  + \partial_\nu g_{\alpha\mu} -\partial_{\alpha} g_{\mu\nu}\right),
\end{eqnarray}where the operator $\partial_\mu$ denotes partial differentiation with respect to the $x^\mu$ coordinate. I.e., $\partial_\mu := \partial/\partial x^\mu$. 

\section{Electromagnetic fields}

The electromagnetic field $F_{\mu\nu}$ \index{electromagnetic field} is  a closed 2-form, and locally there exists an $A$-field $A_\mu$ such that:
\begin{eqnarray}\label{F=dA}
F_{\mu\nu} &=& \partial_\mu A_\nu - \partial_\nu A_\mu \ \ \
\left(=\nabla_\mu A_\nu - \nabla_\nu A_\mu\right).
\end{eqnarray}

Note  the   Bianchi  identity: \index{Bianchi identity}
\begin{eqnarray}\label{Bianchi}0=
 \partial_\lambda F_{\mu\nu} + \partial_\nu F_{\lambda\mu} + \partial_\mu F_{\nu\lambda} = \ ( \nabla_\lambda F_{\mu\nu} + \nabla_\nu F_{\lambda\mu} + \nabla_\mu F_{\nu\lambda}).
\end{eqnarray}

The antisymmetrization \index{antisymmetrization} of an indexed quantity is indicated by placing square brackets around indices. For example,  given the tensor $T_{\alpha_1\alpha_2\cdots \alpha_n}$, we can write:
\begin{eqnarray}
T_{[\alpha_1\alpha_2\cdots \alpha_n]}=\frac{1}{n!}\sum_{\sigma} \textrm{sgn}(\sigma)T_{\sigma(\alpha_1)\sigma(\alpha_2)\cdots \sigma(\alpha_n)},
\end{eqnarray} where the sum is taken over all permutations $\sigma$  of the indices $\alpha_1,\alpha_2,\cdots,\alpha_n$. The value of $\textrm{sgn}(\sigma)$ is $+1$ if $\sigma$ is an even permutation of the sequence  $\alpha_1\alpha_2\cdots\alpha_n$ and is $-1$ if $\sigma$ is an odd permutation. 

Using antisymmetrization, Equation (\ref{Bianchi}) can be written as: \index{Bianchi identity}
\begin{eqnarray}0=
\partial_{[\lambda}F_{\mu\nu]} \ (= \nabla_{[\lambda}F_{\mu\nu]}).
\end{eqnarray}

The Levi-Civita tensor $\varepsilon_{\alpha\beta\mu\nu}$ \index{Levi-Civita tensor} is defined such that:
\begin{eqnarray}
\varepsilon_{\alpha\beta\mu\nu}:= 4!\sqrt{|g|}\delta^0_{\phantom{0}[\alpha}\delta^1_{\phantom{1}\beta}\delta^2_{\phantom{2}\mu}\delta^3_{\phantom{3}\nu]},
\end{eqnarray} where $g$ denotes the determinant of the metric $g_{\mu\nu}$. Note that the value of $\varepsilon_{\alpha\beta\mu\nu}$ is $0$ unless the indices  $\alpha,\beta,\mu,\nu$ are all distinct.  Furthermore, note that  $\varepsilon_{\alpha\beta\mu\nu}$ is $+\sqrt{|g|}$ if the sequence $\alpha\beta\mu\nu$ is an even permutation of the sequence 0123 and is  $-\sqrt{|g|}$ if $\alpha\beta\mu\nu$ is an odd permutation.

The Levi-Civita tensor allows us to express the Hodge dual of $F_{\mu\nu}$ by writing:\index{Hodge dual}
\begin{eqnarray}\label{ABH1}
F^*_{\alpha\beta} := \frac{1}{2}\varepsilon_{\alpha\beta\mu\nu}F^{\mu\nu}.
\end{eqnarray} Unless otherwise specified (e.g., in Section \ref{nullgeo}), indices are always raised or lowered with respect to the \emph{background} metric. So, e.g.,  $F^{\mu\nu}=g^{\alpha\mu}g^{\beta\nu}F_{\alpha\beta}$.

Note that the Bianchi identity (\ref{Bianchi}) can be expressed in terms of the dual tensor by writing (cf. Landau and Lifshitz   \cite{Landau} p. 67):\index{Bianchi identity}
\begin{eqnarray}\label{Bianchi2}
0=
\partial_\mu {F^*}^{\mu\nu} \ (= \nabla_\mu {F^*}^{\mu\nu}).
\end{eqnarray}

\section{Effective Lagrangians}
The physical behavior of the electromagnetic   field is  governed by a Lagrangian $L$ which is a scalar function of the field $A_\mu$, its covariant derivatives $\nabla_\mu A_\nu$, and the background metric. 
The field equations are obtained from the principle of least action:\index{field equations! derived from Lagrangian}
\begin{eqnarray}\label{eulerlagrangeeqns}
\nabla_\mu\frac{\partial L}{\partial(\nabla_\mu A_\nu)} = \frac{\partial L}{\partial A_\nu}.
\end{eqnarray}

In the case of Maxwell's theory, the Lagrangian in the absence of charges and currents  (a field in the vacuum) can be written out as:\index{Lagrangian! Maxwell (vacuum case)}\index{Maxwell! Lagrangian (vacuum case)}
\begin{eqnarray}\label{maxwellslagrangian}
L&=&-\frac{1}{4}g^{\alpha\mu}g^{\beta\nu}\left(\nabla_\mu A_\nu - \nabla_\nu A_\mu \right)(\nabla_\alpha A_\beta - \nabla_\beta A_\alpha).
\end{eqnarray} If we define   $F:=F_{\mu\nu}F^{\mu\nu}$, then Equation (\ref{maxwellslagrangian}) simplifies to $L = -F/4$.  Using (\ref{eulerlagrangeeqns}), one recovers the familiar form of Maxwell's equation for the vacuum:\index{field equations! Maxwell (vacuum case)}\index{Maxwell!  field equations (vacuum case)}
\begin{eqnarray}\label{maxwelleqns}
\nabla_\mu F^{\mu\nu} = 0.
\end{eqnarray}

The quantity $F$ introduced here is a scalar invariant of the electromagnetic field tensor. In fact, there are  only two algebraically independent scalar invariants for  $F_{\mu\nu}$. These are represented by the so-called \emph{Poincar\'e  invariants} $F:=F_{\mu\nu}F^{\mu\nu}$ and $G:= F^{\mu\nu}F^*_{\mu\nu}$ (Landau and Lifshitz   \cite{Landau} p. 64).\index{Poincar\'e  invariants}

We define the class of  \emph{$L(F,G)$-theories} as electromagnetic theories in which the Lagrangian  $L$ \index{$L(F,G)$-theories! defined}   can be expressed in terms of the Poincar\'e   scalars $F$ and $G$ alone, i.e., $L=L(F,G)$. \index{Lagrangian! type $L=L(F,G)$} It is assumed  that the partial derivatives of $L=L(F,G)$, with respect to $F$ and $G$, exist at least up to second-order, and that they are continuous. We will use the notations $L_F:= \partial L/\partial F$, $L_G:=\partial L/\partial G$, $L_{FF} := \partial^2 L/\partial F^2$, $L_{GG} := \partial^2 L/\partial G^2$, $L_{FG} := \partial^2 L/(\partial G\partial F)$, etc.  We observe that  $L(F,G)$-theories are guaranteed to be Lorentz invariant since both  $F$ and $G$ are Lorentz invariant quantities. \index{Lagrangian! Lorentz invariance of} A particularly important $L(F,G)$-theory is the Euler-Heisenberg theory  (to be introduced in Section \ref{ehlagrangiansection}). 

 Chapter \ref{chapter2} is primarily concerned not with a particular theory but with general $L(F,G)$-theories. However, in subsequent chapters, attention is restricted to the Euler-Heisenberg theory. 
 
Although effective field theories more elaborate than the $L(F,G)$-type can be constructed by writing  Lagrangians that include terms involving the covariant derivatives of the field (e.g., terms like $\nabla_\lambda F^\lambda_{\phantom{\lambda}\nu}\nabla_\mu F^{\mu\nu}$)   \cite{SoljacicSegev}, we will not deal with such things in the present work.

\section{Euler-Heisenberg theory}\label{ehlagrangiansection}
Euler-Heisenberg theory is an effective field theory which   approximates the physical theory of quantum electrodynamics in  Minkowski spacetime.
 
 Up to second order in the fine-structure constant $\alpha$,  the \emph{Euler-Heisenberg Lagrangian} is given by (cf. Euler and Heisenberg   \cite{EH}, Schwinger   \cite{Schwinger}, Novello    \cite{ABH} p. 292, Boer and van Holten   \cite{Boer}): \index{Lagrangian! Euler-Heisenberg} \index{Euler-Heisenberg! Lagrangian}
\begin{eqnarray}\label{EH}
L = -\frac{1}{4}F + \frac{\alpha^2}{90}\left(F^2 + \frac{7}{4}G^2\right).
 \end{eqnarray} Since we are presently working in  natural  units,   the speed of light $c$, the reduced Planck constant $\hbar$, the mass of the electron $m_e$, and the permittivity of free space $\epsilon_0$, are here set equal to 1.  \index{units! natural} \index{natural units} Conventional Lorentz-Heaviside   units can be restored by replacing $\alpha^2$ with $ \alpha^2 \hbar^3 m_e^{-4} c^{-5}$.\index{units! Lorentz-Heaviside} \index{Lorentz-Heaviside units} (Sometimes we use the letter $\alpha$ as a tensor or   pseudotensor index, but no confusion between $\alpha$ as the fine-structure constant and   $\alpha$ as an index should arise  because  the context will make the meaning of $\alpha$ clear.)
 
Equation (\ref{EH}) applies to fields having  strength $A$ and frequency $\omega$ such that   \cite{EH}:
\begin{eqnarray}\label{physicalrestrictions1}
A&\ll& \frac{1}{\sqrt{4\pi \alpha}}\\
\omega&\ll&  1.\label{physrest2}
\end{eqnarray} In other words, the field strength should be  much weaker than the  critical field $1/\sqrt{4\pi \alpha}$   and it should be  approximately constant on scales much less than the Compton wavelength of the electron (which is unity, in our units). In the present work however, we will not worry  about these physical restrictions (\ref{physicalrestrictions1}) and (\ref{physrest2}).

\section{Field equations}
\begin{thm}
For a Lagrangian of type $L=L(F,G)$, the principle of least action (\ref{eulerlagrangeeqns}) yields  the field equations: \index{field equations! derived from Lagrangian}
\begin{eqnarray}\label{EOM2}
\nabla_\mu\left(L_F F^{\mu\nu} + L_G {F^*}^{\mu\nu}\right)=0.
\end{eqnarray} 
\end{thm}
\begin{proof}For a Lagrangian of type $L=L(F,G)$, we get that $\partial L/\partial A_\nu = 0$. So Equation (\ref{eulerlagrangeeqns}) reduces to
\begin{eqnarray}
\nabla_\mu \frac{\partial L}{\partial(\nabla_\mu A_\nu)} =0.
\end{eqnarray}
We get that:
\begin{eqnarray}\label{th1.0}
\frac{\partial L}{\partial(\nabla_\mu A_\nu)} = L_F \frac{\partial F}{\partial(\nabla_\mu A_\nu)} + L_G\frac{\partial G}{\partial(\nabla_\mu A_\nu)}.
\end{eqnarray} because:
\begin{eqnarray}\label{th1c1}
\frac{\partial F}{\partial(\nabla_\mu A_\nu)} =4F^{\mu\nu},
\end{eqnarray}and:
\begin{eqnarray}\label{th1cl2}
\frac{\partial G}{\partial(\nabla_\mu A_\nu)} =4{F^*}^{\mu\nu}.
\end{eqnarray} When an index $\mu$ is free and not to be summed over by the Einstein convention, we draw a bar over it:
\begin{eqnarray}
\frac{\partial F}{\partial(\nabla_{\bar\mu} A_{\bar\nu})}&=& \frac{\partial}{\partial(\nabla_{\bar\mu} A_{\bar\nu})}\left(F_{\alpha\beta}F^{\alpha\beta}\right)\nonumber\\
&=&\frac{\partial}{\partial(\nabla_{\bar\mu} A_{\bar\nu})}\left(g^{\alpha\lambda}g^{\beta\rho}F_{\alpha\beta}F_{\lambda\rho}\right)\nonumber\\
&=&\frac{\partial}{\partial(\nabla_{\bar\mu} A_{\bar\nu})}\left(g^{\alpha\lambda}g^{\beta\rho}\left(\nabla_\alpha A_\beta - \nabla_\beta A_\alpha\right)\left(\nabla_\lambda A_\rho -\nabla_\rho A_\lambda \right)\right)\nonumber\\
&=&\frac{\partial}{\partial(\nabla_{\bar \mu} A_{\bar \nu})}\left(g^{{\bar\mu}\lambda}g^{{\bar\nu}\rho}\left(\nabla_{\bar\mu} A_{\bar\nu} - \nabla_{\bar\nu}A_{\bar\mu}\right)\left(\nabla_\lambda A_\rho -\nabla_\rho A_\lambda \right)\right.\nonumber\\
&&\phantom{\frac{\partial}{\partial(\nabla_{\bar \mu} A_{\bar \nu})}}+g^{{\bar\nu}\lambda}g^{{\bar\mu}\rho}\left(\nabla_{\bar\nu} A_{\bar\mu} - \nabla_{\bar\mu} A_{\bar\nu}\right)\left(\nabla_\lambda A_\rho -\nabla_\rho A_\lambda \right)\nonumber\\
&&\phantom{\frac{\partial}{\partial(\nabla_{\bar \mu} A_{\bar \nu})}}+g^{\alpha{\bar\mu}}g^{\beta{\bar\nu}}\left(\nabla_\alpha A_\beta - \nabla_\beta A_\alpha\right)\left(\nabla_{\bar\mu} A_{\bar\nu} -\nabla_{\bar\nu} A_{\bar\mu} \right)\nonumber\\
&&\phantom{\frac{\partial}{\partial(\nabla_{\bar \mu} A_{\bar \nu})}}+g^{\alpha{\bar\nu}}g^{\beta{\bar\mu}}\left(\nabla_\alpha A_\beta - \nabla_\beta A_\alpha\right)\left(\nabla_{\bar\nu} A_{\bar\mu} -\nabla_{\bar\mu} A_{\bar\nu} \right)\nonumber\\
&=&g^{\bar\mu  \lambda}g^{\bar\nu   \rho}\left(\nabla_\lambda A_\rho  - \nabla_\rho A_\lambda\right)\nonumber\\
&&-g^{\bar\nu \lambda}g^{\bar\mu \rho}\left(\nabla_\lambda A_\rho  - \nabla_\rho A_\lambda\right)\nonumber\\
&&+g^{\alpha \bar\mu}g^{\beta\bar\nu }\left(\nabla_\alpha A_\beta  - \nabla_\beta A_\alpha\right)\nonumber\\
&&+g^{\alpha \bar\nu}g^{\beta\bar\mu }\left(\nabla_\alpha A_\beta  - \nabla_\beta A_\alpha\right)\nonumber\\
&=&F^{\bar\mu\bar\nu}-F^{\bar\nu\bar\mu}+F^{\bar\mu\bar\nu}-F^{\bar\nu\bar\mu}\nonumber\\
&=&4F^{\bar\mu\bar\nu}.\nonumber
\end{eqnarray}

Equation (\ref{th1cl2}) can be established similarly. 

The result (\ref{EOM2}) follows from Equations (\ref{th1.0}) - (\ref{th1cl2}).
\end{proof}

Specializing  (\ref{EOM2})  to the Euler-Heisenberg Lagrangian (\ref{EH}), we get  the Euler-Heisenberg field equations:\index{field equations! Euler-Heisenberg}\index{Euler-Heisenberg! field equations}
\begin{eqnarray}\label{EHFE}
\nabla_\mu F^{\mu\nu} = \frac{\alpha^2}{45}\nabla_\mu\left(4F F^{\mu\nu} + 7 G {F^*}^{\mu\nu}\right).
\end{eqnarray}

For an arbitrary $L(F,G)$-theory, Equation (\ref{EOM2}) implies:
\begin{eqnarray}\label{EOM3}
0&=&
\nabla_\mu\left(L_F F^{\mu\nu}\right) + \nabla_\mu\left(L_G{F^*}^{\mu\nu}\right)\nonumber\\
&=&\left(L_{FF}\nabla_\mu F + L_{FG}\nabla_\mu G\right)F^{\mu\nu}
+ L_F\nabla_\mu F^{\mu\nu}\nonumber\\
&&+\left(L_{FG}\nabla_\mu F + L_{GG}\nabla_\mu G\right){F^*}^{\mu\nu}
+ L_G\underbrace{\nabla_\mu {F^*}^{\mu\nu}}_{zero}.
\end{eqnarray} The last term is zero by the Bianchi identity  (\ref{Bianchi2}). 
By computation, one notes that $\nabla_\mu F = 2F^{\alpha\beta}\nabla_\mu F_{\alpha\beta}$ and $\nabla_\mu G = 2{F^*}^{\alpha\beta}\nabla_\mu F_{\alpha\beta}$. Hence, if we define the tensor   \cite{DeLorenci}:
\begin{eqnarray}\label{Q}
Q^{\alpha\beta\mu\nu}:= L_{FF} F^{\alpha\beta}F^{\mu\nu} + L_{FG}\left(F^{\alpha\beta}{F^*}^{\mu\nu} + {F^*}^{\alpha\beta}F^{\mu\nu}\right) + L_{GG}{F^*}^{\alpha\beta}{F^*}^{\mu\nu},
\end{eqnarray} then we can rewrite the field equations (\ref{EOM2}) as: 
\begin{eqnarray}\label{EOM4}
L_F\nabla_\mu F^{\mu\nu} + 2Q^{\alpha\beta\mu\nu}\nabla_\mu F_{\alpha\beta} =0.
\end{eqnarray}

Assuming $L_F\neq 0$, Equation (\ref{EOM4}) can be rearranged:
\begin{eqnarray}\label{vaccurrent}
\nabla_\mu F^{\mu\nu} = -\frac{2}{L_F} Q^{\alpha\beta\mu\nu}\nabla_\mu F_{\alpha\beta}. 
\end{eqnarray} (The case $L_F=0$ is discarded since it is not physically interesting.)

\section{Stress-energy tensor}

Given a Lagrangian $L$, one can define  a stress-energy tensor $T_{\mu\nu}$ through the equation:
\begin{eqnarray}\label{sefromL}
T_{\mu\nu} := 2\frac{\partial L}{\partial g^{\mu\nu}} - Lg_{\mu\nu}.
\end{eqnarray} 
This expression for the stress-energy tensor is implicit in e.g. Novello   \cite{ABH} pp. 271, 275,  Landau and Lifshitz   \cite{Landau} p. 77,   Hawking and Ellis   \cite{HawkingEllis} p. 66,  and Poisson  \cite{Poisson} p. 125. Observe  that different  authors disagree on the overall sign on $T_{\mu\nu}$ due to differing signature conventions for the metric.

For a Lagrangian of the form $L=L(F,G)$, Equation (\ref{sefromL}) gives:
\begin{eqnarray}
T_{\mu\nu} &=& 2\left(L_F\frac{\partial F}{\partial g^{\mu\nu}} + L_G\frac{\partial G}{\partial g^{\mu\nu}}  \right) -L g_{\mu\nu} \nonumber\\
&=& -4L_F F_\mu^{\phantom{\mu}\alpha} F_{\alpha\nu} -4L_G F_\mu^{\phantom{\mu}\alpha} F^*_{\alpha\nu} -L g_{\mu\nu}.
\end{eqnarray} Using the well-known identity $4F_\mu^{\phantom{\mu}\alpha}F^*_{\alpha\nu} = -Gg_{\mu\nu}$ (cf. Novello   \cite{ABH} p. 272), we get  (in agreement with Novello   \cite{ABH} p. 275):
\begin{eqnarray}\label{SETensor1}
T_{\mu\nu} = -4L_F F_\mu^{\phantom{\mu}\alpha}F_{\alpha\nu} - (L-GL_G)g_{\mu\nu}.
\end{eqnarray}

\chapter{Effective geometries
}\label{chapter2}

The purpose of this chapter is to give a quick self-contained review of  Novello's   theory of  effective geometries in nonlinear electrodynamics   \cite{ABH,DeLorenci,Novello2,NovelloSalim, NovelloPerezBergliaffa}.  Our exposition  is informed by the existing literature, most notably the work of Novello  \cite{ABH}. However, we do not follow any  specific work too closely.

\section{Electromagnetic shock waves } 
The \emph{wave front} of an electromagnetic shock wave \index{shock wave} is defined by a hypersurface $\Sigma$ across which the field derivatives are discontinuous. Given a set of local coordinates $x^\mu$ for the (background) spacetime manifold,  this hypersurface $\Sigma$ can be described as the set of solutions to the equation:
\begin{eqnarray}
z(x^\mu) = 0.
\end{eqnarray}We will need to assume that the first-order partial derivatives of  $z(x^\mu)$ exist and are continuous on  $\Sigma$, and that the gradient $k_\mu := \partial_\mu z$ does not vanish on $\Sigma$.  The hypersurface $\Sigma$, at least locally, splits the manifold into two regions $\mathfrak{M}^+:=\{x^\mu:\ z(x^\mu)>0\}$ and $\mathfrak{M}^-:=\{x^\mu:\ z(x^\mu)<0\}$. 

The jump of an arbitrary function \index{jump of a function} $J$ through  $\Sigma$ is denoted  by the \emph{Hadamard bracket} $[J]_\Sigma$. \index{Hadamard bracket} For each point $p$ of $\Sigma$, we define:
\begin{eqnarray}
[J]_\Sigma (p) := \lim_{p^+\rightarrow p}J(p^+) - \lim_{p^- \rightarrow p}J(p^-),
\end{eqnarray}
where the points $p^+$ and $p^-$, which tend towards $p$, belong to  the regions  $\mathfrak{M}^+$ and $\mathfrak{M}^-$ respectively (Papapetrou  \cite{Papapetrou} p. 170). Note that if $J$ is continuous across $\Sigma$, then $[J]_\Sigma= 0$. The converse is not strictly true, since it is possible to have a function with a so-called \emph{simple discontinuity} \index{simple discontinuity} 
whereby  $\lim_{p^+\rightarrow p} J (p^+)= \lim_{p^-\rightarrow p}J(p^-)\neq J(p)$. On the other hand, the derivative of a function can be discontinuous but the discontinuity is never of the simple type (Rudin   \cite{Rudin} p. 109).\index{discontinuity of simple type} In light of this,   a partial derivative $\partial_\mu J$ is discontinuous  across $\Sigma$ if and only if $[\partial_\mu J]_\Sigma\neq0$.  

Since  $\Sigma$ is the front of an electromagnetic shock wave, the electromagnetic field is continuous across $\Sigma$ but some of its  derivatives are discontinuous across $\Sigma$. We express this by writing:\index{shock wave}
\begin{eqnarray}\label{sw1}
[F_{\mu\nu}]_\Sigma = 0,
\end{eqnarray}and:
\begin{eqnarray}\label{sw2}
[\nabla_\lambda F_{\mu\nu}]_\Sigma \neq 0 \textrm{ for  some $\lambda, \mu,\nu$}.
\end{eqnarray} Similar conditions hold for the dual tensor $F^*_{\mu\nu}$. 

Note that since the  field $F_{\mu\nu}$ and the Christoffel symbols $\Gamma^\lambda_{\mu\nu}$ are continuous, we have:
\begin{eqnarray}
[\nabla_\lambda F_{\mu\nu}]_\Sigma 
&=&[\partial_\lambda F_{\mu\nu}]_\Sigma.
\end{eqnarray}

Now consider a second  coordinate system $\{x^{\breve\mu}\}$ such that $x^{\breve0} = z(x^{\mu})$ (cf. Papapetrou  \cite{Papapetrou} p. 171). Then $\Sigma$ can be described by the equation $x^{\breve0}=0$. 
Moreover,  since $z(x^\mu)$ has continuous first-order partial derivatives, we get that:
\begin{eqnarray}\label{propagationvector}
[\partial_\lambda F_{\mu\nu}]_\Sigma&=&
\left[(\partial{}_{\breve0}F_{\mu\nu})\frac{\partial x^{\breve0}}{\partial x^\lambda}\right]_\Sigma\nonumber\\
&=&\left[(\partial{}_{\breve0}F_{\mu\nu})\partial_\lambda z\right]_\Sigma\nonumber\\
&=&[\partial_{\breve0}F_{\mu\nu}]_\Sigma \cdot \partial_\lambda z\nonumber\\
&=&
f_{\mu\nu} k_\lambda,
\end{eqnarray}where $f_{\mu\nu}:=[\partial{}_{\breve0}F_{\mu\nu}]_\Sigma
$ is the so-called \emph{discontinuity} or \emph{disturbance} in the field,\index{discontinuity! in a field} and the 1-form $k_\lambda :=\partial_\lambda z$ is called the \emph{propagation vector}.\index{propagation vector} It is required that $k_\lambda$ be nonzero.

\begin{thm}The quantity $f_{\mu\nu}$ is a tensor. Moreover, it is a 2-form.  
 \end{thm}
 \begin{proof} We verify that if we go to another coordinate system $\{x^{\mu'}\}$, then the quantity $f_{\mu\nu}$ transforms as a tensor should.
\begin{eqnarray}
f_{\mu'\nu'} &=& [\partial{}_{\breve0}F_{\mu'\nu'}]_\Sigma\nonumber\\
&=&\left[\partial{}_{\breve0}\left(F_{\mu\nu}\frac{\partial x^{\mu}}{\partial x^{\mu'}}\frac{\partial x^{\nu}}{\partial x^{\nu'}}\right)\right]_\Sigma\nonumber\\
&=&\left[\left(\partial{}_{\breve0}F_{\mu\nu}\right)\frac{\partial x^{\mu}}{\partial x^{\mu'}}\frac{\partial x^{\nu}}{\partial x^{\nu'}} + F_{\mu\nu}\partial{}_{\breve0}\left(\frac{\partial x^{\mu}}{\partial x^{\mu'}}\frac{\partial x^{\nu}}{\partial x^{\nu'}}\right)\right]_\Sigma\nonumber\\
&=&\left[\left(\partial{}_{\breve0}F_{\mu\nu}\right)\frac{\partial x^{\mu}}{\partial x^{\mu'}}\frac{\partial x^{\nu}}{\partial x^{\nu'}}\right]_\Sigma + \left[\underbrace{F_{\mu\nu}\partial{}_{\breve0}\left(\frac{\partial x^{\mu}}{\partial x^{\mu'}}\frac{\partial x^{\nu}}{\partial x^{\nu'}}\right)}_{\textrm{continuous}}\right]_\Sigma\nonumber\\
&=&\left[\left(\partial{}_{\breve0}F_{\mu\nu}\right)\right]_\Sigma\left(\frac{\partial x^{\mu}}{\partial x^{\mu'}}\frac{\partial x^{\nu}}{\partial x^{\nu'}}\right)\nonumber\\
&=&f_{\mu\nu}\frac{\partial x^{\mu}}{\partial x^{\mu'}}\frac{\partial x^{\nu}}{\partial x^{\nu'}}.
\end{eqnarray} Moreover, $f_{\mu\nu}$ is a 2-form since $f_{\mu\nu} = - f_{\nu\mu}$. (Note that the above underlined ``continuous" term is  continuous since the  background spacetime, which is Minkowskian, is $C^2$.)
\end{proof}
Note that for the dual $F^*_{\mu\nu}$  we write, in analogy with Equation (\ref{propagationvector}):
\begin{eqnarray}\label{propagationvector2}
[\partial_\lambda F^*_{\mu\nu}]_\Sigma = f^*_{\mu\nu}k_\lambda,
\end{eqnarray} where $f^*_{\mu\nu} :=[\partial_{0'}F^*_{\mu\nu}]_{\Sigma}$ is the discontinuity of the dual field.\index{discontinuity! in a field! dual of}\index{field discontinuity! dual of}  Analogously to the relation $F^*_{\alpha\beta} = \frac{1}{2}\varepsilon_{\alpha\beta\mu\nu}F^{\mu\nu}$, one has:
\begin{eqnarray}
f^*_{\alpha\beta} = \frac{1}{2}\varepsilon_{\alpha\beta\mu\nu}f^{\mu\nu}.
\end{eqnarray} Moreover:
\begin{eqnarray}\label{fup}
[\partial_\lambda F^{\mu\nu}]_\Sigma = f^{\mu\nu}k_\lambda,
\end{eqnarray}and: 
\begin{eqnarray}
[\partial_\lambda {F^*}^{\mu\nu}]_\Sigma = {f^*}^{\mu\nu}k_\lambda.
\end{eqnarray}

\section{Dispersion laws and polarization}\label{polarization}

In nonlinear field theory, field discontinuities (or \emph{photons}, in a classical corpuscular sense)  can exhibit birefringent behavior  \cite{ABH, Novello2, DeLorenci}. This means that the way a photon propagates through the field depends on its polarization state.\index{birefringence} Whether a theory predicts birefringence or not depends on the Lagrangrian used. For example, in Born-Infeld electrodynamics, there is no birefringence (see e.g. Novello  \cite{ABH} p. 276). In the Euler-Heisenberg theory however,  there is. 

The goal of the present section is to  derive the dispersion laws for $L(F,G)$-theories. We begin with the following observation:

\begin{thm}\label{lemmapolarize3}Locally, there exists a 1-form $p_\mu$ such that $f_{\mu\nu} = p_\mu  k_\nu - p_\nu k_\mu.$
\end{thm}
\begin{proof}
Applying the Hadamard bracket  to both sides of Equation (\ref{Bianchi2}) gives: 
\begin{eqnarray}
{f^*}^{\mu\nu} k_\mu = 0,
\end{eqnarray} which implies that  $\det\left({f^*}^{\mu\nu}\right)= -|g|^{-1/2} \det\left(f_{\mu\nu}\right)=0$. Since we are in four dimensions, it follows that $f_{\mu\nu}$ is simple (e.g., Penrose and Rindler  \cite{PenroseRindler} p. 166).\index{discontinuity! in a field! simplicity of} \index{field discontinuity! simplicity of} That is, locally there exist 1-forms $u_\mu$ and $v_\mu$ such that:
\begin{eqnarray}\label{polarize1}
f_{\mu\nu} = u_{[\mu}v_{\nu]} =\frac{1}{2}( u_\mu v_\nu -u_\nu v_\mu). 
\end{eqnarray}Taking the Hadamard bracket of the Bianchi identity (\ref{Bianchi})  gives:
\begin{eqnarray}\label{polarize2}
f_{[\mu\nu} k_{\lambda]}=0.
\end{eqnarray} Equations (\ref{polarize1}) and (\ref{polarize2}) imply that  the triple wedge product of $u_\mu$, $v_\nu$ and $k_\lambda$ vanishes.  Hence   $u_\mu$, $v_\nu$ and $k_\lambda$   must be linearly dependent (e.g., Madsen and Tornehave   \cite{MadsenTornehave} pp. 11 - 12) and so locally  there exists a 1-form $p_\mu$ for which we have the decomposition: \index{discontinuity! in a field! decomposition of}\index{field discontinuity! decomposition of}
\begin{eqnarray}\label{polarize3}
f_{\mu\nu} = p_\mu k_\nu - p_\nu k_\mu.
\end{eqnarray}
\end{proof}
Theorem \ref{lemmapolarize3}/Equation (\ref{polarize3}) says that the field discontinuity $f_{\mu\nu}$  (or \emph{photon}, as we are apt to call it) is   the wedge product of the propagation vector $k_\mu$ together with $p_\mu$.  Without loss of generality  we can assume that $p_\mu$ is orthogonal to $k_\mu$ and  thereby interpret $p_\mu$ as being  the (non-normalized)  polarization vector (actually a 1-form) for the photon.\index{polarization}

Let us take  the Hadamard bracket of both sides of  Equation (\ref{EOM4}). After a bit of rearranging, one gets that:
\begin{eqnarray}\label{[EOM4]}
L_F g^{\lambda\nu}f_{\mu\nu}k_\lambda =-2Q_\mu^{\phantom{\mu}\nu\alpha\beta}f_{\alpha\beta}k_\nu.
\end{eqnarray}  Substituting  (\ref{polarize3}) into (\ref{[EOM4]}), and using  the assumption that $p_\mu$ is orthogonal to $k_\mu$, it follows that:
\begin{eqnarray}\label{pstates1-}
L_F k^2 p_\mu =-4Q_\mu^{\phantom{\mu}\alpha\nu\beta}k_\alpha k_\beta p_\nu,
\end{eqnarray}where $k^2 := g^{\mu\nu}k_\mu k_\nu$.  Assuming $L_F\neq 0$, we can write: 
\begin{eqnarray}\label{pstates1}
k^2 p_\mu =- \frac{4}{L_F}Q_\mu^{\phantom{\mu}\alpha\nu\beta}k_\alpha k_\beta p_\nu.
\end{eqnarray}

For convenience, define the tensor  (cf.  De Lorenci \emph{et al.}   \cite{DeLorenci}):
\begin{eqnarray}
S^{\mu\nu}:=k^2g^{\mu\nu} + \frac{4}{L_F} Q^{\mu\alpha\nu\beta}k_\alpha k_\beta.
\end{eqnarray} Then Equation (\ref{pstates1}) can be expressed as:
\begin{eqnarray}\label{pstates2}
S^\mu_{\phantom{\mu}\nu}p_\mu = 0.
\end{eqnarray} 

If $k^2\neq0$, Equation (\ref{pstates1}) implies that $p_\mu$ can be expressed  as a linear combination of $h_\mu:= F_\mu^{\phantom{\mu}\lambda}k_\lambda$ and $h^*_\mu:={F^*_\mu}^{\lambda}k_\lambda$. (Note that both $h_\mu$ and $h^*_\mu$ are orthogonal to $k_\mu$: since $F^{\mu\nu}$ and ${F^*}^{\mu\nu}$ are skew-symmetric we get that $h^\mu k_\mu = F^{\mu\lambda} k_\lambda k_\mu = 0$ and ${h^*}^\mu k_\mu ={F^*}^{\mu\nu}k_\lambda k_\mu =0$.)

Writing:
\begin{eqnarray}\label{polarizationdecomp}
p_\mu = ah_\mu  + bh^*_\mu.
\end{eqnarray}

We get that:
\begin{eqnarray}\label{actS1}
S^\mu_{\phantom{\mu}\nu} h_\mu= \frac{4}{L_F }\left(\left(\frac{L_F k^2}{4} + L_{FF}h^2 + L_{FG}h^\alpha h^*_\alpha\right)h_\nu + \left(L_{GG}h^\alpha h^*_\alpha + L_{FG}h^2\right)h^*_\nu\right),
\end{eqnarray}
and
\begin{eqnarray}\label{actS2}
S^\mu_{\phantom{\mu}\nu} h^*_\mu= \frac{4}{L_F }\left(\left(L_{FF}h^\alpha h^*_\alpha + L_{FG}{h^*}^\alpha h^*_\alpha\right)h_\nu + \left(\frac{L_F k^2}{4}  + L_{GG}{h^*}^\alpha h^*_\alpha+ L_{FG}h^\alpha h^*_\alpha\right)h^*_\nu\right).
\end{eqnarray}

Equations (\ref{actS1}) and (\ref{actS2}) can be recast into a somewhat more useful form by exploiting the well-known identities (cf. Novello  \cite{ABH} p. 272):
 \begin{eqnarray}\label{Gid}
 F^\nu_{\phantom{\nu}\lambda} {F^*}^{\mu\lambda} =\frac{1}{4}G g^{\mu\nu},
\end{eqnarray} and:
 \begin{eqnarray}\label{F**id}
{F^*}^\mu_{\phantom{\mu}\alpha}{F^*}^{\alpha\nu} -  F^\mu_{\phantom{\mu}\alpha} F^{\alpha\nu} = \frac{1}{2}F g^{\mu\nu}.
\end{eqnarray}For  (\ref{Gid}) and (\ref{F**id}) respectively, contracting both sides  with $k_\mu k_\nu$ gives:
\begin{eqnarray}\label{hG}
h^\alpha h^*_\alpha = \frac{1}{4}Gk^2,
\end{eqnarray} and:
\begin{eqnarray}\label{hF}
-{h^*}^\alpha h^*_\alpha + h^2 =\frac{1}{2}Fk^2.
\end{eqnarray} 

Using (\ref{hG}) and (\ref{hF}), Equations (\ref{actS1}) and (\ref{actS2}) become:
\begin{eqnarray}\label{actS1*}
S^\mu_{\phantom{\mu}\nu} h_\mu= \frac{4}{L_F }\left(\left(\left(\frac{L_F }{4} +\frac{1}{4}G L_{FG} \right)k^2+ L_{FF}h^2\right)h_\nu + \left( \frac{1}{4}GL_{GG}k^2 + L_{FG}h^2\right)h^*_\nu\right),
\end{eqnarray}
and:
\begin{eqnarray}\label{actS2*}
S^\mu_{\phantom{\mu}\nu} h^*_\mu&=& \frac{4}{L_F }\left(\left(\left( \frac{1}{4}GL_{FF}  - \frac{1}{2}FL_{FG}\right)k^2+ L_{FG}h^2 \right)h_\nu\right.\nonumber\\
&&\left. + \left(\left(\frac{L_F}{4}    - \frac{1}{2}FL_{GG}+\frac{1}{4}GL_{FG}\right)k^2 + L_{GG}h^2\right)h^*_\nu\right).
\end{eqnarray}

Equations (\ref{pstates2}), (\ref{polarizationdecomp}), (\ref{actS1*}), and (\ref{actS2*}) give:

\begin{eqnarray} 
0&=&\left(a\left(\left(\frac{L_F }{4} +\frac{1}{4}G L_{FG} \right)k^2+ L_{FF}h^2\right)+b\left(\left( \frac{1}{4}GL_{FF}  - \frac{1}{2}FL_{FG}\right)k^2+ L_{FG}h^2 \right)\right)h_\nu\nonumber\\
\nonumber\\
&&+ \left(a\left( \frac{1}{4}GL_{GG}k^2 + L_{FG}h^2\right)+b \left(\left(\frac{L_F}{4}    - \frac{1}{2}FL_{GG}+\frac{1}{4}GL_{FG}\right)k^2 + L_{GG}h^2\right)\right)h^*_\nu.\nonumber\\
\end{eqnarray}

First we consider the case where $h_\nu$ and $h^*_\nu$ are linearly independent. In this case, we have the following linear system in the variables $a$ and $b$:
\begin{eqnarray}\label{system=}
\left\{ \begin{array}{l}
a\left(\left(\frac{L_F}{4}+\frac{1}{4}GL_{FG}\right)k^2 + L_{FF}h^2\right)+ b\left(\left(\frac{1}{4}GL_{FF}-\frac{1}{2}FL_{FG}\right)k^2 +L_{FG}h^2 \right)
=0\\
\\
a\left(  \frac{1}{4}GL_{GG}k^2+L_{FG}h^2\right)+ b\left(\left(\frac{L_F}{4} -\frac{1}{2}FL_{GG}+\frac{1}{4}GL_{FG}\right)k^2 +L_{GG}h^2\right)=0.\end{array}\right.
\end{eqnarray}  The determinant of this system has to be zero (there is a nontrivial solution for $a$ and $b$ because the polarization vector  $p_\mu = ah_\mu + b h^*_\mu$ is nonzero). Thus:
\begin{eqnarray}\label{det0`}
\left(\left(\frac{L_F}{4}+\frac{1}{4}GL_{FG}\right)k^2 + L_{FF}h^2\right)
\left(\left(\frac{L_F}{4} -\frac{1}{2}FL_{GG}+\frac{1}{4}GL_{FG}\right)k^2 +L_{GG}h^2\right)
\nonumber\\
=\left(\left(\frac{1}{4}GL_{FF}-\frac{1}{2}FL_{FG}\right)k^2 +L_{FG}h^2 \right)
\left(  \frac{1}{4}GL_{GG}k^2+L_{FG}h^2\right).
\end{eqnarray}

In the case where $h_\mu$ and $h^*_\mu$ are linearly dependent,  it follows that $S^\mu_{\phantom{\mu}\nu}h_\mu=S^\mu_{\phantom{\mu}\nu}{h^*}_\mu=0$ since $p_\mu = ah_\mu + b{h^*}_\mu$ and $S^\mu_{\phantom{\mu}\nu}p_\mu=0$. We thereby obtain the system:

\begin{eqnarray}\label{lindepend}
\left\{ \begin{array}{l}
\left(\left(\frac{L_F}{4}+\frac{1}{4}GL_{FG}\right)k^2 + L_{FF}h^2\right)h_\nu  + \left(  \frac{1}{4}GL_{GG}k^2+L_{FG}h^2\right)h^*_\nu
=0\\
\\
\left(\left(\frac{1}{4}GL_{FF}-\frac{1}{2}FL_{FG}\right)k^2 +L_{FG}h^2 \right)h_\nu+ \left(\left(\frac{L_F}{4} -\frac{1}{2}FL_{GG}+\frac{1}{4}GL_{FG}\right)k^2 +L_{GG}h^2\right)h^*_\nu=0.\end{array}\right.\nonumber\\
\end{eqnarray}  Since we require at least one component of $h_\mu$ or $h^*_\mu$  to be nonzero ($p_\mu = ah_\mu + b h^*_\mu$ is nonzero),  Equation (\ref{det0`})  holds even if  $h_\mu$ and $h^*_\mu$ are linearly dependent.

Expanding the products and combining like terms, Equation (\ref{det0`}) can be put in the form:
\begin{eqnarray}\label{Fresnel}
\Lambda_1 k^4 + \Lambda_2 h^2 k^2 + \Lambda_3 h^4 = 0,
\end{eqnarray}
where we define:
\begin{eqnarray}\label{lambdadefstart}
\Lambda_1 &:=& (L_F +GL_{FG})^2 -L_{GG}(2F L_F + G^2L_{FF}),\\
\Lambda_2&:=&4\left(L_F(L_{FF}+L_{GG})+2F(L_{FG}^2-L_{FF}L_{GG})\right) ,\\
\Lambda_3&:=& 16(L_{FF}L_{GG} -L_{FG}^2).\label{lambdadefstop}
\end{eqnarray}  

We now have the following result:
\begin{thm}\label{thm2.4}  Assuming that $L_F\neq0$, $k^2\neq0$, 
 $\Lambda_1\neq0$ and $\Lambda_2^2 - 4\Lambda_1\Lambda_3\geq0$ (in order  to   ensure that  Equation (\ref{Fresnel}) gives real solutions  for $k^2$), we have the dispersion law(s) \cite{BB}: \index{discontinuity! in a field! dispersion/propagation of} \index{field discontinuity! dispersion/propagation of} \index{dispersion laws}
\begin{eqnarray}\label{dispersion}
k^2 = \Lambda_\pm h^2,
\end{eqnarray}where:\index{birefringence} 
\begin{eqnarray}\label{lambdapmDef}
\Lambda_\pm:=\frac{-\Lambda_2 \pm \sqrt{(\Lambda_2)^2 -4\Lambda_1\Lambda_3}}{2\Lambda_1}.
\end{eqnarray} 
\end{thm}\
\\
Next, we will show that Equation (\ref{dispersion}) continues to hold even if $k^2 =0$. More precisely:

\begin{thm} If  $k^2 = 0$ and  $L_{FF}L_{GG} - L_{FG}^2 \neq 0$, then  $h^2=0$.
\
\\
\emph{(Note that with the Euler-Heisenberg Lagrangian (\ref{EH}), we have $L_{FF}L_{GG}-L_{FG}^2\neq0$.)}\end{thm}
\begin{proof}
For an indirect proof, suppose  that $k^2 =0$ and  $h^2 \neq 0$. Note that when $k^2=0$, Equations (\ref{hG}) and (\ref{hF}) give  $h^\alpha h^*_\alpha=0$ and ${h^*}^\alpha h^*_\alpha = h^2$. Note that $h_\mu$ cannot be timelike because otherwise $h^*_\mu$ would also be timelike and one cannot have two orthogonal timelike vectors in Minkowski spacetime. The only remaining possibility is that $h_\mu$ (and consequently $h^*_\mu$) is spacelike. We show that this implies $L_{FF} L_{GG} - L_{FG}^2 =0$.

To this end, note that with $k^2 = 0$, Equation (\ref{pstates1-}) becomes:
\begin{eqnarray}\label{----}
0&=&Q^{\mu\alpha\nu\beta}k_\alpha k_\beta p_\nu\nonumber\\
&=&L_{FF}h^\mu\left(h^\nu p_\nu \right) + L_{FG}\left(h^\mu\left({h^*}^\nu p_\nu \right) +{h^*}^\mu\left(h^\nu p_\nu \right)  \right) +L_{GG}{h^*}^\mu \left({h^*}^\nu p_\nu \right).
\end{eqnarray} Contracting (\ref{----}) with $h_\mu$ and ${h^*}_\mu$ respectively, and using the relations ${h^*}^\alpha h^*_\alpha=h^2\neq0$ and $h^\alpha h^*_\alpha=0$, it follows that:
\begin{eqnarray}\label{----++++}
\left\{ \begin{array}{l}
L_{FF}\left(h^\nu p_\nu \right) + L_{FG} \left({h^*}^\nu p_\nu \right)
=0\\
L_{FG}\left(h^\nu p_\nu \right) + L_{GG} \left({h^*}^\nu p_\nu \right)=0.\end{array}\right.
\end{eqnarray} We claim that $h^\nu p_\nu$ and ${h^*}^\nu p_\nu$ cannot simultaneously be zero. To establish this claim, note that since $k_\mu$ is null, and since $p_\mu$ is orthogonal (and not parallel)  to $k_\mu$, it must be that $p_\mu$ is spacelike. So if $h^\nu p_\nu$ and ${h^*}^\nu p_\nu$ were both simultaneously zero, we would have an  orthogonal basis consisting of three spacelike 1-forms $h_\mu$, $h^*_\mu$, $p_\mu$ and a null 1-form  $k_\mu$, which is not possible.

Consequently, the system (\ref{----++++}) has a nontrivial solution for $h^\nu p_\nu$ and ${h^*}^\nu p_\nu$, and so it must have   a  vanishing determinant: $L_{FF}L_{GG} - L_{FG}^2 =0$.
\end{proof}

Now, according to Theorem \ref{thm2.4},  there can be two possible values of $k^2$. This has to do with the fact that a given photon (field disturbance) is in one of two polarization states. To understand why this is true,  note that (as explained in  e.g., De Lorenci \emph{et al.}  \cite{DeLorenci}.)\index{polarization! and dispersion (birefringence)} for each possible  value of $k^2$  there corresponds a certain solution space for the unknowns $a$ and $b$ in the system (\ref{system=}), and  $p_\mu = a h_\mu + b h^*_\mu$. (Here we identify the solution space of $a$ and $b$ with the ``polarization state.")

By defining:
\begin{eqnarray}
\Omega_\pm : = -\frac{4L_{FF} + (L_F + G L_{FG})\Lambda_\pm}{4L_{FG} + GL_{GG}\Lambda_\pm},
\end{eqnarray}  Equation (\ref{dispersion}) becomes:
\begin{eqnarray}
k^2 = -4\left(\frac{L_{FF} + L_{FG} \Omega_\pm}{L_F + (L_{FG} + L_{GG}\Omega_\pm)G}\right) h^2,
\end{eqnarray} which matches Equation (24) in De Lorenci \emph{et al.}  \cite{DeLorenci}.

We note that with the Euler-Heisenberg Lagrangian (\ref{EH}), Equation (\ref{lambdapmDef}) reads:
\begin{eqnarray}\label{EHlambda}
\Lambda_\pm = \frac{224\alpha^2}{495 + 12 F \alpha^2 \mp \sqrt{18225 - 18360 F \alpha^2 +
4624 F^2\alpha^4 + 3136 G^2 \alpha^4}}.
\end{eqnarray}

\section{Effective null geodesics}\label{nullgeo}

The dispersion laws in nonlinear electrodynamics have an appealing geometric interpretation, where they are thought of as  light cone conditions \index{effective light cone} in a so-called \emph{effective geometry}. \index{effective geometry}  Using the fact that $h^2 = -F^\mu_{\phantom{\mu}\alpha} F^{\alpha\nu}k_\mu k_\nu$, we can write Equation (\ref{dispersion}) in the form:
\begin{eqnarray}\label{egdisp1}
\left(g^{\mu\nu} + \Lambda_\pm F^\mu_{\phantom{\mu}\alpha}F^{\alpha\nu}\right)k_\mu k_\nu = 0. 
\end{eqnarray} Provided that the symmetric tensor $\widetilde{g}^{\mu\nu}$ defined by:\index{cometric! effective}
\begin{eqnarray}\label{cometric effective}
\widetilde{g}^{\mu\nu}:=g^{\mu\nu} + \Lambda_\pm F^\mu_{\phantom{\mu}\alpha}F^{\alpha\nu},
\end{eqnarray} is nonsingular, an \emph{effective metric} $\widetilde{g}_{\mu\nu}$ can be defined such that $\widetilde{g}^{\mu\lambda}\widetilde{g}_{\lambda\nu}=\delta^\mu_{\phantom{\mu}\nu}$. We get the \emph{effective geometry}  \index{effective geometry}  by treating  $\widetilde{g}_{\mu\nu}$ as if it were  the metric for spacetime.

We can  think of  $k_\mu$ as being  null with respect to the effective metric \index{metric! effective} since:
\begin{eqnarray}\label{egdisp2}
\widetilde{g}^{\mu\nu}k_\mu k_\nu = 0.
\end{eqnarray}

Moreover, the integral curves of $k_\mu$ (i.e., \emph{photon worldlines}) turn out to be geodesics with respect to the effective metric:\index{effective geometry! null geodesics in} 

\begin{thm}The integral curves of $k_\mu$  are null geodesics with respect to the effective metric.
\end{thm}
\begin{proof}
The proof is given in Novello   \cite{ABH} pp. 273 - 274.  We repeat it in order to be self-contained. The first step is to take the partial derivative  of Equation  (\ref{egdisp2}) to get:
\begin{eqnarray}\label{enpd}
2(\partial_\lambda k_\mu) k_\nu \widetilde{g}^{\mu\nu}  + k_\mu k_\nu \partial_\lambda \widetilde{g}^{\mu\nu} = 0.
\end{eqnarray}
Next, exploit the fact that the effective metric ${\widetilde g}^{\mu\nu}$, provided it is nonsingular, determines a set of torsion-free connection coefficients ${\widetilde\Gamma}^\lambda_{\phantom{\lambda}\mu\nu}$ (through the usual Christoffel formulae, i.e., Equation (\ref{Christoffel}))\index{connection! derived from effective metric} \index{Christoffel symbols (connection coefficients)! derived from effective metric} and thereby determines a covariant differential operator $\widetilde{\nabla}_\lambda$ such that:\index{covariant derivative! derived from effective metric}
\begin{eqnarray}\label{effcovconst}
\widetilde\nabla_\lambda {\widetilde g}^{\mu\nu} = \partial_\lambda \widetilde{g}^{\mu\nu} + \widetilde\Gamma^\mu_{\phantom{\mu}\alpha\lambda}\widetilde g^{\alpha\nu} + \widetilde\Gamma^\nu_{\phantom{\nu}\alpha\lambda}\widetilde g^{\alpha\mu}=0.
\end{eqnarray}
Contracting  Equation (\ref{effcovconst}) with $k_\mu k_\nu$ one  gets:
\begin{eqnarray}\label{ecccontract}
k_\mu k_\nu \partial_\lambda \widetilde{g}^{\mu\nu} = - 2k_\mu k_\nu \widetilde\Gamma^\mu_{\phantom{\mu}\alpha \lambda} \widetilde{g}^{\alpha\nu}.
\end{eqnarray} By substituting (\ref{ecccontract}) into  Equation (\ref{enpd}), it follows that:
\begin{eqnarray}\label{geo-}
{\widetilde{g}^{\mu\nu}}\left(\widetilde\nabla_\lambda k_\mu\right) k_\nu&:=&\widetilde{g}^{\mu\nu}\left(\partial_\lambda k_\mu -\widetilde\Gamma^\alpha_{\phantom{\alpha}\mu\lambda}k_\alpha\right)k_\nu \nonumber\\
&=&0.
\end{eqnarray} Since $k_\mu :=\partial_\mu z$, and since partial derivatives commute, one gets that:
\begin{eqnarray}\label{exactgradient}
\widetilde\nabla_\lambda k_\mu = \widetilde\nabla_\mu k_\lambda.
\end{eqnarray} Defining $k^\mu := \widetilde{g}^{\mu\nu}k_\nu$, and using (\ref{exactgradient}), Equation (\ref{geo-}) can be rewritten as:
\begin{eqnarray}\label{geo}
(\widetilde{\nabla}_\lambda k_\mu)k^\lambda =0,
\end{eqnarray}  which implies that we have a geodesic (e.g., Poisson  \cite{Poisson} p. 61).
\end{proof}

Note that a given effective metric is only  defined up to conformal equivalence because all that we have are  \emph{null} geodesics. More precisely, an effective geometry is an equivalence class of conformally equivalent Lorentzian metrics. However, we will only work with one representative at a time, choosing whichever conformal factor suits our fancy.

In the case of  Euler-Heisenberg theory, the  effective cometric (\ref{cometric effective}) becomes:
\begin{eqnarray}\label{cometric effective EH}
\widetilde{g}^{\mu\nu}=g^{\mu\nu} + \frac{224\alpha^2 F^\mu_{\phantom{\mu}\lambda}F^{\lambda\nu}}{495 + 12 F \alpha^2 \mp \sqrt{18225 - 18360 F \alpha^2 +
4624 F^2\alpha^4 + 3136 G^2 \alpha^4}}.
\end{eqnarray}

\section{The relationship between the effective metric and the stress-energy tensor}

For an electromagnetic field governed by a Lagrangian  of the type $L=L(F,G)$, we found that the stress-energy tensor is (\ref{SETensor1}): \index{stress-energy tensor} 
\begin{eqnarray}\label{SETensor2}
T_{\mu\nu} = -4L_F F_\mu^{\phantom{\mu}\alpha}F_{\alpha\nu} - (L-GL_G)g_{\mu\nu}.
\end{eqnarray} 
Raising the indices of Equation (\ref{SETensor2}), and rearranging,  we get that  the \emph{substress} tensor $F^\mu_{\phantom{\mu}\alpha}F^{\alpha\nu}$ is (assuming $L_F\neq 0$): \index{substress tensor}
\begin{eqnarray}\label{SETensor4}
F^\mu_{\phantom{\mu}\alpha}F^{\alpha\nu} =- \frac{1}{4L_F}\left(T^{\mu\nu} + (L-GL_G)g^{\mu\nu}\right).
\end{eqnarray}
Using Equation (\ref{SETensor4}), we can rewrite Equation (\ref{cometric effective}) as:  \index{effective geometry! and stress-energy tensor}  \index{stress-energy tensor! and effective geometry} 
\begin{eqnarray}\label{effcometricSET}
\widetilde{g}^{\mu\nu}=\left(1 + \frac{\Lambda_\pm(GL_G - L)}{4L_F}\right)g^{\mu\nu} -\frac{ \Lambda_\pm}{4L_F}T^{\mu\nu}
\end{eqnarray} Equation (\ref{effcometricSET}) is analogous to Einstein's field equation from  general relativity in that it relates the effective geometry with the stress-energy tensor.

We note that since the effective metric has this direct dependence on the stress-energy tensor, de Oliveira and Perez Bergliaffa  \cite{deOliveira} have suggested that the Segr\`e classification of the stress-energy tensor yields a simple  classification scheme for  effective geometries in nonlinear electrodynamics.  \index{effective geometry! classification scheme}

\chapter{Plane waves
}\label{planewaves}

The main purpose of this chapter is to investigate  the effective geometry of  a circularly polarized   plane wave.   Our findings independently confirm  those of a  2002  paper by Denisov and Denisova \cite{DenisovDenisova}. The effective geometry of a plane wave is conformally equivalent to Minkowski spacetime  but it is distinguishable from the Minkowski background because the effective null geodesics  are not necessarily null in the background.

Moreover, we show that, as viewed from the background coordinates,  shock  disturbances  in a circularly polarized monochromatic plane wave field propagate with a directionally dependent index of refraction which can be easily computed. 
 Our result for the index of refraction confirms, to lowest order,    the   calculations performed by   
Affleck \cite{Affleck}   in 1988. Lorentz invariance is manifestly preserved throughout in our approach. This is in contrast  to Affleck's non-invariant approximation. 

We close the chapter  with a discussion anticipating the possibility of  optical black holes in vacuum. 

\section{Effective geometry of null fields}\label{egonf}
Recall that a null field \index{null fields} is one such that $F^2 + G^2 \equiv 0$. In the case of a null field, the nonlinear field equations (\ref{EOM2}) reduce  exactly to Maxwell's equations in the absence of charges and currents, provided that $L_G  = 0$ when $F^2 + G^2 =0$. Note that the Euler-Heisenberg Lagrangian (\ref{EH}) indeed satisfies this latter condition. Consequently,  a null field is an exact solution to the  Euler-Heisenberg equations if and only if it is an exact solution to Maxwell's equations.  

Let us consider the effective geometry corresponding to a general null field. 
 
 In Section \ref{polarization}, we found that  field disturbances (shock waves)   disperse according to:
\begin{eqnarray}\label{dispCH2}
\left(\left(1 + \frac{\Lambda_\pm(GL_G - L)}{4L_F}\right)g^{\mu\nu} -\frac{ \Lambda_\pm}{4L_F}T^{\mu\nu}\right)k_\mu k_\nu = 0, 
\end{eqnarray} where $\Lambda_\pm$ is given by Equation (\ref{lambdapmDef}). The choice of $\pm$ depends on the polarization of the disturbance.  Accordingly, we will refer to the polarization modes as being either $+$ or $-$ modes.

Specializing to null fields, Equation (\ref{dispCH2})  leads to an effective cometric given by:
\begin{eqnarray}\label{ecometricNF}
\widetilde{g}^{\mu\nu} = g^{\mu\nu} + P_\pm T^{\mu\nu},
\end{eqnarray} where:
\begin{eqnarray}\label{P_pm def}
P_\pm= -\frac{ \Lambda_\pm}{4L_F + \Lambda_\pm(GL_G - L)}\Big|_{F^2 + G^2 =0}.
\end{eqnarray}
For the Euler-Heisenberg Lagrangian, we get that $P_\pm=  (22\pm6)\alpha^2/45$  (cf. De Lorenci \emph{et al.} \cite{DeLorenci}).

Locally, the stress-energy tensor for a null  electromagnetic field satisfying the dominant energy condition (see e.g., Hawking and Ellis \cite{HawkingEllis} p. 91) can be expressed as \cite{HallNegm}: \index{stress-energy tensor! of null field}
\begin{eqnarray}
T^{\mu\nu} =l^\mu l^\nu,
\end{eqnarray}  where $l^\mu$ is a null vector. 
Thus  effective cometrics corresponding to null fields can be expressed by equations of the form:\index{null fields! effective geometry of} \index{effective geometry! of null fields}
\begin{eqnarray}\label{effnull}
\widetilde{g}^{\mu\nu} = g^{\mu\nu} + P_\pm l^\mu l^\nu.
\end{eqnarray}

\section{Effective geometry of  plane waves}

Taking the usual $t, x, y, z$ coordinates for Minkowski   spacetime. The field tensor: 
\begin{eqnarray}\label{Ftensor}
F_{\mu\nu}&=&\left(\begin{array}{cccc}
F_{tt}&F_{tx} &F_{ty}& F_{tz}\\
F_{xt}&F_{xx} &F_{xy} &F_{xz}\\
F_{yt} & F_{yx} &F_{yy}&F_{yz}\\
F_{zt} & F_{zx}&F_{zy}&F_{zz}\\
\end{array}\right)\nonumber\\
&=&\left(\begin{array}{cccc}
0& A\cos\left(\omega(t-z)\right) &B\sin\left(\omega(t-z)\right) & 0\\
-A\cos\left(\omega(t-z)\right) &0 &0 &A\cos\left(\omega(t-z)\right)\\
-B\sin\left(\omega(t-z)\right) & 0  &0&B\sin\left(\omega(t-z)\right)\\
0 & -A\cos\left(\omega(t-z)\right)&-B\sin\left(\omega(t-z)\right)  &0\\
\end{array}\right)\nonumber\\
\end{eqnarray} describes a monochromatic plane wave propagating along the $+z$ direction, having  frequency $\omega$, and elliptical polarization with fixed amplitudes $A$ and $B$. The stress-energy tensor corresponding to this field is:\index{stress-energy tensor! of plane waves! elliptical polarization}\index{plane wave! stress-energy tensor of! elliptical polarization}
 \begin{eqnarray}
T^{\mu\nu}&=&\left(\begin{array}{cccc}
T^{tt}&T^{tx} &T^{ty}& T^{tz}\\
T^{xt} & T^{xx} &T^{xy}&T^{xz}\\
T^{yt} & T^{yx} &T^{yy}&T^{yz}\\
T^{zt} &T^{zx}& T^{zy}& T^{zz}\\
\end{array}\right)\nonumber\\
&=&\left(\begin{array}{cccc}
A^2\cos^2(\omega(t-z)) + B^2\sin^2(\omega(t-z))& 0 &0 & A^2\cos^2(\omega(t-z)) + B^2\sin^2(\omega(t-z))\\
0 &0 &0 &0\\
0& 0  &0&0\\
A^2\cos^2(\omega(t-z)) + B^2\sin^2(\omega(t-z)) & 0&0&A^2\cos^2(\omega(t-z)) + B^2\sin^2(\omega(t-z))\\
\end{array}\right)\nonumber\\
&=&l^\mu l^\nu,
\end{eqnarray} where:
\begin{eqnarray}
l^\mu &=& \left(\begin{array}{c}
l^t\\
l^x \\
l^y\\
l^z\\
\end{array}\right)
\nonumber\\
&=& \left(\begin{array}{c}
\sqrt{A^2\cos^2(\omega(t-z)) + B^2\sin^2(\omega(t-z)) }\\
0 \\
0\\
\sqrt{A^2\cos^2(\omega(t-z)) + B^2\sin^2(\omega(t-z))} \\
\end{array}\right).
\end{eqnarray}

Note that in the case of circular polarization ($A=B$), the stress-energy tensor takes on a particularly simple form; it becomes constant: \index{stress-energy tensor! of plane waves! circular polarization} \index{plane wave! stress-energy tensor of! circular polarization}

\begin{eqnarray}\label{circPWsetensor}
T^{\mu\nu} 
&=&\left(\begin{array}{cccc}
A^2& 0 & 0  & A^2\\
0 &0 &0 & 0\\
0& 0  &0&0\\
A^2& 0&0  &A^2\\
\end{array}\right).
\end{eqnarray} 
 
 Since the stress-energy tensor (\ref{circPWsetensor}) is constant, the corresponding effective geometry (\ref{ecometricNF}) must be flat. In the case of an arbitrary elliptically polarized plane wave, the stress-energy tensor is no longer constant but one can nevertheless calculate that the  Riemann curvature of the effective geometry still vanishes identically. \index{effective geometry! of plane waves}  This confirms results first published by Denisov and Denisova in 2002, who found that the effective geometries corresponding to monochromatic plane waves are flat using the Euler-Heisenberg Lagrangian (\ref{EH}) \cite{DenisovDenisova}. \index{plane wave! effective geometry of} More generally, we note that if we have a Lagrangian of the form $L=L(F,G)$ and if $L_G =0$ when $F^2 + G^2 = 0$, then the resulting field theory will have plane waves as exact solutions and these  will yield   flat effective geometries. \index{effective geometry! of plane waves} \index{plane wave! effective geometry of}  
 
 Although the effective geometry of the plane wave is flat, the effective null geodesics are not necessarily null with respect to the flat background metric (the  precise manner in which the effective light cones embed in the background is studied in  Section \ref{lightconePW}). 
 Thus we see that effective geometries which are flat  can nevertheless be distinguishable from the flat background spacetime. This phenomenon is not limited to plane waves or even to null fields. Just to give a concrete example, a constant uniform electric field (e.g., in a rest frame, $(A_t,A_x,A_y,A_z )= (0,0,0,Et)$ with $E$ constant),  satisfies  the Euler-Heisenberg equations   (\ref{EHFE}) and its effective geometry is a copy of Minkowksi spacetime. Indeed if the stress energy tensor of a field is constant then its effective geometry must be flat.  Considering the  symmetry of such a situation, this  should be expected.

\section{Refraction in a circularly polarized plane wave}\label{section3.3}

The purpose of this section is to calculate the index of refraction for  field disturbances (low-intensity external ``photons")  propagating in a circularly polarized plane wave field. 

Note that the stress-energy tensor of a  null field given by two constant perpendicularly crossed electric and magnetic fields\index{crossed fields} is the  \index{crossed fields! stress-energy tensor of} same as that  of the  circularly polarized plane wave (\ref{circPWsetensor}). \index{stress-energy tensor! of crossed fields} Thus, the effective geometry of  crossed null fields is the same as that of the circularly polarized plane wave.

Using the stress-energy tensor (\ref{circPWsetensor}), together with Equation (\ref{ecometricNF}), we calculate that the effective metrics corresponding to a circularly polarized plane wave (propagating in the $+z$ direction, with amplitude $A$) are:\index{effective geometry! of plane waves! with circular polarization}
\begin{eqnarray}\label{PWeg}
ds^2 &=& (1-P_\pm A^2)dt^2 + 2P_\pm A^2 dt dz
 - dx^2 - dy^2 - (1+P_\pm A^2) dz^2.
\end{eqnarray} Since the Christoffel symbols for the effective metric (\ref{PWeg})  vanish identically, the   null geodesics in the effective geometry are    simply rectilinear curves  in the coordinates $t,x,y,z$.

Consider an effective null geodesic that passes through the origin of the coordinates $t,x,y,z$. The corresponding  projection (i.e., light ray)\index{light ray} for  this geodesic in the three-dimensional $x,y,z$ space  issues from the origin and intersects the unit sphere at spherical coordinates $\theta$ by $\varphi$ (see Figure \ref{SphericalCoordLightRay}). The angle  $\varphi$ measures the angle that the ray   makes with respect to the $+z$ axis, as measured in the $x,y,z$ system.

\begin{figure}[h]
\center
    \includegraphics[height=4.5in]{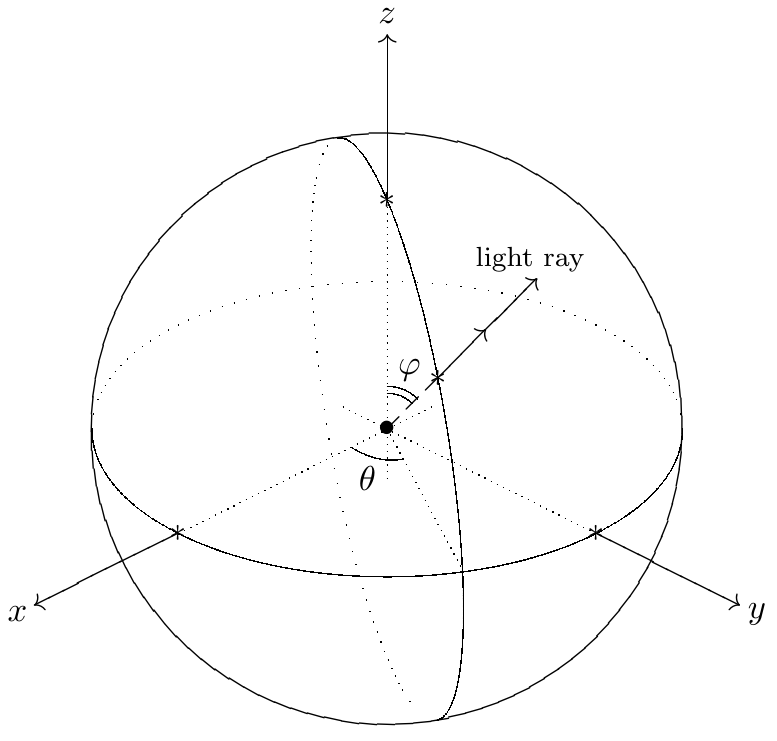}

    \caption[A light ray issuing from the origin]{A light ray\index{light ray} issuing from the origin intersects the unit sphere at $\theta$ by $\varphi$.}

    \label{SphericalCoordLightRay}
\end{figure}

Using the standard conversion formulae between rectilinear and spherical coordinates ($x = r\sin\varphi \cos\theta$, $y=r\sin\varphi \sin\theta$, $z=r \cos\varphi$),  Equation (\ref{PWeg}) implies that, along an effective null geodesic through the origin:
\begin{eqnarray}
\left(1 + P_\pm A^2 \cos^2\varphi\right)\left(\frac{dz}{dt}\right)^2
-2P_\pm A^2 \cos^2\varphi \left(\frac{dz}{dt}\right)
-(1-P_\pm A^2)\cos^2\varphi =0.
\end{eqnarray}Thus:
\begin{eqnarray}
\frac{dz}{dt} = \frac{P_\pm A^2 \cos^2\varphi + \cos\varphi\sqrt{1-P_\pm A^2\sin^2\varphi}}{1 + P_\pm A^2\cos^2\varphi}.
\end{eqnarray} 
Consequently, as measured in  the   $t,x,y,z$ coordinates with respect to the background metric, discontinuities in the plane wave field propagate with a $\varphi$-dependent velocity:
\begin{eqnarray}\label{vphi1}
v(\varphi)&=&\sqrt{\frac{dx^2 + dy^2 + dz^2}{dt^2}}\nonumber\\
&=& \frac{P_\pm A^2 \cos\varphi +\sqrt{1-P_\pm A^2\sin^2\varphi}}{1 + P_\pm A^2\cos^2\varphi}.
\end{eqnarray}
So  the plane wave has an index of refraction, $n(\varphi) = 1/v(\varphi)$: \index{index of refraction}
\begin{eqnarray}\label{PWrefraction}
n(\varphi)&=&\frac{1 + P_\pm A^2\cos^2\varphi}{P_\pm A^2 \cos\varphi +\sqrt{1-P_\pm A^2\sin^2\varphi}}.
\end{eqnarray} Expressing Equation (\ref{PWrefraction}) as a power series in $A$,  one gets:
\begin{eqnarray}\label{PWrefractionTaylor}
n(\varphi)&=& 1 + 2P_\pm A^2 \sin^4\left(\frac{\varphi}{2}\right) +P_\pm^2O(A^4).
\end{eqnarray} Affleck \cite{Affleck} approximated a formula for  $n(\varphi)$ using  methods different from ours. The formula which he  obtained (correcting for typos) is nothing but  the  first two nonzero terms in the expansion (\ref{PWrefractionTaylor}). (Note that Affleck's formula for $n(\varphi)$ was apparently  published with a small  typo: in his Equation (14), the factor $(e E_0/m^2)$ should be  $(e E_0/m^2)^2$.)

\section{Visualizing effective light cones}\label{lightconePW}
The purpose of this section is to describe how the light cone structure of the effective geometry given by Equation (\ref{PWeg}) embeds in the background spacetime.   

From Equation (\ref{vphi1}) we get that, as seen in the background, the plane wave induces a drag effect for field disturbances. Low-intensity photons  that probe the field along the direction of the plane wave (the direction given by the so-called Poynting vector) will continue to travel  at the usual speed of light: $v(0)= 1$. Along  any other direction, the field disturbances  are made to travel at less than the speed of light.   This  drag effect  is most pronounced for $\varphi = \pi$,  the direction exactly opposite to the Poynting vector.

According to Equation (\ref{vphi1}), when $P_\pm A^2\geq1$,  field disturbances  are confined to propagate only in directions such that $\csc^2\varphi \leq P_\pm A^{2}$.

If a shock wave issues from the origin, then we can calculate  the location of the  wavefront in $x,y,z$ space after a unit $t$-time  by plotting Equation (\ref{vphi1}) in the $xz$-plane using $(v,\varphi)$-polar coordinates (i.e., $z=v\cos\varphi$ and $x=v\sin\varphi$). One can then rotate this graph about the $z$-axis ($\varphi=0$) in order to visualize the wavefront as a surface of revolution (in fact, the surfaces in this case are ellipsoids). 
Plots of (\ref{vphi1}), representing the range of qualitative behaviors, are given   in  Figure \ref{PWegFIG}. Note  that disturbances   propagating against  the direction of the field experience a kind of drag effect. The case $A=0$ (solid red in the figure) is a limiting case in which the plane wave has vanishing intensity. The standard propagation law for field discontinuities (propagation at the speed of light) is recovered in this limiting case.
\begin{figure}[h]
\center
    \includegraphics[height=4.5in]{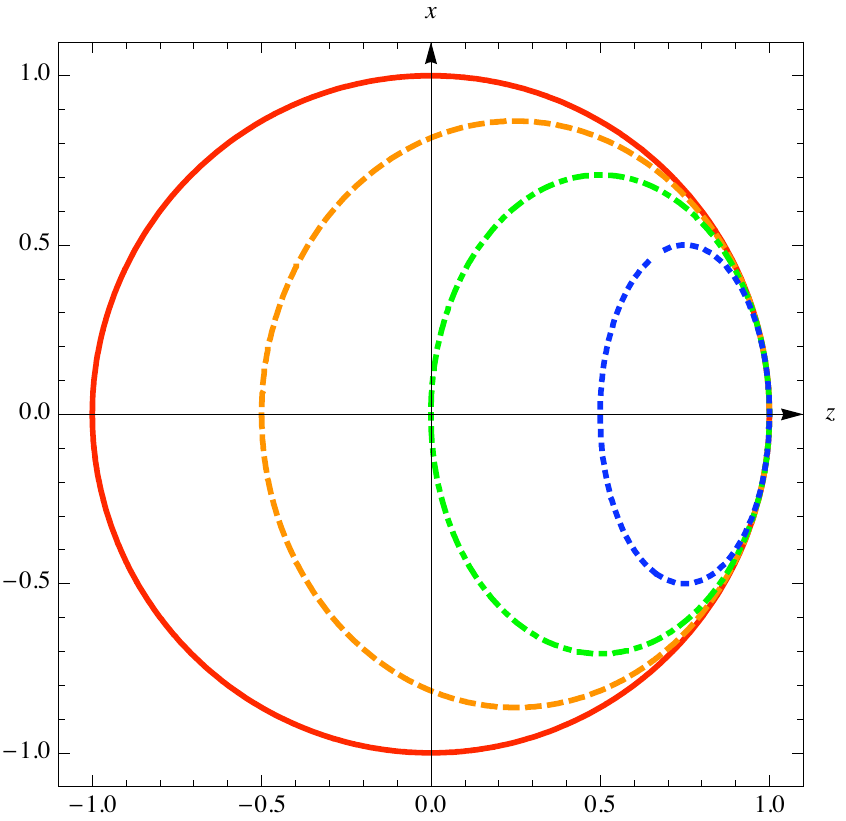}

    \caption[Shock waves in  plane wave fields]{
  Shock wave fronts in circularly polarized plane wave fields.  Four separate cases are shown simultaneously for comparison purposes. In a monochromatic  circularly polarized plane wave of amplitude $A$ propagating in the $+z$ direction,  a shock wave  initiated at the origin is allowed to propagate for a unit $t$-time. The resulting shock fronts are plotted according to Equation (\ref{vphi1}) for four   representative cases: (1) $P_\pm A = 0$ (solid red), (2) $0< P_\pm A<1$ (dashed orange, plotted using $P_\pm A = 1/3$), (3) $P_\pm A = 1$ (dot-dashed green), and (4) $P_\pm A >1$ (dotted blue,  plotted using $P_\pm A = 3$). }

    \label{PWegFIG}
\end{figure}

It is straightforward to verify that 
the polar plots of Equation (\ref{vphi1}) are genuine ellipses with eccentricity:
\begin{eqnarray}
\epsilon = \sqrt{\frac{P_\pm A^2}{1+ P_\pm A^2}}.
\end{eqnarray} Similar observations were made by Boillat, who was however interested in the Born-Infeld  rather than the Euler-Heisenberg Lagrangian. Note  Equation (2.39) in his  1970 paper \cite{Boillat1970}.

\begin{figure}[h]
\center
    \includegraphics[height=4.5in]{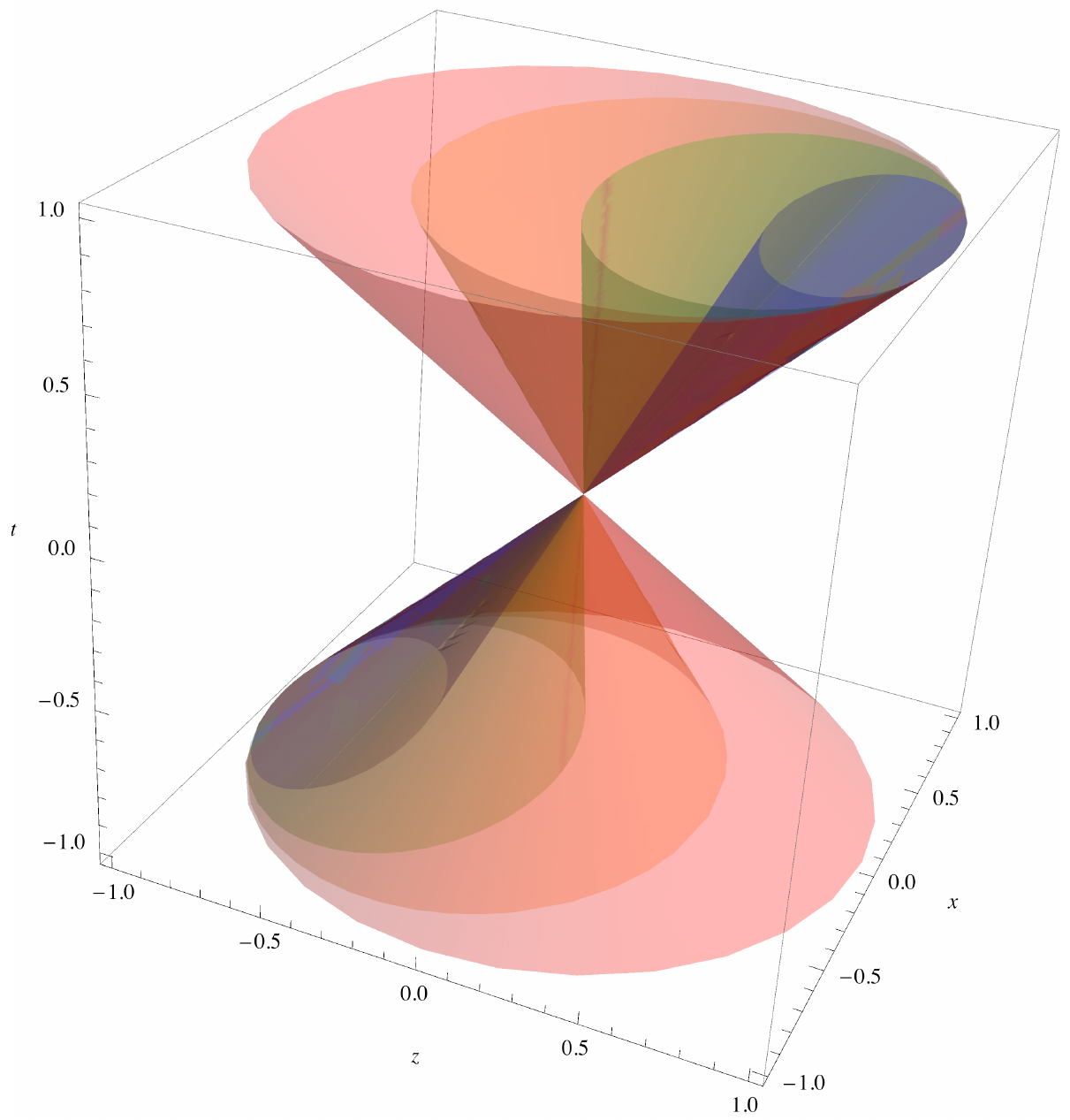}

    \caption[Effective light cones in plane wave fields]{Effective light cones\index{effective light cone! in plane wave field} in circularly polarized plane wave fields. The transparent red, orange, green, and blue cones (outermost to innermost) correspond to the cases plotted in  Figure \ref{PWegFIG}. In fact, Figure \ref{PWegFIG} is just  the cross section  through the plane $t=1$. The transparent red cone (outermost) corresponds to a standard  light cone in  the background Minkowksi spacetime. The transparent orange, green, and blue cones show how the effective  light cones  embed in the background  for plane waves of increasingly intense amplitude.
}

    \label{pwLightCones}
\end{figure}

Figure \ref{pwLightCones} visualizes how the effective light cones of   (\ref{PWeg}) embed  in the background geometry. The $y$-dimension is suppressed.  Again, four representative cases (the same cases used in Figure \ref{PWegFIG}) are presented simultaneously for comparison purposes.  We note that de Oliveira Costa and Perez Bergliaffa \cite{deOliveira} have classified effective light cones according to the Segr\'e type of the stress-energy tensor for  the  field.

Due to birefringence, there are actually  two different effective light cones for a given field configuration. We did not try to depict both of them in the cases shown in   Figures \ref{PWegFIG} and \ref{pwLightCones} in order to  avoid  unnecessary clutter. Note however that the difference between the light cones of the two polarization states  becomes more pronounced at higher intensities.

As we see from these calculations, the effective light cones are tilted in the direction given by the Poynting vector of the field. Based on this observation, we make the following conjecture: 

\begin{conj}\label{conjecture1}
For any given carrier field, the corresponding effective light cones tilt into the direction of the Poynting vector of the field. Moreover, the higher the field intensity, the more pronounced  the tilt. 
\end{conj}

This conjecture suggests that an optical black hole can form if one contrives to create a field, with an inwardly-directed Poynting vector, intense enough to tilt  the  effective light cones all the way to form a trapped surface.

\section{Distortion of clock readings}

Consider an observer at rest at the origin in $t,x,y,z$ coordinates. Surround the observer with clocks, so that in the coordinates these clocks form a sphere $S$ of unit radius. Let these clocks be set in such a way that if light travels along null geodesics in the background geometry, then the clocks appear to the observer as if precisely synchronized. 

Now assume that the observer is immersed in a plane wave of amplitude $A$ such that the effective geometries   given by Equation (\ref{PWeg}) pertain.  We stipulate that the observer sees objects only by way of small  disturbances in the plane wave field;  the  ``photons" seen by  the observer follow null geodesics in the effective geometries (\ref{PWeg}). 
The readings on the stationary clocks at $S$ will no longer appear to be synchronized  since the effective null geodesics propagate anisotropically with respect to the  $t,x,y,z$ coordinates. Moreover, due to birefringence, two clock readings may be seen at once. 
An additional consequence of such  birefringence effects would be that moving bodies could appear to have  double images.

Since the field disturbances that travel in the direction $\varphi = 0$ travel at the usual speed of light, the apparent reading of a clock on $S$ as viewed in the direction $\varphi = \pi$ will not be affected by the effective geometry. By contrast, the  other clock readings will  be affected. One can show that the  difference in  readings $\Delta\tau$ between a clock viewed at angle $\varphi$, and the unaffected clock at $\varphi = \pi$, is given by the formula:
\begin{eqnarray}
\Delta\tau &=& 1-n(\pi - \varphi)\nonumber\\
&=&1+\frac{1 + P_\pm A^2\cos^2\varphi}{P_\pm A^2 \cos\varphi -\sqrt{1-P_\pm A^2\sin^2\varphi}}.
\end{eqnarray} Since there are actually two distinct values of $P_\pm$ corresponding to birefringence, there are double images. If one sees both polarization states, then two clock readings can be seen. In the critical case $P_\pm A^2 =1$, the clock reading viewed through $\varphi = 0$ is infinitely delayed (and not visible). If $P_\pm A^2 \geq 1$, then the only visible   clocks  are at angles $\varphi$ such that $\csc^2\varphi > P_\pm A^{2}$. In a special range of cases where $P_-A^2 <1  \leq P_+ A^2$, the +  polarization modes cannot be seen  at all when viewed through  $\varphi$-angles such that  $\csc^2\varphi \leq P_\pm A^{2}$.

Since the  wave fronts of field disturbances are ellipsoidal in the $t,x,y,z$ coordinates, one might consider reconfiguring $S$ into an ellipsoidal arrangement so that the clocks will appear to be synchronized to the observer (provided that the field intensity is kept small enough that all points on $S$ are visible to the observer). However, due to birefringence, one would  only be able to manage the appearance of the  clocks as viewed through one polarization mode at a time.

\section{Hints of an optical black hole?}\label{Hints of an optical black hole} 
In the present chapter, we have found that  effective light cones in plane wave fields are tilted towards the direction of the Poynting vector of the plane wave  (Figure \ref{pwLightCones}). Field disturbances propagating in the direction of the Poynting travel at the usual speed of light, but in  other directions there is a drag effect. This drag effect is most pronounced for  disturbances that propagate in the direction exactly opposite to the Poynting vector. The speed of these field disturbances, as measured with respect to the background coordinates, is:
\begin{eqnarray}\label{eq1,3.6}
v(\pi) = \frac{1-P_\pm A^2}{1+P_\pm A^2},
\end{eqnarray}  where $A$ is the intensity of the plane wave.

Comparing  effective light cones for plane wave fields of higher and higher intensities as in Figure  \ref{pwLightCones},  one finds that the light cones become progressively more tilted. A similar phenomenon occurs in the geometry of gravitational black holes,\index{black holes} where light cones become progressively more and more tilted as one approaches the event horizon. At the event horizon, the light cones are so tilted that information cannot flow from the event horizon to the outside world.\index{event horizon} We suggest the notion   that an optical black hole would form  if one could increase the intensity of an electromagnetic wave by a sufficiently large amount in a localized region of space.\index{black holes}

Though such a field will no longer correspond to a true plane wave, we propose the following Gedankenexperiment. Consider an electromagnetic wave that is focusing to a point. Let us consider a spherical point-like implosion in which the intensity of the wave is assumed to follow the inverse square law.  Assuming that the wave front is locally like a plane wave, field disturbances that propagate radially outwards would travel at a coordinate  speed:
\begin{eqnarray}\label{eq2,3.6}
v=\frac{dr}{dt}=\frac{r^4-P_\pm A^2}{r^4 + P_\pm A^2}.
\end{eqnarray} Here, spherical coordinates $(t,r,\theta,\varphi)$ are implied. Equation (\ref{eq2,3.6}) is calculated by replacing $A$ with $A/r^2$ in Equation (\ref{eq1,3.6}) - in order to take the inverse square law into account.

Equation (\ref{eq2,3.6}) suggests that a spherical event horizon will form at  a critical radius  $P_\pm^{1/4}A^{1/2}$. That is, within the critical radius, ``outgoing" disturbances are not able to escape to infinity.

However, it is not possible to have a nontrivial spherically symmetric electromagnetic wave. The fundamental reason for this is that  the polarization vectors due to such a  field configuration would introduce a continuous nowhere-vanishing vector field tangent to the 2-sphere, thereby contradicting the well-known fact that the 2-sphere  is not parallelizable.

For this reason, we turn our attention to other configurations. Since the cylinder $S^1\times \mathbb{R}$ is parallelizable, the case of cylindrical collapse can be considered. The next chapter will look into this. 

As a tentative calculation for the cylindrical case, replacing $A$ with $A/r$ into Equation (\ref{eq1,3.6}) - in order to take the inverse distance law for cylindrical radiation into account - we have:
\begin{eqnarray}\label{outcylindricalapprox1}
v=\frac{dr}{dt}= \frac{r^2 - P_\pm A^2}{r^2 + P_\pm A^2},
\end{eqnarray}with cylindrical coordinates $(t,r,\theta,z)$ implied. Equation (\ref{outcylindricalapprox1}) suggests an effective event horizon at $r=P_\pm^{1/2} A$. Slightly more refined approximations,  done in Chapter \ref{chapter4},  yield an effective horizon at a radius that is proportional to the square of the intensity $A$ and inversely proportional to the frequency.

Using Equation (\ref{outcylindricalapprox1}) as an estimate for the coordinate velocity of radial null effective geodesics which are ``outgoing" in a cylindrically imploding wave field, and assuming that the ingoing null geodesics propagate at the speed of light, we can use \emph{Mathematica} to draw a graph of the coordinate velocities (see Figure \ref{Figure36A}). More detailed calculations, done in Chapter \ref{chapter4}, will confirm that this guess is qualitatively on the right track. In Figure \ref{Figure36A}, we are guessing ingoing rays will propagate at the usual speed of light. This guess is motivated by our experience with plane waves; we have seen that field disturbances that propagate along with the flow of a plane wave simply propagate at the usual speed of light.

\begin{figure}[h]
\center
    \includegraphics[height=3.75in]{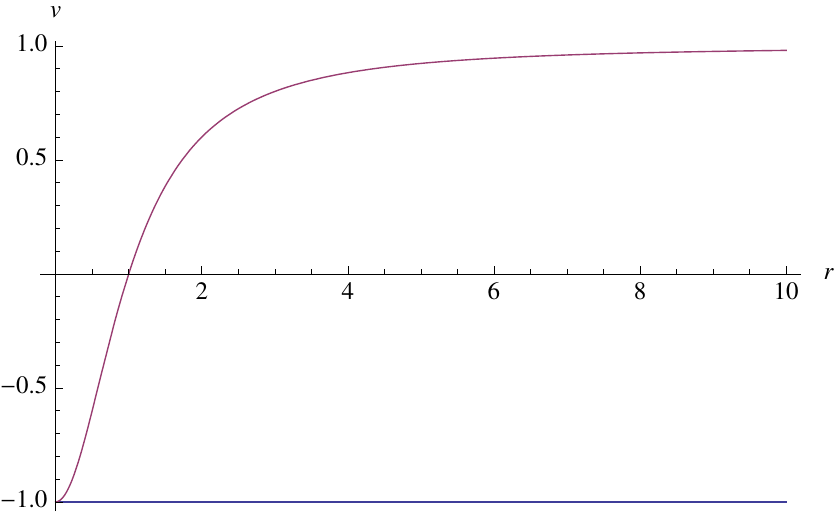}

    \caption[Thought experiment]{Thought experiment. In this graph, we have plotted (in red)  our  tentative guess on how the coordinate velocity $v$ of a radially outgoing light ray will vary with distance $r$ in   a cylindrically symmetric imploding field. This plot of Equation (\ref{outcylindricalapprox1}) uses $P_\pm A^2 =1$. The horizontal blue line (at $v=-1$) corresponds to our tentative  guess that the ingoing ray propagate at the usual speed of light. }
    \label{Figure36A}
\end{figure}

\chapter{Cylindrical fields
}\label{chapter4}

In Section \ref{Hints of an optical black hole} we suggested, by means of a crude  Gedankenexperiment, that optical black holes might arise in the effective geometry of an imploding electromagnetic wave. As we have noted, a spherically symmetric implosion is not possible because the 2-sphere $S^2$  is not parallelizable. On the other hand, the cylinder $S^1 \times \mathbb{R}$ by contrast is parallelizable. For this reason, in the present chapter we move our focus to cylindrically symmetric fields. The main formulas derived in Sections \ref{type} and \ref{maxwellianapproximation} have been checked with \emph{Mathematica}, as we detail in  Appendix \ref{Appendix:Key1}.

Naturally then, we will be working with cylindrical coordinates $(t, r, \theta, z)$ in which the background (Minkowskian) metric $g_{\mu\nu}$ is given by:
\begin{eqnarray}
ds^2 = dt^2 - dr^2 - r^2 d\theta^2 - dz^2.
\end{eqnarray} The nonvanishing Christoffel symbols are:
\begin{eqnarray}
\Gamma^r_{\theta\theta} = -r,\\
\Gamma^\theta_{r\theta} = \frac{1}{r} = \Gamma^\theta_{\theta r}.
\end{eqnarray}

We will define a \emph{cylindrically symmetric electromagnetic field} as being an  $A$-field whose components are functions of the coordinates $t$ and $r$ only:\index{cylindrically symmetric field}\index{electromagnetic field! cylindrically symmetric}
\begin{eqnarray}\label{cylindricalafield}
A_t &=&A_t(t,r), \\
A_r &=&A_r(t,r), \\
A_\theta &=&A_\theta(t,r), \\
A_z &=& A_z(t,r).\label{caf2}
\end{eqnarray}
Although this is not a necessary condition for making the physical field $F_{\mu\nu}= \partial_\mu A_\nu - \partial_\nu A_\mu$  cylindrically symmetric, it is a sufficient one. 

The Euler-Heisenberg field equation (\ref{EHFE}), is:
\begin{eqnarray}\label{EHFEch4}
\nabla_\mu F^{\mu\nu} = \frac{\alpha^2}{45}\left( 4 \nabla_\mu(F F^{\mu\nu})  +  7\nabla_\mu(G {F^*}^{\mu\nu})\right).
\end{eqnarray} This amounts to a system of nonlinear PDEs  (as shown explicitly  in  Section \ref{type} below). There are no known general methods  for finding exact solutions to such systems. We note however, that to first order in $\alpha$, Equation (\ref{EHFEch4}) reproduces the Maxwell vacuum equations: $\nabla_\mu F^{\mu\nu}=0$. This means that the familiar exact solutions   from Maxwell's theory (whether they are cylindrically symmetric or not) are approximations of solutions to Equation (\ref{EHFEch4}), up to first-order in $\alpha$. 

 In Section \ref{maxwellianapproximation} we will   treat ingoing cylindrical wave solutions from Maxwell's theory as  approximate first-order   solutions to the nonlinear theory. We will plug the Maxwellian field into the equations for the effective geometry from Chapter \ref{chapter2} and we will study the  resulting geometry  to second-order in $\alpha$.  As expected, we find that the effective geometry is analogous to a black hole. 
This is evidence for our main conjecture on the existence of black hole\index{black holes} soliton solutions as we explained in the Introduction.
 
 Since the Maxwellian solution is only valid up to first-order, one might worry about whether it is meaningful to do second-order calculations with it. However, as we show in Section \ref{maxwellianapproximation}, the geometric quantities that we calculate to second-order only depend on the first-order (Maxwellian) part of the exact solution, so the  approximation is justified.  
 
 In Section \ref{static}, we will derive an exact static  solution to Equation (\ref{EHFEch4}). This solution corresponds to a constant electric field running in the $z$-direction together with a magnetic field circulating around the $z$-axis. In the linear case, such a field configuration corresponds to that of a constant current through an infinitely long straight wire together with a constant electric field.  With the fields arranged so as to give an inwardly-directed Poynting vector,  the resulting effective geometry is  analogous to that of a black hole.

\section{Cylindrical  fields of a  particular type}\label{type}
Here we present formulas for the effective metrics and  field equations for cylindrically symmetric fields of a particular type.  These will be needed  in  later sections. 

\subsection{Field tensors}\label{5.1.1}
The particular type of cylindrical field considered here is one in which the $t$ and $r$-components of the $A$-field vanish. This would be the case, for example, with   an elliptically polarized imploding cylindrical radiation field  in the radiation gauge. 

Let us write:
\begin{eqnarray}\label{at}
A_t&\equiv&0,\\
A_r&\equiv &0,\\
A_\theta &=& u(t,r),\\
A_z&=& v(t,r).\label{az}
\end{eqnarray} 

The corresponding electromagnetic field tensor $F_{\mu\nu} = \partial_\mu A_\nu - \partial_\nu A_\mu$ is:
\begin{eqnarray}\label{farad}
F_{\mu\nu} &=& \left(\begin{array}{cccc}
F_{tt}&F_{tr} &F_{t\theta}& F_{tz}\\
F_{rt}&F_{rr} &F_{r\theta} &F_{rz}\\
F_{\theta t} & F_{\theta r} &F_{\theta \theta}&F_{\theta z}\\
F_{zt} & F_{zr}&F_{z\theta}&F_{zz}\\
\end{array}\right)\nonumber\\
&=&\left(\begin{array}{cccc}
0& 0 &\partial_t u & \partial_t v \\
0&0 &\partial_r u&\partial_r v\\
-\partial_t u& -\partial_r u  &0&0\\
-\partial_t v & -\partial_r v &0  &0\\
\end{array}\right).
\end{eqnarray} 
Raising the first index, one gets that $F^\mu_{\phantom{\mu}\nu} = g^{\lambda\mu} F_{\lambda_\nu}$ is: 
\begin{eqnarray}\label{farad1}
F^\mu_{\phantom{\mu}\nu} 
&=&\left(\begin{array}{cccc}
0& 0 & \partial_t u & \partial_t v \\
0&0 &-\partial_r u&-\partial_r v\\
\frac{1}{r^2}\partial_t u& \frac{1}{r^2}\partial_r u  &0&0\\
\partial_t v & \partial_r v &0  &0\\
\end{array}\right).
\end{eqnarray} 
Raising the second index, one gets that $F^{\mu\nu} = g^{\lambda\nu}F^\mu_{\phantom{\mu}\lambda}$ is:
\begin{eqnarray}\label{farad2}
F^{\mu\nu}
&=&\left(\begin{array}{cccc}
0& 0 &-\frac{1}{r^2}\partial_t u & -\partial_t v \\
0&0 &\frac{1}{r^2}\partial_r u&\partial_r v\\
\frac{1}{r^2}\partial_t u& -\frac{1}{r^2}\partial_r u  &0&0\\
\partial_t v & -\partial_r v &0  &0\\
\end{array}\right).
\end{eqnarray} Using (\ref{farad}) and (\ref{farad2}), one gets that the invariant $F=F_{\mu\nu}F^{\mu\nu}$ is: 
\begin{eqnarray}\label{F1}
F&=&\frac{2}{r^2}\left[ (\partial_r u)^2 - (\partial_t u)^2 + r^2(\partial_r v)^2 - r^2(\partial_t v)^2\right].
\end{eqnarray}
The dual tensor $F^*_{\mu\nu}= \frac{1}{2}\varepsilon_{\alpha\beta\mu\nu} F^{\alpha\beta}$ is:
\begin{eqnarray}\label{max}
F^*_{\mu\nu} &=&\left(\begin{array}{cccc}
0& 0 &-r\partial_r v & \frac{1}{r}\partial_r u \\
0&0 &-r \partial_t v&\frac{1}{r}\partial_t u\\
r \partial_r v& r \partial_t v  &0&0\\
-\frac{1}{r}\partial_r u & -\frac{1}{r}\partial_t u &0  &0\\
\end{array}\right).
\end{eqnarray} Using (\ref{farad2}) and (\ref{max}), one gets that $G=F^*_{\mu\nu}F^{\mu\nu}$ is:
\begin{eqnarray}\label{G1}
G&=&\frac{4}{r}\left[(\partial_t u)( \partial_ r v)  - (\partial_r u)(\partial_t v)\right].
\end{eqnarray}

Raising the indices of the dual, one gets that ${F^*}^{\mu\nu} = g^{\alpha\mu}g^{\beta\nu}F^*_{\alpha\beta}$ is:
\begin{eqnarray}\label{max2}
{F^*}^{\mu\nu} &=&\left(\begin{array}{cccc}
0& 0 &\frac{1}{r}\partial_r v & -\frac{1}{r}\partial_r u \\
0&0 &-\frac{1}{r} \partial_t v&\frac{1}{r}\partial_t u\\
-\frac{1}{r} \partial_r v& \frac{1}{r} \partial_t v  &0&0\\
\frac{1}{r}\partial_r u & -\frac{1}{r}\partial_t u &0  &0\\
\end{array}\right).
\end{eqnarray} 

\subsection{Effective metric}\label{5.1.2}
As we saw in Section \ref{nullgeo},  for the Euler-Heisenberg theory, the effective metric $\widetilde{g}_{\mu\nu}$ is given by the inverse of:
\begin{eqnarray}\label{cometricch4}
\widetilde{g}^{\mu\nu} &=&g^{\mu\nu} + \Lambda_\pm F^\mu_{\phantom{\mu}\lambda} F^{\lambda \nu},
\end{eqnarray} where:
\begin{eqnarray}
\Lambda_\pm &=& \frac{224\alpha^2 }{495+ 12 F \alpha^2 \mp \sqrt{18225 - 18360 F\alpha^2 + 4624 F^2 \alpha^4 + 3136 G^2 \alpha^4 }}.\nonumber\\
\end{eqnarray}
Henceforth let us write $\Lambda:= \Lambda_\pm$, keeping in mind that there is a choice of $\pm$ involved in the calculation of $\Lambda$ (this $\Lambda$ has nothing to do with the cosmological constant). As we explained in Section  \ref{polarization}, this choice of $\pm$ depends on the polarization state of the field disturbance.

Using Equations (\ref{farad1}) and (\ref{farad2}), we get that the only nonvanishing components of  the substress tensor $F^\mu_{\phantom{\mu}\lambda} F^{\lambda \nu}$ are:
\begin{eqnarray}\label{substressTT}
F^t_{\phantom{t}\lambda} F^{\lambda t}&=&\frac{1}{r^2}\left[ (\partial_t u)^2 + r^2 (\partial_t v)^2\right]\\
F^t_{\phantom{t}\lambda} F^{\lambda r}=F^r_{\phantom{r}\lambda} F^{\lambda t}&=&-\frac{1}{r^2}\left[(\partial_t u)(\partial_r u) +r^2 (\partial_t v )(\partial_r v)\right] \\
F^r_{\phantom{r}\lambda} F^{\lambda r}&=&\frac{1}{r^2}\left[(\partial_r u)^2 + r^2 (\partial_r v)^2\right]\\
F^\theta_{\phantom{\theta}\lambda} F^{\lambda \theta}&=&\frac{1}{r^4}\left[ (\partial_r u)^2 - (\partial_t u)^2 \right]\\
F^\theta_{\phantom{\theta}\lambda} F^{\lambda z}=F^z_{\phantom{z}\lambda} F^{\lambda \theta}&=&\frac{1}{r^2}\left[ (\partial_r u)(\partial_r v)-(\partial_t u)(\partial_t v)\right] \\
F^z_{\phantom{z}\lambda} F^{\lambda z}&=&(\partial_r v)^2 - (\partial_t v)^2\label{substressZZ}
\end{eqnarray}

Plugging  our result for the substress tensor into Equation (\ref{cometricch4}),  we get the effective cometric $\widetilde{g}^{\mu\nu}$. Taking the inverse of $\widetilde{g}^{\mu\nu}$, we find  that the only nonvanishing components of the effective metric are (up to a conformal factor $\kappa$): 
\begin{eqnarray}\label{emetric4.28}
\kappa \widetilde{g}_{tt}&=&1-\frac{\Lambda(\partial_r u )^2}{r^2} - \Lambda(\partial_r v)^2 \\
\kappa \widetilde{g}_{tr}=\kappa\widetilde{g}_{rt}&=&-\frac{\Lambda (\partial_t u)(\partial_r u)}{r^2} -  \Lambda (\partial_t v)(\partial_r v)\\
\kappa\widetilde{g}_{rr}&=&-1 -\frac{\Lambda(\partial_t u)^2}{r^2}-\Lambda(\partial_t v)^2\\
\kappa\widetilde{g}_{\theta\theta}&=&-r^2 +\Lambda r^2 (\partial_r v)^2 -\Lambda r^2 (\partial_t v)^2\\
\kappa\widetilde{g}_{\theta z}=\kappa\widetilde{g}_{\theta z}&=&\Lambda  (\partial_t u)(\partial_t v) - \Lambda(\partial_r u)(\partial_r v )\\
\kappa\widetilde{g}_{zz}&=&-1 +\frac{\Lambda (\partial_r u)^2}{r^2} - 
\frac{\Lambda (\partial_t u)^2}{r^2},\label{emetric4.33}
\end{eqnarray} \index{effective geometry! of cylindrically symmetric fields}
Since only the null geodesics are important, we can drop the conformal factor $\kappa$. Note that the effective metric becomes conformally equivalent to  the background metric in the limit where $\Lambda\rightarrow 0$. 

\subsection{Radial null geodesics}\label{5.1.3}

Since the components of the metric tensor only depend on the coordinates $t$ and $r$, it follows that   radial null curves (that is, null curves with constant $\theta$ and $z$ coordinates), are automatically geodesics. For such a curve we can write:
\begin{eqnarray}\label{radial1}
0&=& \widetilde{g}_{tt} + 2\widetilde{g}_{tr} \frac{dr}{dt} + \widetilde{g}_{rr}\left(\frac{dr}{dt}\right)^2.
\end{eqnarray} Solving (\ref{radial1}) for $dr/dt$ gives:
\begin{eqnarray}\label{radial2}
\frac{dr}{dt} &=& \frac{-\widetilde{g}_{tr}\pm\sqrt{\widetilde{g}_{tr}^{\ 2}- \widetilde{g}_{tt}\widetilde{g}_{rr}}}{\widetilde{g}_{rr}}\nonumber\\
&=&-\frac{(\partial_t u)(\partial_r u) + r^2 (\partial_t v)(\partial_r v)}{(\partial_t u)^2 + r^2(\partial_t v)^2+\frac{r^2}{\Lambda} }\nonumber\\
&&\ \mp\frac{\sqrt{\left[(\partial_t u)(\partial_r u) + r^2(\partial_t v)(\partial_r v)\right]^2 + \left(\frac{r^2}{\Lambda}+r^2(\partial_t v)^2+ (\partial_t u)^2\right)\left(\frac{r^2}{\Lambda}-r^2(\partial_r v)^2-(\partial_r u)^2\right)}}{(\partial_t u)^2 + r^2(\partial_t v)^2+\frac{r^2}{\Lambda}}.\nonumber\\
\end{eqnarray} We define \emph{outgoing} geodesics as corresponding to choosing $+$ in (the second line of) Equation (\ref{radial2}) and \emph{ingoing} geodesics as corresponding to choosing $-$.

For radial geodesics of the ingoing type,  one gets that  to second-order in $\alpha$:
\begin{eqnarray}\label{in2o}
\frac{dr}{dt}\Big|_{\textrm{in}}&=&-1+\frac{(11\pm 3)\alpha^2}{45r^2}\left((\partial_t u - \partial_r u)^2 \phantom{\Big|}+\phantom{\Big|} r^2 (\partial_t v-\partial_r v)^2\right) + O(\alpha^4).
\end{eqnarray}
For the outgoing type:
\begin{eqnarray}\label{out2o}
\frac{dr}{dt}\Big|_{\textrm{out}} &=&1-\frac{(11\pm 3)\alpha^2}{45r^2}\left((\partial_t u + \partial_r u)^2 \phantom{\Big|}+\phantom{\Big|} r^2 (\partial_t v + \partial_r v)^2\right) + O(\alpha^4).\end{eqnarray} In Equations (\ref{in2o}) and (\ref{out2o}), the choice of $\pm$ has to do with the polarization state of the disturbance (there is birefringence). With the  birefringence averaged out, we have:
\begin{eqnarray}\label{in2ave}
\Big\langle\frac{dr}{dt}\Big\rangle\Big|_{\textrm{in}}&=&-1+\frac{11\alpha^2}{45r^2} \left((\partial_t u - \partial_r u)^2 \phantom{\Big|}+\phantom{\Big|} r^2 (\partial_t v-\partial_r v)^2\right) + O(\alpha^4),
\end{eqnarray}
and:
\begin{eqnarray}\label{out2ave}
\Big\langle\frac{dr}{dt}\Big\rangle \Big|_{\textrm{out}}&=& 1-\frac{11\alpha^2}{45r^2}\left((\partial_t u + \partial_r u)^2 \phantom{\Big|}+\phantom{\Big|} r^2 (\partial_t v + \partial_r v)^2\right) +  O(\alpha^4).
\end{eqnarray}

\subsection{Field equations}\label{5.1.4}
 Our next task is to express the field equations as a system of nonlinear PDEs. To this end, 
using (\ref{farad2}), 
and introducing $\hat u(t,r) := u(t,r)/r$, we get that  the left-hand side of the field equation (\ref{EHFEch4}) is:
\begin{eqnarray}\label{lefthandside}
\nabla_\mu F^{\mu\nu} &=&\left\{\begin{array}{cl}
0& \textrm{for } \nu = t\\
0&\textrm{for } \nu = r\\
\frac{1}{r^2}\partial_r (r\partial_r \hat u) - \frac{1}{r}\partial_t^2 \hat u  - \frac{\hat u}{r^3}  &\textrm{for }\nu = \theta\\
\frac{1}{r}\partial_r(r \partial_r v) - \partial_t^2 v& \textrm{for }\nu = z.\\
\end{array}\right.
\end{eqnarray}
Using (\ref{farad2}), (\ref{F1}), (\ref{max2}) and (\ref{G1}), we get that the right-hand side of the field equation (\ref{EHFEch4}) is:
\begin{eqnarray}\label{righthandside}
 \frac{\alpha^2}{45}\left( 4 \nabla_\mu(F F^{\mu\nu})  +  7\nabla_\mu(G {F^*}^{\mu\nu})\right)&=&\left\{\begin{array}{cl}
0& \textrm{for } \nu = t\\
0&\textrm{for } \nu = r\\
\frac{4\alpha^2}{45r^5}\mathcal{U} &\textrm{for }\nu = \theta\\
\frac{4\alpha^2}{45r^3}\mathcal{V}& \textrm{for }\nu = z.\\
\end{array}\right.
\end{eqnarray}
Here,
\begin{eqnarray}\label{mathfrakU}
\mathcal{U}&=& 2r\hat u^2 \left(3r\partial_r^2 \hat u - 3\partial_r \hat u -r \partial_t^2 \hat u \right)-6 \hat u^3 +\nonumber\\
&&r^2 \hat u \left[6(\partial_r \hat u)^2 - 2(\partial_r v)^2 -5(\partial_t v)^2 -2(\partial_t \hat u)\left(\partial_t \hat u + 4r \partial_t\partial_r \hat u\right) +  \right.\nonumber\\
&&\left.3r(\partial_t v)(\partial_t \partial_r v) + 4r(\partial_r \hat u) \left(3\partial_r^2 \hat u - \partial_t^2 \hat u\right)+r(\partial_r v)\left(4\partial_r^2 v -7\partial_t^2 v\right)\right]+\nonumber\\
&&r^3\left\{6(\partial_r \hat u)^3 -(\partial_r v)\left[7(\partial_t v)\left(\partial_t \hat u + 2 r \partial_t\partial_r \hat u\right) - 3r(\partial_t\hat u)(\partial_t \partial_r v)\right] + \right.\nonumber\\
&&r(\partial_r v)^2\left(2\partial_r^2 \hat u + 5\partial_t^2 \hat u\right)+ (\partial_r \hat u)^2\left(6r\partial_r^2 \hat u - 2r\partial_t^2\hat u\right)+\nonumber\\
&&(\partial_r \hat u)\left[2(\partial_r v)^2 - 6 (\partial_t \hat u)^2-8r(\partial_t \hat u)(\partial_t \partial_r \hat u)+\right.\nonumber\\
&&\left. (\partial_t v)(5\partial_t v + 3r\partial_t\partial_r v) +r(\partial_r v)\left(4\partial_r^2 v  - 7\partial_t^2 v\right)\right] +\nonumber\\
&&r\left[(\partial_r^2 \hat u)\left(5(\partial_t v)^2-2(\partial_t\hat u)^2\right) -7(\partial_r^2 v)(\partial_t \hat u)(\partial_t v)\right.+\nonumber\\
&&\left.\left.2(\partial_t^2 \hat u)\left(3(\partial_t\hat u)^2 + (\partial_t v)^2\right) +4(\partial_t \hat u)(\partial_t v)(\partial_t^2 v)\right]\right\},
\end{eqnarray} 

and:
\begin{eqnarray}\label{mathfrakV}
\mathcal{V}&=&\hat u^2\left(2r\partial_r^2 v -2\partial_r v + 5r\partial_t^2 v\right)+\nonumber\\
&&r\hat u\left[(\partial_t v)\left(10\partial_t \hat u + 3r\partial_t \partial_r \hat u\right) - 14 r (\partial_t \hat u)(\partial_t \partial_r v)+\right.\nonumber\\
&&\left.r(\partial_r v)\left(4\partial_r^2\hat u - 7 \partial_t^2 \hat u\right)+ 2(\partial_r \hat u)\left(2\partial_r v + 2r\partial_r^2 v + 5 r\partial_t^2 v\right)\right]+\nonumber\\
&&r^2\left\{2(\partial_r v)^3 -(\partial_r v)\left[2(\partial_t\hat u)^2 - 3 r (\partial_t\hat u)(\partial_t\partial_r \hat u)+\right.2(\partial_t v)\left(\partial_t v + 4 r \partial_t\partial_r v\right)\right] + \nonumber\\
&&(\partial_r \hat u)\left[3r(\partial_t v)(\partial_t\partial_r \hat u) -\right.2(\partial_t\hat u ) \left(2\partial_t v + 7 r \partial_t \partial_r v\right) +\nonumber\\
&&\left.r(\partial_r v)\left(4\partial_r^2 \hat u - 7 \partial_t^2 \hat u\right)\right]+(\partial_r v)^2\left(6r\partial_r^2 v - 2r \partial_t^2 v\right)+\nonumber\\
&&(\partial_r \hat u)^2\left(6\partial_r v + 2r\partial_r^2 v + 5r\partial_t^2 v\right) +r\left[(\partial_r^2 v)\left(5(\partial_t\hat u)^2 - 2 (\partial_t v)^2\right)+\right.\nonumber\\
&&(\partial_t \hat u)(\partial_t v)\left(4\partial_t^2 \hat u - 7 \partial_r^2 \hat u\right) +\left.\left.2(\partial_t^2 v ) \left((\partial_t \hat u )^2 + 3 (\partial_t v)^2\right)\right]\right\}.
\end{eqnarray}

The field equations for a field of the kind specified by Equations (\ref{at}) - (\ref{az}) can thus be expressed as a nonlinear system of PDEs:
\begin{eqnarray}\label{system4}
\left\{\begin{array}{l}
\frac{1}{r}\partial_r (r \partial_r \hat u)-\partial_t^2 \hat u - \frac{\hat u}{r^2}= \frac{4\alpha^2}{45r^4}\mathcal{U} \\
\frac{1}{r}\partial_r(r\partial_r v)- \partial_t^2 v=\frac{4\alpha^2}{45r^3}\mathcal{V}\\
\end{array}\right.
\end{eqnarray}

\section{Maxwellian approximation}\label{maxwellianapproximation}
To first-order in $\alpha$, in which Maxwell's theory is recovered, the field equations  (\ref{system4}) become:
\begin{eqnarray}\label{maxwellappsystem}
\left\{\begin{array}{l}
\frac{1}{r}\partial_r (r \partial_r \hat u)-\partial_t^2 \hat u - \frac{\hat u}{r^2}=0 \\
\frac{1}{r}\partial_r(r\partial_r v)- \partial_t^2 v=0.\\
\end{array}\right.
\end{eqnarray}Seeking solutions to (\ref{maxwellappsystem}) for a monochromatic field of constant frequency $\omega>0$, we  express the components of the $A$-field  (\ref{at}) - (\ref{az}) in the form:\begin{eqnarray}\label{ansatz0}
A_\mu = \textrm{Re}\left[S_\mu \exp(i\omega t)\right],
\end{eqnarray} where $S_\mu$ is a function of $r$ only. Note that the ansatz (\ref{ansatz0}) implies that $\partial_t^2 A_\mu = -\omega^2 A_\mu$, and so (\ref{maxwellappsystem}) becomes:
\begin{eqnarray}\label{besselsystem1}
\left\{\begin{array}{l}
\frac{1}{r}\partial_r (r \partial_r \hat u)+\left(\omega^2 - \frac{1}{r^2} \right)\hat u=0 \\
\frac{1}{r}\partial_r(r\partial_r v)+ \omega^2 v=0.\\
\end{array}\right.
\end{eqnarray} These equations are of the form:
\begin{eqnarray}\label{Besseltypeeq}
\frac{1}{r}\partial_r \left(r\partial_r \psi \right) + \left(\omega^2 - \frac{n}{r^2}\right)\psi =0,
\end{eqnarray} where $n$ is either 0 or 1. Equation (\ref{Besseltypeeq}) is a Bessel-type differential equation, having solutions of the form (see e.g., Bowman \cite{Bowman} p. 116):\index{Bessel functions}
\begin{eqnarray}
\psi =c_{1}J_n(\omega r) + c_{2}Y_n(\omega r),
\end{eqnarray} where $c_1$ and $c_2$ are complex constants,  and $J_n$ and $Y_n$ denote the $n$th order Bessel functions of the first and second kinds respectively. \index{Bessel functions}
Monochromatic solutions to (\ref{maxwellappsystem}) are therefore given by:
\begin{eqnarray}
\hat u &=& \textrm{Re}\left[\left(c_{\hat u1}J_1(\omega r) + c_{\hat u2}Y_1(\omega r)\right)\exp(i\omega t)\right]\\
v&=&\textrm{Re}\left[\left(c_{v1}J_0(\omega r) + c_{v2}Y_0(\omega r)\right)\exp(i\omega t)\right],
\end{eqnarray} where the $c_{ij}$ are complex constants.

In order to choose the constants $c_{ij}$ so that one gets radially propagating solutions, consider the fact that the graphs of $J_n$ and $Y_n$ look like dampened sine and cosine graphs. In this sense,  combinations such as $J_n(\omega r)\pm i Y_n(\omega r)$ are like dampened versions of $\exp(i \omega r)$. So, after taking the real part, the functions $\left(J_n(\omega r) \pm iY_n(\omega r)\right)\exp(i\omega t)$ describe  radially propagating waves (which can be either ingoing or outgoing depending on the choice of $\pm$).  

Accordingly, for an elliptically polarized ingoing cylindrical wave,\index{cylindrical wave (Maxwell)} we write:\index{Maxwell! cylindrical wave}
\begin{eqnarray}\label{4.53}
u(t,r) &=&r\cdot \textrm{Re}\left[\frac{U}{\omega}\left(J_1(\omega r) + iY_1(\omega r)\right)\exp(i\omega t)\right]\nonumber\\
&=& \frac{Ur}{\omega} \left(J_1(\omega r) \cos(\omega t)\phantom{\Big|} -\phantom{\Big|}Y_1(\omega r)\sin (\omega t)\right)  ,
\end{eqnarray} and:
\begin{eqnarray}\label{4.54}
v(t,r)&=&  \textrm{Re}\left[\frac{V}{\omega}\left(J_0(\omega r)+iY_0(\omega r)\right)\exp(i\omega t)\right]\nonumber\\
&=& \frac{V}{\omega} \left(J_0(\omega r) \cos(\omega t)\phantom{\Big|} - \phantom{\Big|}Y_0(\omega r)\sin (\omega t)\right)  ,
\end{eqnarray} where $U$ and $V$ are real-valued constants (not to be confused with the  functions $\mathcal{U}$ and $\mathcal{V}$ as defined by (\ref{mathfrakU}) and (\ref{mathfrakV})). 

Our task is to study the effective metric corresponding to this wave field. More specifically, we want to have a look  at the effective radial null geodesics using (\ref{in2ave}) and (\ref{out2ave}) to calculate $dr/dt$ to second-order in $\alpha$, where we use the Maxwellian solution to evaluate the field variables. We claim that, up to second-order in $\alpha$, the calculation of $dr/dt$ only depends on the first-order (Maxwellian) part  of the exact solution to (\ref{system4}). In other words:

\begin{thm}\label{justificationofMA} Maxwellian approximations for  $dr/dt$ (along radial null geodesics in the effective geometry) are   accurate up to second-order. 
\end{thm}
\begin{proof}  
Suppose that we have an exact solution $(\hat u, v)= (s_\theta /r, s_z )$ for the nonlinear system (\ref{system4}).  In the limit $\alpha^2\rightarrow 0$, this solution becomes a  solution to  (\ref{maxwellappsystem}). So expanding the exact solution as a power series in $\alpha$ would yield $s_\theta=  m_\theta + O(\alpha^2) $ and $s_z= m_z + O(\alpha^2)$,  where $(\hat u,v)=(m_\theta/r,m_z)$ is an exact solution to (\ref{maxwellappsystem}). (Note: \emph{a priori} the series may only be asymptotic.) Using Equation (\ref{out2ave}):
\begin{eqnarray}
\Big\langle\frac{dr}{dt}\Big\rangle \Big|_{\textrm{out}}&=& 1-\frac{11\alpha^2}{45r^2}\left((\partial_t s_\theta+ \partial_r s_\theta)^2 \phantom{\Big|}+\phantom{\Big|} r^2 (\partial_t s_z + \partial_r s_z)^2\right) +  O(\alpha^4)\nonumber\\
&=&1-\frac{11\alpha^2}{45r^2}\left((\partial_t m_\theta+ \partial_r m_\theta + O(\alpha^2))^2 \phantom{\Big|}+\phantom{\Big|} r^2 (\partial_t m_z + \partial_r m_z + O(\alpha^2))^2\right) +  O(\alpha^4)\nonumber\\
&=&1-\frac{11\alpha^2}{45r^2}\left((\partial_t m_\theta+ \partial_r m_\theta )^2 \phantom{\Big|}+\phantom{\Big|} r^2 (\partial_t m_z + \partial_r m_z ))^2\right) +  O(\alpha^4).
\end{eqnarray} Similar calculations can be done using Equations (\ref{in2o}) - (\ref{in2ave}). \end{proof}

Proceeding now, by specializing  Equations (\ref{in2ave}) and (\ref{out2ave}) to the Maxwellian solution  (\ref{4.53}) and  (\ref{4.54}),  we find that the radial null geodesics are described by:
\begin{eqnarray}\label{inMEave}
\Big\langle\frac{dr}{dt}\Big\rangle\Big|_{\textrm{in}}&=&-1+\frac{11\alpha^2}{45} 
\left\{U^2\left[\left(J_0(\omega r) + Y_1(\omega r)\right)\cos(\omega t)  \phantom{\Big|}+ \phantom{\Big|} \left(J_1(\omega r) - Y_0(\omega r)\right)\sin(\omega t)\right]^2 
+ \right.\nonumber\\
&&\left.
\phantom{\cdot-1+\frac{11\alpha^2}{45} 
}V^2\left[\left(Y_0(\omega r) - J_1(\omega r)\right)\cos(\omega t)  \phantom{\Big|}+ \phantom{\Big|} \left(Y_1(\omega r) +J_0(\omega r)\right)\sin(\omega t)\right]^2
\right\} + O(\alpha^4),\nonumber\\
\end{eqnarray}
and:
\begin{eqnarray}\label{outMEave}
\Big\langle\frac{dr}{dt}\Big\rangle \Big|_{\textrm{out}}&=& 1-\frac{11\alpha^2}{45}
\left\{U^2\left[\left(Y_1(\omega r) - J_0(\omega r)\right)\cos(\omega t)  \phantom{\Big|}+ \phantom{\Big|} \left(J_1(\omega r) + Y_0(\omega r)\right)\sin(\omega t)\right]^2 
+ \right.\nonumber\\
&&\left.
\phantom{-1+\frac{11\alpha^2}{45} 
}V^2\left[\left(J_1(\omega r) + Y_0(\omega r)\right)\cos(\omega t)  \phantom{\Big|}+ \phantom{\Big|} \left(J_0(\omega r) -Y_1(\omega r)\right)\sin(\omega t)\right]^2
\right\} + O(\alpha^4),\nonumber\\
\end{eqnarray} We note that the oscillatory terms involving trigonometric functions of $t$ disappear in the case of circular polarization (where $U=V$). One might have expected this out of consideration of the fact that, as we saw in Chapter \ref{planewaves}, a similar simplification occurs in the effective geometry of circularly polarized  plane waves. In  fact, the stress-energy tensor (at least, as computed using the Maxwellian Lagrangian $L=-F/4$) for the   field given by Equations (\ref{4.53}) and (\ref{4.54}) does not have any oscillatory terms in the case where $U=V$.

Henceforth, let us assume that the wave is circularly polarized, with $U=V=:A$. In this case, we get that the effective radial geodesics are described by:
\begin{eqnarray}\label{inMCave}
\Big\langle\frac{dr}{dt}\Big\rangle\Big|_{\textrm{in}}&=&-1+\frac{11\alpha^2 A^2}{45} \left(-\frac{4}{\pi \omega r} + J_0(\omega r)^2 +J_1(\omega r)^2 + Y_0(\omega r)^2 + Y_1(\omega r)^2\right) + O(\alpha^4),\nonumber\\
\end{eqnarray}
and:
\begin{eqnarray}\label{outMCave}
\Big\langle\frac{dr}{dt}\Big\rangle \Big|_{\textrm{out}}&=& 1-\frac{11\alpha^2 A^2}{45}\left(\frac{4}{\pi \omega r} + J_0(\omega r)^2 +J_1(\omega r)^2 + Y_0(\omega r)^2 + Y_1(\omega r)^2\right) +  O(\alpha^4).\nonumber\\
\end{eqnarray}For simplicity, we are using the formulas in which the birefringence is averaged out. To recover the birefringence,  replace the factor $11\alpha^2$ with $(11\pm 3)\alpha^2$.  \index{cylindrical wave (Maxwell)! effective geometry} \index{Maxwell! cylindrical wave! effective geometry} 

We note that, for large $x$, one has the approximations (Arfken and Weber \cite{ArfkenWeber} p. 718):
\begin{eqnarray}
J_n(x)\approx \sqrt{\frac{2}{\pi x}}\cos\left[x-\left(n+\frac{1}{2}\right)\left(\frac{\pi}{2}\right)\right],
\end{eqnarray}and:
\begin{eqnarray}
Y_n(x)\approx \sqrt{\frac{2}{\pi x}}\sin\left[x-\left(n+\frac{1}{2}\right)\left(\frac{\pi}{2}\right)\right].
\end{eqnarray} Consequently, for large $x$:
\begin{eqnarray}\label{identityapprox}
J_0(x)^2 +J_1(x)^2 + Y_0(x)^2 + Y_1(x)^2
&\approx&\frac{2}{\pi x}\cos^2\left[x-\left(\frac{1}{2}\right)\left(\frac{\pi}{2}\right)\right]\nonumber\\
&&
+ \frac{2}{\pi x}\cos^2\left[x-\left(1+\frac{1}{2}\right)\left(\frac{\pi}{2}\right)\right]\nonumber\\
&&
+ \frac{2}{\pi x}\sin^2\left[x-\left(\frac{1}{2}\right)\left(\frac{\pi}{2}\right)\right]
\nonumber\\
&&
+\frac{2}{\pi x}\sin^2\left[x-\left(1+\frac{1}{2}\right)\left(\frac{\pi}{2}\right)\right]\nonumber\\
&=&\frac{4}{\pi x}.
\end{eqnarray}

In fact, one could have guessed at Equation (\ref{identityapprox}) using the following idea. Far from the origin, the cylindrical wave should look like a plane wave and we know that field disturbances in a plane wave, when they travel along the same direction as the plane wave (the Poynting vector), travel at the usual speed of light. Hence one should have $\Big\langle\frac{dr}{dt}\Big\rangle\Big|_{\textrm{in}} \approx -1$ in the limit where  $\omega r$ is large, and this implies (\ref{identityapprox}). 

So in the limit where the quantity $\omega r$ is large, we have, to second-order in $\alpha$:
\begin{eqnarray}\label{inasym}
\Big\langle\frac{dr}{dt}\Big\rangle\Big|_{\textrm{in}}&\approx&-1,
\end{eqnarray}and:
\begin{eqnarray}\label{outasym}
\Big\langle\frac{dr}{dt}\Big\rangle\Big|_{\textrm{out}}&\approx&1-\frac{88\alpha^2 A^2}{45\pi\omega r}.
\end{eqnarray}   \index{cylindrical wave (Maxwell)! effective geometry! asymptotic approximation} \index{Maxwell! cylindrical wave! effective geometry!}  
Equation (\ref{outasym}) suggests that within radii $r\leq r_c$, where:
\begin{eqnarray}\label{ehapprox1}
r_c \approx \frac{88\alpha^2 A^2}{45\pi \omega},
\end{eqnarray} even the ``outward" geodesics are compelled to fall inward.  Hence the critical radius $r_c$ is the event horizon\index{event horizon} of a black hole. (Note: we have only checked this for outward geodesics in the radial direction.)  

Since (\ref{ehapprox1}) was derived assuming that $\omega r$ is large,  we only expect this approximation to hold in  the limit of very large $A^2$. (We could also mention that, due to birefringence, there are actually two event horizons. If we had taken this into account in the above, then we would have estimated the critical radii as occurring at $r_c\approx(88\pm24)\alpha^2 A^2/(45\pi\omega)$, where the $\pm$ depends on the polarization of the disturbance.)

Plotting Equations (\ref{inasym}) and (\ref{outasym}) on the same graph (Figure \ref{figure40}), we can compare  to our  earlier  na\"ive guess of (\ref{outcylindricalapprox1}). We note  that there are substantial quantitative differences between our initial guess and our slightly more refined calculation, but the qualitative picture is basically the same. 

\begin{figure}[h]
\center
    \includegraphics[height=4in]{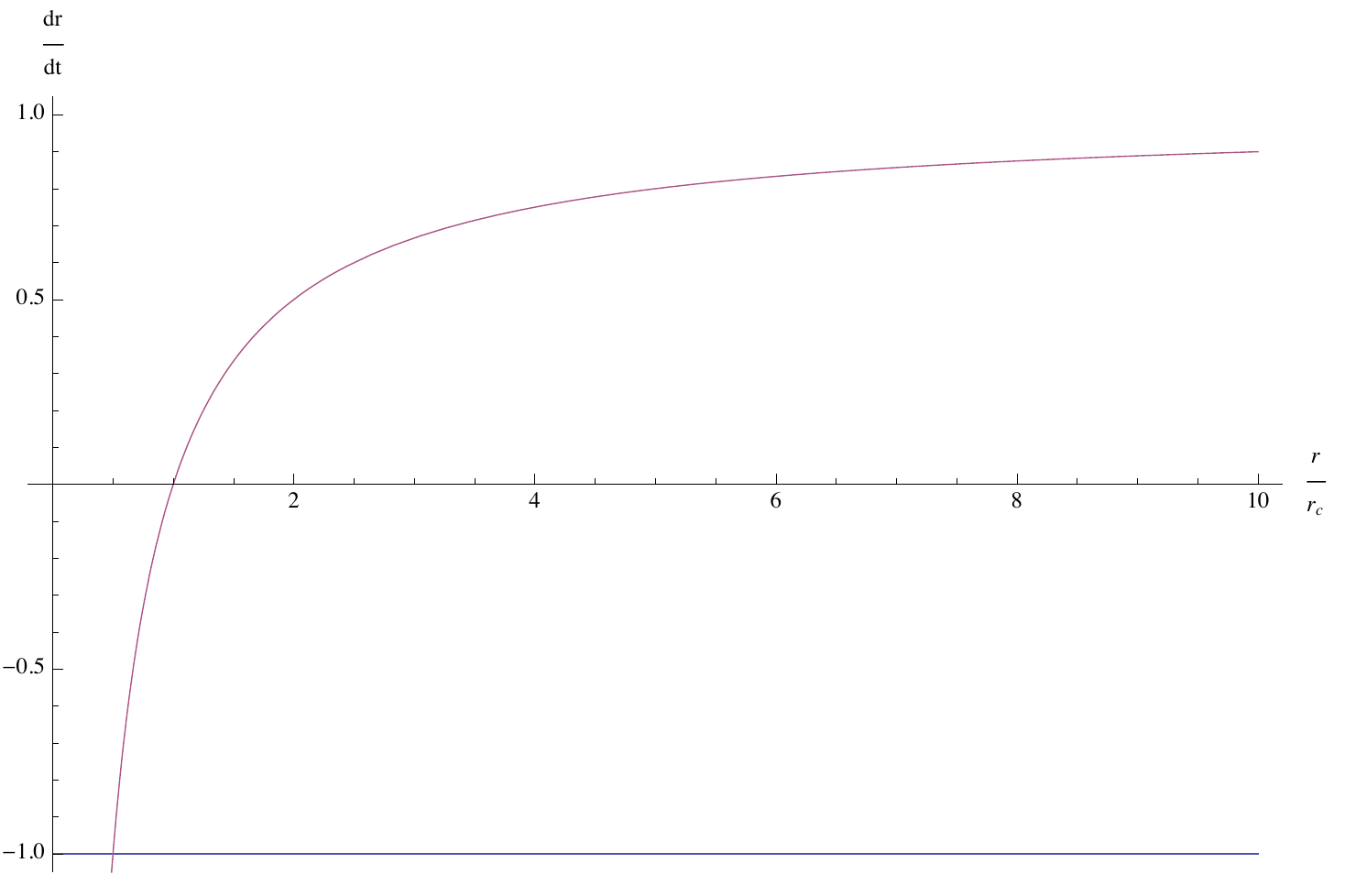}

    \caption[Coordinate velocities of  null geodesics in the effective geometry of an ingoing cylindrical wave according to the ``asymptotic approximation"]{In this graph, the coordinate velocities of  effective null geodesics are plotted using the asymptotic approximations (\ref{inasym}) and (\ref{outasym}), which assume that the quantity  $\omega r$ is large. The horizontal blue line at $dr/dt =-1$ corresponds to the ingoing geodesics, and the red curve corresponds to the ``outgoing" radial geodesics. Compare to Figure \ref{Figure36A}. }

    \label{figure40}
\end{figure}

Treating the asymptotic approximations (\ref{inasym}) and (\ref{outasym}) as ordinary differential equations, and solving them by integration, we obtain approximate  equations  for the radial null geodesics.  Specifically:
\begin{eqnarray}\label{inasymapproxt}
t &=& -r + r_0,
\end{eqnarray}
for the ingoing geodesics ($r_0$ := the radial coordinate of the geodesic when $t=0$), and:
\begin{eqnarray}\label{outasymapproxt}
t&=&\left\{\begin{array}{ll}
r+r_c\ln(r-r_c)-r_0-r_c\ln(r_0-r_c)&\textrm{if }r_0 > r_c\\
r+r_c\ln(r_c-r)-r_0-r_c\ln(r_c-r_0)&\textrm{if }r_0<r_c,
\end{array}\right.
\end{eqnarray}for the ``outgoing" geodesics.  An outgoing radial  null geodesic that initiates from $r=r_c$ would just remain there. Using \emph{Mathematica}, we have plotted Equations (\ref{inasymapproxt}) and (\ref{outasymapproxt}) for a few values of $r_0$. The resulting plot is shown in Figure \ref{figure4.1}. Note that the effective light cones are tilted in towards the origin, just like in the situation of gravitational black holes. 
\begin{figure}[h]
\center
    \includegraphics[height=4in]{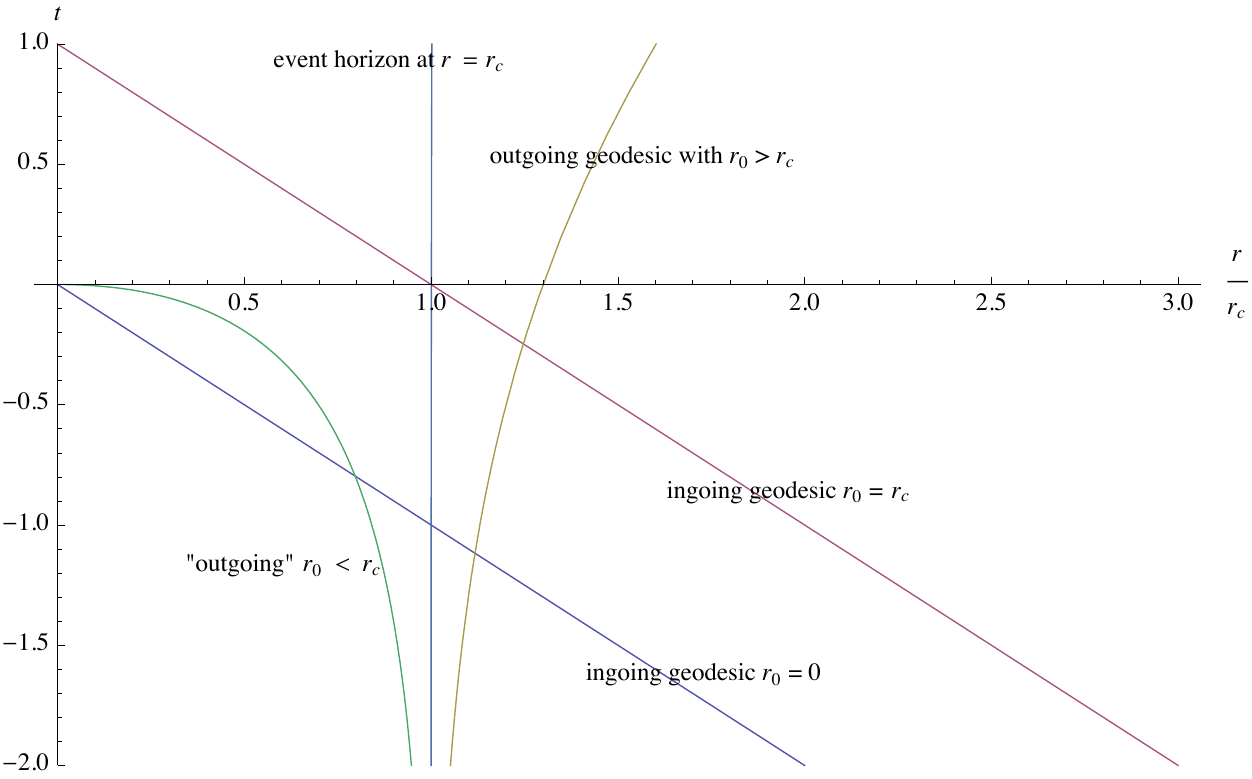}

    \caption[Effective null geodesics in the effective geometry of an ingoing cylindrical wave according to the ``asymptotic approximation"]{In this graph, effective null geodesics are plotted using Equations (\ref{inasymapproxt}) and (\ref{outasymapproxt}). Equations (\ref{inasymapproxt}) and (\ref{outasymapproxt}) are themselves based on the asymptotic approximations (\ref{inasym}) and (\ref{outasym}), which assumes that the quantity  $\omega r$ is large. }

    \label{figure4.1}
\end{figure}

In fact, using the numerical integration capabilities of \emph{Mathematica}, we can make spacetime diagrams for the effective null geodesics, as described by Equations (\ref{inMCave}) and (\ref{outMCave}) to second-order in $\alpha$, without recourse to the asymptotic approximations (\ref{inasym}) and (\ref{outasym}). These diagrams, shown in Figures \ref{figure42} - \ref{figure44},   are similar to  Figures \ref{figure40} - \ref{figure4.1}.  Again  we find that the effective geometry contains a black hole. That is,   the effective light cones are tilted towards the origin, and there is an effective event horizon. In making the plots for Figures \ref{figure42} - \ref{figure44}, 
we have set $\omega = 1$ and we have chosen $A^2$ such that $88\alpha^2 A^2/(45\pi\omega)=1$.

\begin{figure}[h]
\center
    \includegraphics[height=4in]{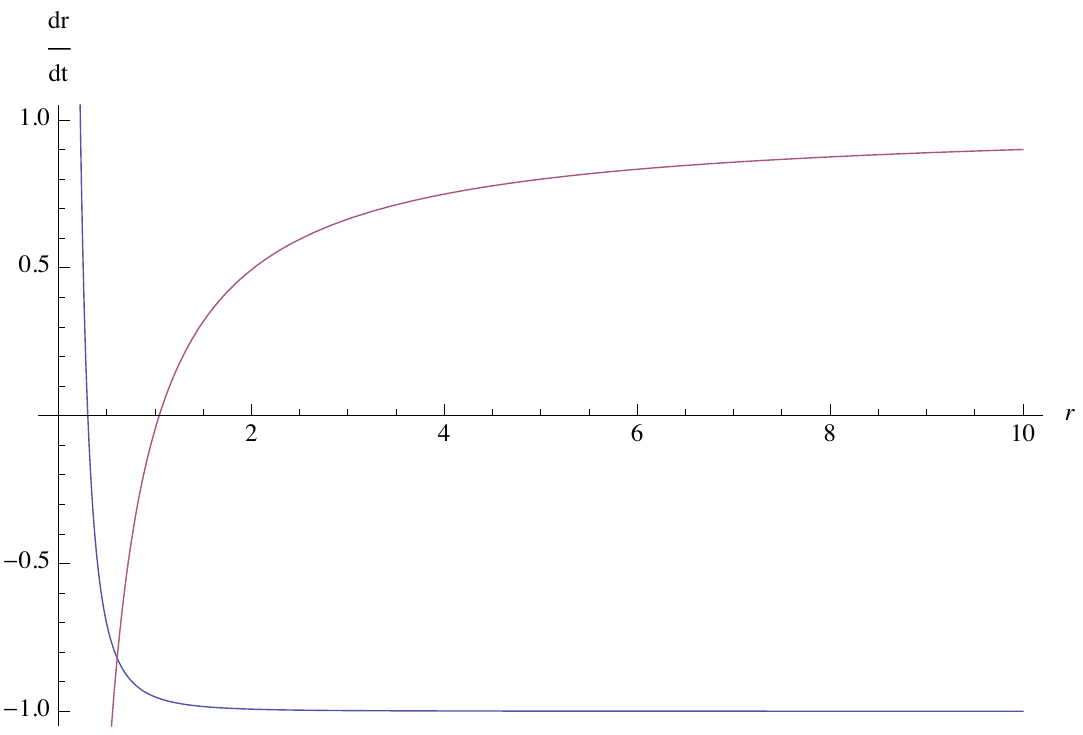}

    \caption[Coordinate velocities of  null geodesics in the effective geometry of an ingoing cylindrical wave according to the ``Maxwellian approximation"]{In this graph, the coordinate velocities of  effective null geodesics are plotted using Equations   (\ref{inMCave}) and (\ref{outMCave}) up to to second-order in $\alpha$. In plotting this graph, we have set $\omega = 1$ and we have chosen $A^2$ so that the quantity $88\alpha^2 A^2/(45\pi\omega)$ (our crude estimate for the effective horizon radius) is unity.  The  blue curve corresponds to the ingoing geodesics, and the red curve corresponds to the ``outgoing" radial geodesics. Compare to Figure \ref{figure40}. }

    \label{figure42}
\end{figure}

\begin{figure}[p]
\center
    \includegraphics[height=6.5in]{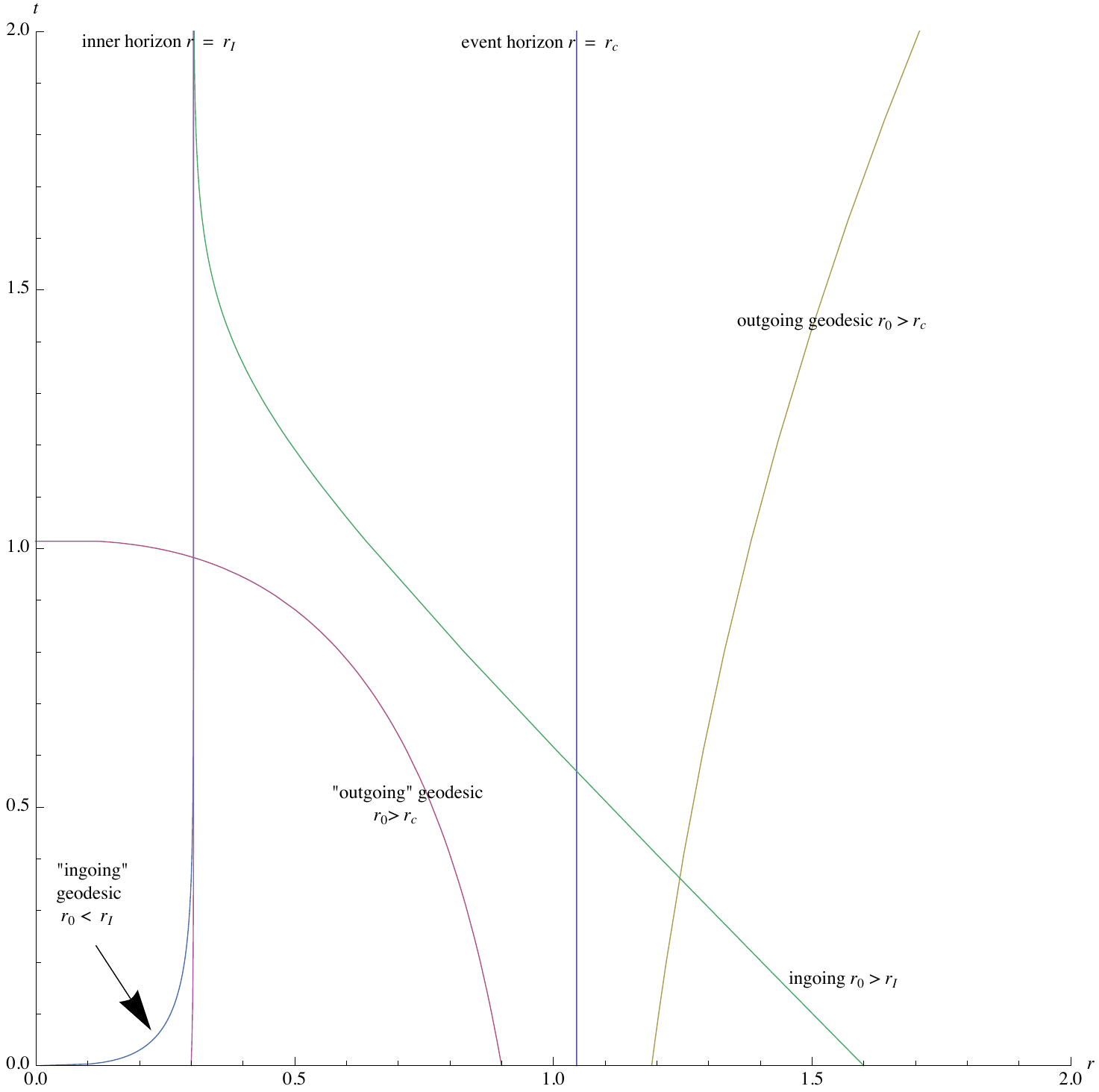}

    \caption[Null geodesics in the effective geometry of an ingoing cylindrical wave obtained by numerically integrating the ``Mawellian approximation"]{In this graph, effective null geodesics are plotted by numerically integrating Equations (\ref{inMCave}) and (\ref{outMCave}). We have set $\omega=1$ and we have chosen $A^2$ so that the quantity $88\alpha^2 A^2/(45\pi \omega)$ is unity. Compare to Figure \ref{figure4.1}. }

    \label{figure43}
\end{figure}

As we see from these plots, the radially outgoing rays are significantly slowed down near the critical radius. This means that if we slowly move a clock radially inwards, then an observer at infinity would see it ticking at a progressively slower   rate. Actually, due to birefringence, the situation is even more complicated since there will also be double images, but we are ignoring birefringent effects for now. When the clock reaches the critical radius, its light rays will not travel beyond the critical radius, and the clock will no longer be visible from the outside.

We note that, according to the second-order approximations   (\ref{inMCave}) and (\ref{outMCave}),  there is a small radius within the event horizon where the ingoing geodesics are brought to zero coordinate velocity. Thereby the ingoing geodesics coming in from infinity do not penetrate all the way to the origin, but instead are blocked by an ``inner horizon" at $r=r_I$ (see Figure (\ref{figure43})).   Also, we note that null geodesics exceed the usual speed of light, as viewed in the background coordinates.\index{superluminal photons} Some of the phenomena shown in Figure \ref{figure43}, especially at the smaller radii, may be mere artifacts of the  approximation. We note however that the superluminal photons, if such exist, will not  violate causality if the effective spacetime which they propagate is a causal spacetime. 

Note also that there is a radius between the inner and outer horizons where the ``ingoing" and ``outgoing" geodesics cannot be locally distinguished. At this special radius, the ingoing and outgoing radial geodesics travel in the same direction at the same velocity, so the effective light cone is degenerate at this radius.

We remark that if one were to take $\omega$ as negative, which amounts to turning  our ingoing wave into an \emph{outgoing} cylindrical wave, then one would obtain an effective spacetime which contains an optical white hole  rather than a black hole.\index{white holes}

\section{A static exact solution }\label{static}

In  the case of a field with $z$-polarization, $u\equiv 0$, the  field equations (\ref{system4}) reduce to a single nonlinear PDE:
\begin{eqnarray}\label{zPDE}
\partial_r^2v - \partial_t^2v+\frac{\partial_r v}{r} &=&\frac{8\alpha^2}{45r}\left[(\partial_rv)^3 - (\partial_t v) (\partial_r v) (\partial_t v + 4r \partial_t \partial_r v)\right.\nonumber\\
 &&\left.+r(\partial_t v)^2(3\partial_t^2v - \partial_r^2v)+r(\partial_r v)^2(3\partial^2_{r}v - \partial_t^2v)\right].
\end{eqnarray} A particularly interesting   case in which Equation (\ref{zPDE}) can be solved exactly is that of a field where $\partial_t v \equiv E=$ constant.  In this case (\ref{zPDE}) reads:
\begin{eqnarray}\label{zODE1}
\frac{d^2 v}{dr^2} +\frac{1}{r}\frac{dv}{dr} &=&\frac{8\alpha^2}{45r}\left[\left(\frac{dv}{dr}\right)^3 -E^2\left(\frac{dv}{dr}\right) - r E^2 \left(\frac{d^2 v}{dr^2}\right) +3r\left(\frac{dv}{dr}\right)^2 \left(\frac{d^2v}{dr^2}\right)\right].\nonumber\\
\end{eqnarray}Introducing the function $B(r):=dv/dr$, Equation (\ref{zODE1}) becomes:
\begin{eqnarray}
\frac{dB}{dr} + \frac{1}{r}B  &=& \frac{8\alpha^2}{45r}\left[B^3 -E^2 B - r E^2 \left(\frac{dB}{dr}\right)+ 3rB^2 \left(\frac{dB}{dr}\right)\right].
\end{eqnarray}This can be rearranged into:
\begin{eqnarray}\label{easy}
\frac{dB}{dr}=-\left(\frac{8\alpha^2 B^3  -  8\alpha^2 E^2 B - 45B}{r(24\alpha^2 B^2 - 8\alpha^2 E^2 -45)}\right).
\end{eqnarray} Equation (\ref{easy}) is   solvable  by integration. The general solution is given through the relation:
\begin{eqnarray}\label{bIMP}
B+\frac{8\alpha^2}{45}\left(E^2 B - B^3\right)= \frac{k}{r},
\end{eqnarray} where $k$ is an arbitrary constant. A sketch of the graph of Equation (\ref{bIMP}) is shown in Figure \ref{figure44}. There are three asymptotic values of $B$ as $r\rightarrow \infty$, namely: $B=0$, and $B=\pm\sqrt{E^2 + 45/(8\alpha^2)} $.

Note that $dB/dr$ has a singularity when $r=0$ and when $B=\pm\sqrt{E^2/3 + 15/(8\alpha^2)}$. Let us define $B_s:=\sqrt{E^2/3 + 15/(8\alpha^2)}$, and let $r_s$ denote the radius where $B^2 = B_s^2$.

We remark that (\ref{bIMP})  corresponds physically to a constant electric field $E$ directed along the $\pm z$-direction ($\pm$ depending on where $E$ is positive or negative, respectively), together with a magnetic field $B$, which  in the Maxwellian limit $\alpha^2\rightarrow 0$,  would be    produced by a constant current $2\pi k$ running  along the $\pm z$-direction ($\pm$ depending on whether $k$ is positive or negative, respectively). The magnetic field circulates around counterclockwise around the the $z$-axis if $B$ is positive, clockwise if negative. In other words, Equation (\ref{bIMP})  refines   the familiar undergraduate physics formula  $B=k/r$. 

\begin{figure}[h]
\center
    \includegraphics[height=3.4in]{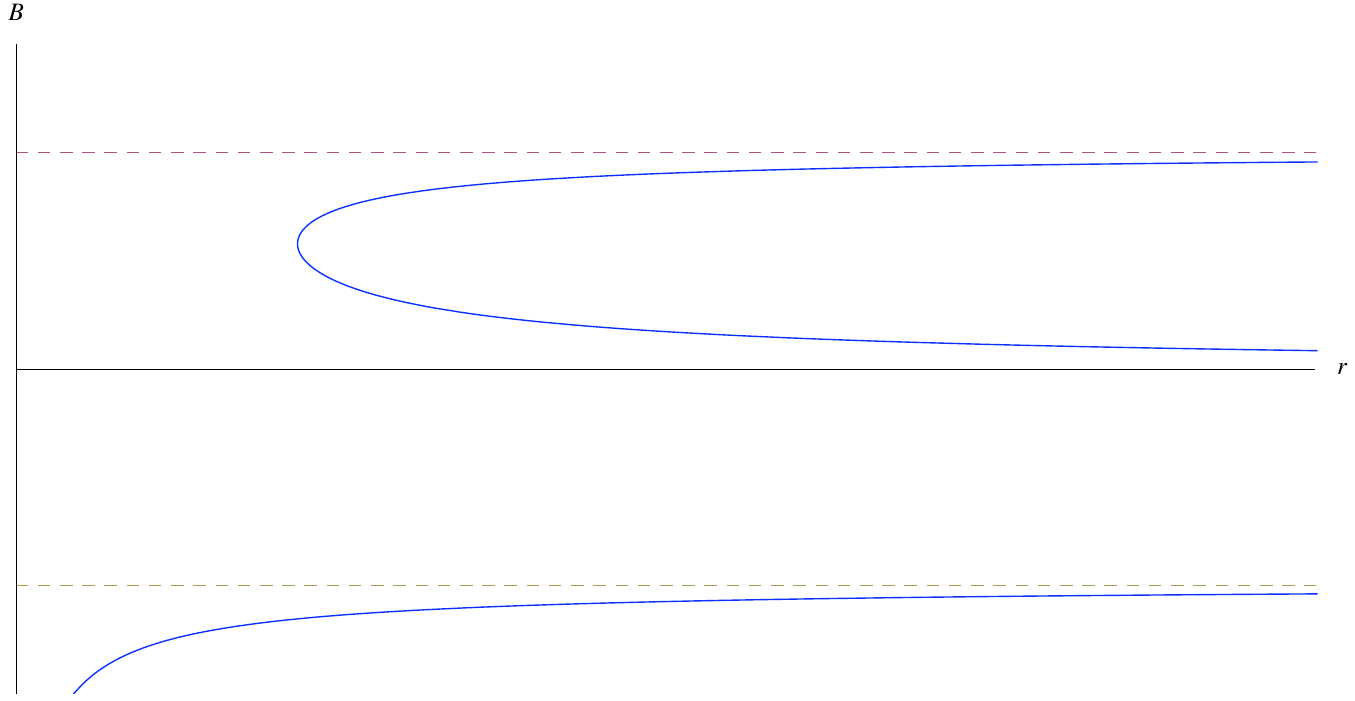}

    \caption[Cylindrically symmetric static magnetic field]{A sketch of the graph of Equation (\ref{bIMP}) for $k>0$. }

    \label{figure44}
\end{figure}

Using Equations (\ref{F1}) and (\ref{G1}), we get that the $F$ and $G$ invariants are:
\begin{eqnarray}
F&=& 2\left(B^2 - E^2 \right),
\end{eqnarray}and:
\begin{eqnarray}
G&=& 0.
\end{eqnarray} So Equation (\ref{EHlambda}) gives:
\begin{eqnarray}\label{4.73}
\frac{1}{\Lambda} & = & \frac{495 + 24\alpha^2(B^2 - E^2)\mp \Big|135-136\alpha^2(B^2 - E^2)\Big|}{224\alpha^2} ,\nonumber\\
\end{eqnarray}where the choice of $\pm$ depends on the polarization state of the field disturbance.  We shall call the polarization  corresponding to choosing $+$ in Equation (\ref{4.73}) the ``$(+)$ polarization state," and we call the other state the ``$(-)$ polarization state." (Note that our naming schemes for the polarization states are always \emph{ad hoc} and the naming scheme in the preset section is not meant to be consistent with  e.g., Section \ref{egonf} or Appendix \ref{Appendix:Key1}.)

Using Equations (\ref{emetric4.28}) - (\ref{emetric4.33}), we get that the only nonzero components of the effective metric are, up to a conformal factor:
\begin{eqnarray}\label{4.74}
 \widetilde{g}_{tt}&=&1 - \Lambda B^2 \\
\widetilde{g}_{tr}=\widetilde{g}_{rt}&=&-  \Lambda E B\\
\widetilde{g}_{rr}&=&-1-\Lambda E^2\\
\widetilde{g}_{\theta\theta}&=&-r^2 +\Lambda r^2 \left(B^2 - E^2\right)\\
\widetilde{g}_{zz}&=&-1.\label{4.78}
\end{eqnarray}

Using  Equation (\ref{radial2}), we get that the radial null geodesics are given by:
\begin{eqnarray}\label{4.79}
\frac{dr}{dt} &=& \frac{- E B \mp \sqrt{\frac{1}{\Lambda^2}-\frac{1}{\Lambda}\left(B^2 - E^2\right)}}{E^2+\frac{1}{\Lambda}}.
\end{eqnarray} The outgoing radial geodesics correspond to choosing $+$ in Equation (\ref{4.79}), and the ingoing geodesics correspond to choosing $-$.

\subsection{Proof of the main theorem}\label{proofmaintheorem}
 
  Since Equation (\ref{bIMP}) gives a cubic equation in $B$, there are three branches giving  $B$  as a real function of $r$ (see  the Table below).
 
 \begin{table}[h]\label{branchnames}
 \begin{center}
 \begin{tabular}{|c|c|}
 \hline 
branch  & range of $B^2$\\
\hline
I & 0 to $B_s^2$\\
II & $B_s^2$ to $3B_s^2$\\
III & $B_s^2$  to $\infty$ \\
 \hline
 \end{tabular} 
 \end{center}\caption[The three branches of $B$]{The three branches  of $B$ classified according to their ranges. }
 \end{table}

Let us consider the case where  the field is  prescribed by  branch I, with  $k>0$. This field is defined only for $r\geq r_s$. Note that  Equation (\ref{easy}) can be rewritten as:
\begin{eqnarray}\label{as pi (or maybe it's cake?)}
\frac{dB}{dr}&=&\frac{B\left(3B_s^2 - B^2\right)}{3r\left(B^2 - B_s^2\right)}.
\end{eqnarray} Thus we have that $B$ is a strictly decreasing function of $r$. As shown below, both   effective geometries  for this field contain black holes,  if $ \sqrt{\frac{45}{34\alpha^2}} < 
E  <   \sqrt{\frac{9}{4\alpha^2}}$. This section constitutes proof of Theorem \ref{mainthm} from Chapter \ref{chapter0}.

\begin{thm}For $E^2\geq \frac{45}{34\alpha^2 }$, the effective geometry corresponding to the $(+)$ polarization state contains a black hole (if $E>0$), or a white hole (if $E<0$), with the event horizon at $r=r_s$. \index{black holes}\index{white holes}
\end{thm}
\begin{proof} For the $(+)$ polarization state, with $E^2 \geq 45/(34\alpha^2)$, Equation (\ref{4.73}) gives:
\begin{eqnarray}
\frac{1}{\Lambda} & = & \frac{45}{16\alpha^2} - \frac{1}{2}\left(B^2 - E^2\right).
\end{eqnarray}
At $r=r_s$:
\begin{eqnarray}
\frac{1}{\Lambda}\Big|_{r=r_s} &=&B_s^2.
\end{eqnarray}
In fact, since $1/\Lambda \geq1/\Lambda|_{r=r_s},$ we have that $\Lambda > 0$ for all $r\geq r_s$. 

Equation (\ref{4.79}) yields:
\begin{eqnarray}\label{e583}
\frac{dr}{dt}\Big|_{\textrm{out},\ r\geq r_s} &=&\frac{-EB + \sqrt{\frac{1}{\Lambda^2} - \frac{1}{\Lambda}\left(B^2 - E^2\right)}}{E^2 + \frac{1}{\Lambda}}\nonumber\\
 &=&\frac{-EB + \sqrt{\left( \frac{45}{16\alpha^2} - \frac{1}{2}\left(B^2 - E^2\right)\right)^2 -\left( \frac{45}{16\alpha^2} - \frac{1}{2}\left(B^2 - E^2\right)\right)\left(B^2 - E^2\right)}}{E^2 + \frac{45}{16\alpha^2} - \frac{1}{2}\left(B^2 - E^2\right)}.\nonumber\\
\end{eqnarray} For $r>r_s$, the numerator in Equation (\ref{e583}) is positive by virtue of the fact that $B^2 < B_s^2 $ in the region $r>r_s$. Since $\Lambda > 0$, the denominator is also positive. So $dr/dt|_{\textrm{out},\ r > r_s}$ is positive in the region $r > r_s$.

For ingoing radial null geodesics, we have: 
\begin{eqnarray}\label{e584}
\frac{dr}{dt}\Big|_{\textrm{in},\ r\geq r_s} &=&\frac{-EB - \sqrt{\left( \frac{45}{16\alpha^2} - \frac{1}{2}\left(B^2 - E^2\right)\right)^2 -\left( \frac{45}{16\alpha^2} - \frac{1}{2}\left(B^2 - E^2\right)\right)\left(B^2 - E^2\right)}}{E^2 + \frac{45}{16\alpha^2} - \frac{1}{2}\left(B^2 - E^2\right)},\nonumber\\
\end{eqnarray} and $dr/dt|_{\textrm{in},\ r > r_s}$ is negative   since  $B^2 < B_s^2 $ in the region $r>r_s$.

At $r=r_s$, we get:
\begin{eqnarray}
\frac{dr}{dt}\Big|_\textrm{out, $r=r_s$}&=& \frac{B_s\left(|E| - E\right)}{E^2 + B_s^2},
\end{eqnarray}and:

\begin{eqnarray}
\frac{dr}{dt}\Big|_{\textrm{in, $r=r_s$}} = -\frac{B_s\left(|E| + E\right)}{E^2 + B_s^2}.
\end{eqnarray}
If $E$ is positive, then   $dr/dt\Big|_\textrm{out, $r=r_s$}=0$;  radial outgoing null geodesics at $r=r_s$ are trapped.  On the other hand, if $E$ is negative, then $dr/dt\Big|_\textrm{in, $r=r_s$}=0$; radial ingoing geodesics cannot reach $r=r_s$ from $r>r_s$ (the outside).

It remains to be shown that the nonradial curves are trapped at $r=r_s$. To this end, 
 suppose that we have an arbitrary null curve in the effective spacetime. We write:
 \begin{eqnarray}
 0 = \widetilde{g}_{tt} + 2 \widetilde{g}_{tr}\frac{dr}{dt}+ \widetilde{g}_{rr} \left(\frac{dr}{dt}\right)^2  +  \widetilde{g}_{\theta\theta}\left(\frac{d\theta}{dt}\right)^2  + \widetilde{g}_{zz} \left(\frac{dz}{dt}\right)^2,
 \end{eqnarray} with $\widetilde{g}_{\mu\nu}$ given by Equations (\ref{4.74}) - (\ref{4.78}). 
 Then:
 \begin{eqnarray}
 \frac{dr}{dt} = \frac{-\widetilde{g}_{tr} \pm \sqrt{\widetilde{g}_{tr}^{\ 2}-\widetilde{g}_{tt}\widetilde{g}_{rr} - \widetilde{g}_{rr}\left( \widetilde{g}_{\theta\theta}\left(\frac{d\theta}{dt}\right)^2  +  \widetilde{g}_{zz} \left(\frac{dz}{dt}\right)^2 \right)}}{\widetilde{g}_{rr}}
  \end{eqnarray}
 
We claim that:
 \begin{eqnarray} \frac{-\widetilde{g}_{tr}+ \sqrt{\widetilde{g}_{tr}^{\ 2}-\widetilde{g}_{tt}\widetilde{g}_{rr} }}{\widetilde{g}_{rr}}
\leq \frac{dr}{dt} \leq  \frac{-\widetilde{g}_{tr}- \sqrt{\widetilde{g}_{tr}^{\ 2}-\widetilde{g}_{tt}\widetilde{g}_{rr} }}{\widetilde{g}_{rr}}. \end{eqnarray} That is, an arbitrary null curve cannot climb up to larger radii faster than a radially outward geodesic, and cannot fall down  to smaller radii faster than a radially inward geodesic. In other words,  nonradial null curves are trapped if the radial null geodesics are trapped. 

Note that in order to prove the claim, it suffices to show that $ \widetilde{g}_{rr} < 0$ and  $ \widetilde{g}_{\theta\theta} < 0$.

We get that $\widetilde{g}_{rr}\ (= -1 - \Lambda E^2) $ is negative since  $\Lambda>0$.

To get $\widetilde{g}_{\theta\theta} = -r^2 + \Lambda r^2 (B^2 -E^2 ) <0$, it suffices to show that $B^2 - E^2 < 1/\Lambda$. To this end, note that:
\begin{eqnarray}\label{589}
B^2  - E^2 &\leq & B_s^2 - E^2\nonumber\\
&<& \frac{15}{8\alpha^2}.
\end{eqnarray}Muliplying both sides of (\ref{589}) by $3/2$, we get:
\begin{eqnarray}
\frac{3}{2}\left(B^2 - E^2 \right) < \frac{45}{16\alpha^2},
\end{eqnarray} so:
\begin{eqnarray}
B^2 - E^2 &<&\frac{45}{16\alpha^2} - \frac{1}{2}\left(B^2 - E^2\right)\nonumber\\\nonumber\\
&=&\frac{1}{\Lambda}.
\end{eqnarray}

\end{proof}

\begin{thm}
For  $\frac{45}{34\alpha^2} <
E^2 < \frac{9}{4\alpha^2}$, the effective geometry for the $(-)$ polarization state contains a black hole (if $E>0$), or a white-hole (if $E<0$), with an effective event horizon at  $r=r_c$ such that $r >  r_s$. In fact,
$r_c =9k\sqrt{5}/\left(7\alpha E^2 \sqrt{18 - 8 \alpha^2 E^2}\right)$.\index{black holes}\index{white holes}
\end{thm}
\begin{proof}
For the $(-)$ polarization state, with $E^2 > 45/(34\alpha^2)$, Equation (\ref{4.73}) gives:
\begin{eqnarray}
\frac{1}{\Lambda} &=& \frac{45}{28\alpha^2} + \frac{5}{7}\left(B^2 - E^2\right).
\end{eqnarray}Since $E^2<9/(4\alpha^2)$,  we have that  $\Lambda> 0$.

At $r=r_s$:
\begin{eqnarray}\label{5.90}
\frac{1}{\Lambda}\Big|_{r=r_s} &=& \frac{495 - 80\alpha^2 E^2 }{168\alpha^2}.
\end{eqnarray}
Moreover, for outgoing radial null geodesics, Equation (\ref{4.79}) gives:
\begin{eqnarray}\label{5.88}
\frac{dr}{dt}\Big|_\textrm{out,  $r=r_s$}
&=&\left(\frac{-168\alpha^2 E  + 4\alpha\sqrt{30(99 - 16\alpha^2 E^2)}}{495  +  88 \alpha^2 E^2}\right)B_s,
\end{eqnarray} which is negative if  $E> +\sqrt{ 45/(34\alpha^2)}$. This means that the outgoing radial geodesics, initiated from $r=r_s$, are compelled to fall down to smaller radii.  At the other extreme ($r=\infty$), note that:
\begin{eqnarray}
\frac{dr}{dt}\Big|_{\textrm{out, } r=\infty} &=&  \frac{1}{\sqrt{1+\Lambda E^2}}>0.
\end{eqnarray}
Hence, by the intermediate-value-theorem, there is a radius $r_c$,  which is greater than $r_s$ and less than $\infty$, where $dr/dt\Big|_{\textrm{out, } r= r_c}=0$ (black hole event horizon at $r_c$, if $E>\sqrt{ 45/(34\alpha^2)}$).  

In fact, the critical radius $r_c$  is unique and we can calculate it. If we set the left hand side of Equation (\ref{4.79}) equal to 0, and solve for $B$, then we find that there is only one real-valued positive solution, namely:
\begin{eqnarray}
B_c &=& \sqrt{\frac{45 - 20 \alpha^2 E^2}{8\alpha^2}}.
\end{eqnarray} Thereby, using (\ref{bIMP}), we get:
\begin{eqnarray}\label{explicitrc}
r_c & = & \frac{k}{B_c + \frac{8\alpha^2}{45}\left(E^2B_c - B_c^3\right)}\nonumber\\\nonumber\\
&=&\frac{9k\sqrt{5}}{7\alpha E^2 \sqrt{18 - 8\alpha^2 E^2}}.
\end{eqnarray}Note that $dr/dt|_{\textrm{out} }$ changes sign at $r = r_c$ since $dr/dt|_{\textrm{out,} r = r_s}$ is negative and  $dr/dt|_{\textrm{out,} r = \infty}$ is positive. That is, any outgoing geodesic in the region $r<r_c$ is compelled to fall inward to smaller $r$; any outgoing geodesic in the region $r>r_c$ will escape to larger $r$.

Next we consider the ingoing null geodesics.  For these, Equation (\ref{4.79}) gives:
\begin{eqnarray}\label{5.96}
\frac{dr}{dt}\Big|_\textrm{in,  $r=r_s$}
&=&\left(\frac{-168\alpha^2 E  - 4\alpha\sqrt{30(99 - 16\alpha^2 E^2)}}{495  +  88 \alpha^2 E^2}\right)B_s,
\end{eqnarray} This is positive if  $E< -\sqrt{ 45/(34\alpha^2)}$. At $r=\infty$, we have:
\begin{eqnarray}
\frac{dr}{dt}\Big|_{\textrm{in, } r=\infty} &=&  -\frac{1}{\sqrt{1+\Lambda E^2}}<0.
\end{eqnarray} Thus, if $E$ is negative, there is a  radius $r_c$,  between $r_s$ and  $\infty$, where $dr/dt\Big|_{\textrm{in, } r= r_c}=0$ (white hole  event horizon at $r_c$).  The quantity $dr/dt\Big|_{\textrm{in}}$ changes sign at $r_c$ in such a way that ingoing geodesics issuing from the region $r<r_c$ will be compelled to escape outward to larger $r$, and ingoing geodesics issuing from $r>r_c$ will fall inward to smaller $r$.

As in the proof of the previous theorem, we get that the nonradial geodesics are  trapped by the event horizon by showing that  $\widetilde{g}_{\theta\theta}<0$ and $\widetilde{g}_{rr}<0$. The fact that $\Lambda>0$ gives $\widetilde{g}_{rr}<0$.

As before, to get $\widetilde{g}_{\theta\theta}<0$, it suffices to show that $B^2 - E^2 < 1/\Lambda$. To this end, note that:
\begin{eqnarray}\label{599}
B^2 - E^2 &\leq&B_s^2 -E^2\nonumber\\
&<&\frac{45}{8\alpha^2}.
\end{eqnarray}Multiplying both sides of (\ref{599}) by $2/7$, we get:
\begin{eqnarray}
\frac{2}{7}\left(B^2 - E^2\right)< \frac{45}{28\alpha^2},
\end{eqnarray} so:
\begin{eqnarray}
B^2 - E^2& < &\frac{45}{28\alpha^2}+\frac{5}{7}\left(B^2 - E^2\right)\nonumber\\
\nonumber\\
&=&\frac{1}{\Lambda}.
\end{eqnarray}

\end{proof}

\newcommand{\mathsym}[1]{{}}
\newcommand{\unicode}{{}}

\appendix
\chapter{\emph{Mathematica} notebook}
\label{Appendix:Key1}

The purpose of this appendix is to show how \emph{Mathematica} (version 8) can be used to check or carry out the calculations in \ref{type} and \ref{maxwellianapproximation}. Our  implementation of tensor calculus in  \emph{Mathematica} is modeled on applications found elsewhere, such  
M\"uller and Grave \cite{MuellerGrave} and the downloadable notebooks of  Parker \cite{lparker}.

Let us begin by clearing out the variables that will be used:

\ \\ \noindent\(\pmb{\text{Clear}[\text{coord},t,r,\theta ,z,i,j,k,l,u,v,\text{metric},}
\pmb{\text{cometric},\text{affine},\text{Afield},\text{faraday},\text{faraday1},\text{faraday2},}\\
\pmb{\text{Fspecialcase},\text{maxwell},\text{Gspecialcase},}
\pmb{\text{maxwell1},\text{maxwell2},\text{substress},\text{ecometric},\Lambda, \text{emetric,}}\\
\pmb{\text{simplifiedemetric},\alpha , \text{lambdaplus}, }
\pmb{\ \text{lambdaminus},\text{radA1},\text{radA2},\text{radB1},\text{radB2},\text{CDfaraday2,}}\\
\pmb{s,o,\text{Ffaraday2},\text{CDFfaraday2},}
\pmb{\ \text{Gmaxwell2},\text{CDGmaxwell2},\text{righthandside},U,V,} 
\pmb{\ \omega ,A]}\)
\
\\

Next, we  specify the coordinate system (cylindrical coordinates) and the background metric (Minkowski spacetime):

\ \\ \noindent\(\pmb{\text{coord}\text{:=}\text{coord}=\{t,r,\theta ,z\}}\)
\pmb{}\\

\ \\ \noindent\(\pmb{\text{metric}\text{:=}}
\pmb{\text{metric}=\left\{\{1,0,0,0\},\{0,-1,0,0\},\left\{0,0,-r^2,0\right\},\right.}
\pmb{\{0,0,0,-1\}\}}\)
\
\\

The background cometric is computed by entering:

\ \\ \noindent\(\pmb{\text{cometric}\text{:=}\text{cometric}=\text{Inverse}[\text{metric}]}\)
\
\\

The Christoffel symbols  (for the background metric) are calculated by entering:

\ \\ \noindent\(\pmb{\text{affine}\text{:=}}
\pmb{\text{affine}=\text{FullSimplify}[}\\
\pmb{\text{Table}[}
\pmb{\frac{1}{2}\text{Sum}[(\text{cometric}[[k,l]])*}
\pmb{(D[\text{metric}[[i,l]],\text{coord}[[j]]]+ }\\
\pmb{D[\text{metric}[[j,l]],\text{coord}[[i]]]-}
\pmb{D[\text{metric}[[i,j]],\text{coord}[[l]]]),}\\
\pmb{\{l,1,4\}],\{k,1,4\},\{i,1,4\},\{j,1,4\}]]}\)

\section{The field tensors}

Our first task is to check the results given in Section \ref{5.1.1}. To this end, we input the $A$-field that we wish to study (coinciding with
Equations (\ref{at}) - (\ref{az})):

\ \\ \noindent\(\pmb{\text{Afield}\text{:=}\text{Afield}=\{0,0,u[t,r],v[t,r]\}}\)
\
\\

Next we computer the field tensor \(F_{\mu \nu }\), called {``}faraday:{''}

\ \\ \noindent\(\pmb{\text{faraday}\text{:=}}
\pmb{\text{faraday}=}\\
\pmb{\text{Table}[D[\text{Afield}[[j]],\text{coord}[[i]]]-}
\pmb{D[\text{Afield}[[i]],\text{coord}[[j]]],\{i,1,4\},\{j,1,4\}]}\)
\
\\

The components of \(F_{\mu \nu }\) are displayed by entering:

\ \\ \noindent\(\pmb{\text{faraday}\text{//}\text{MatrixForm}}\)

\begin{eqnarray}\left(
\begin{array}{cccc}
 0 & 0 & u^{(1,0)}[t,r] & v^{(1,0)}[t,r] \\
 0 & 0 & u^{(0,1)}[t,r] & v^{(0,1)}[t,r] \\
 -u^{(1,0)}[t,r] & -u^{(0,1)}[t,r] & 0 & 0 \\
 -v^{(1,0)}[t,r] & -v^{(0,1)}[t,r] & 0 & 0
\end{array}\right)
\end{eqnarray}
In \emph{Mathematica},  $u^{(1,0)}[t,r] $ denotes $\partial_t u$, and $u^{(0,1)}[t,r]$ denotes $\partial_r u$, etc. This output agrees with Equation (\ref{farad}).

Raising the first index, we get   \(F^{\mu }{}_{\nu }\). Call this ``faraday1:"

\ \\ \noindent\(\pmb{\text{faraday1}\text{:=}}
\pmb{\text{faraday1}=}
\pmb{\text{FullSimplify}[}\\
\pmb{\text{Table}[\text{Sum}[\text{cometric}[[i,k]]\text{faraday}[[k,j]],}
\pmb{\{k,1,4\}],\{i,1,4\},\{j,1,4\}]]}\)
\
\\

Displaying the components of  \(F^{\mu }{}_{\nu }\) in matrix form, as in Equation (\ref{farad1}):

\ \\ \noindent\(\pmb{\text{faraday1}\text{//}\text{MatrixForm}}\)

\begin{eqnarray}\left(
\begin{array}{cccc}
 0 & 0 & u^{(1,0)}[t,r] & v^{(1,0)}[t,r] \\
 0 & 0 & -u^{(0,1)}[t,r] & -v^{(0,1)}[t,r] \\
 \frac{u^{(1,0)}[t,r]}{r^2} & \frac{u^{(0,1)}[t,r]}{r^2} & 0 & 0 \\
 v^{(1,0)}[t,r] & v^{(0,1)}[t,r] & 0 & 0
\end{array}\right)
\end{eqnarray}

Raising the second index, we get \(F^{\mu \nu }\) (``faraday2"):

\ \\ \noindent\(\pmb{\text{faraday2}\text{:=}}
\pmb{\text{faraday2}=}
\pmb{\text{FullSimplify}[}\\
\pmb{\text{Table}[\text{Sum}[\text{cometric}[[k,j]]\text{faraday1}[[i,k]],}
\pmb{\{k,1,4\}],\{i,1,4\},\{j,1,4\}]]}\)
\
\\

As in Equation (\ref{farad2}), we have:

\ \\ \noindent\(\pmb{\text{faraday2}\text{//}\text{MatrixForm}}\)

\begin{eqnarray}\left(\begin{array}{cccc}
 0 & 0 & -\frac{u^{(1,0)}[t,r]}{r^2} & -v^{(1,0)}[t,r] \\
 0 & 0 & \frac{u^{(0,1)}[t,r]}{r^2} & v^{(0,1)}[t,r] \\
 \frac{u^{(1,0)}[t,r]}{r^2} & -\frac{u^{(0,1)}[t,r]}{r^2} & 0 & 0 \\
 v^{(1,0)}[t,r] & -v^{(0,1)}[t,r] & 0 & 0
\end{array}\right)
\end{eqnarray}

We get that the  $F$-invariant is, in agreement with (\ref{F1}):

\ \\ \noindent\(\pmb{\text{Fspecialcase}=}
\pmb{\text{Simplify}[\text{Sum}[\text{faraday}[[i,j]]\text{faraday2}[[i,j]],}
\pmb{\{i,1,4\},\{j,1,4\}]]}\)

\begin{eqnarray}2 \left(\frac{u^{(0,1)}[t,r]^2}{r^2}+v^{(0,1)}[t,r]^2-\right.
\left.\frac{u^{(1,0)}[t,r]^2}{r^2}-v^{(1,0)}[t,r]^2\right)\end{eqnarray}

For  the dual tensor \(F^*{}_{\mu \nu }\) (``maxwell'':)

\ \\ \noindent\(\pmb{\text{maxwell}\text{:=}}
\pmb{\text{maxwell}=\text{FullSimplify}[}
\pmb{\text{Table}\left[\frac{1}{2}\text{Sqrt}[-\text{Det}[\text{metric}]]\right.}\\
\pmb{\text{Sum}[\text{Signature}[\{i,j,k,l\}]\text{faraday2}[[i,j]],}
\pmb{\{i,1,4\},\{j,1,4\}],\{k,1,4\},\{l,1,4\}],}
\pmb{\ r\geq 0]}\)
\
\\ 

Displaying \(F^*{}_{\mu \nu }\)  as a matrix, as in Equation (\ref{max}):

\ \\ \noindent\(\pmb{\text{maxwell}\text{//}\text{MatrixForm}}\)

\begin{eqnarray}\left(
\begin{array}{cccc}
 0 & 0 & -r v^{(0,1)}[t,r] & \frac{u^{(0,1)}[t,r]}{r} \\
 0 & 0 & -r v^{(1,0)}[t,r] & \frac{u^{(1,0)}[t,r]}{r} \\
 r v^{(0,1)}[t,r] & r v^{(1,0)}[t,r] & 0 & 0 \\
 -\frac{u^{(0,1)}[t,r]}{r} & -\frac{u^{(1,0)}[t,r]}{r} & 0 & 0
\end{array}
\right)
\end{eqnarray}

The $G$-invariant is, in agreement with Equation (\ref{G1}):

\ \\ \noindent\(\pmb{\text{Gspecialcase}=\text{Sum}[\text{maxwell}[[i,j]]\text{faraday2}[[i,j]],}
\pmb{\{i,1,4\},\{j,1,4\}]}\)

\begin{eqnarray}\frac{4 v^{(0,1)}[t,r] u^{(1,0)}[t,r]}{r}-\frac{4 u^{(0,1)}[t,r] v^{(1,0)}[t,r]}{r}
\end{eqnarray}

To get \(F^{*\mu }{}_{\nu }\), enter:

\ \\ \noindent\(\pmb{\text{maxwell1}\text{:=}}
\pmb{\text{maxwell1}=}
\pmb{\text{FullSimplify}[}\\
\pmb{\text{Table}[\text{Sum}[\text{cometric}[[i,k]]\text{maxwell}[[k,j]],}
\pmb{\{k,1,4\}],\{i,1,4\},\{j,1,4\}]]}\)
\
\\

For  \(F^{*\mu \nu }\), enter:

\ \\ \noindent\(\pmb{\text{maxwell2}\text{:=}}
\pmb{\text{maxwell2}=}
\pmb{\text{FullSimplify}[}\\
\pmb{\text{Table}[\text{Sum}[\text{cometric}[[k,j]]\text{maxwell1}[[i,k]],}
\pmb{\{k,1,4\}],\{i,1,4\},\{j,1,4\}]]}\)
\
\\

The components of  \(F^{*\mu \nu }\), as in (\ref{max2}), are:

\ \\ \noindent\(\pmb{\text{maxwell2}\text{//}\text{MatrixForm}}\)

\begin{eqnarray}\left(
\begin{array}{cccc}
 0 & 0 & \frac{v^{(0,1)}[t,r]}{r} & -\frac{u^{(0,1)}[t,r]}{r} \\
 0 & 0 & -\frac{v^{(1,0)}[t,r]}{r} & \frac{u^{(1,0)}[t,r]}{r} \\
 -\frac{v^{(0,1)}[t,r]}{r} & \frac{v^{(1,0)}[t,r]}{r} & 0 & 0 \\
 \frac{u^{(0,1)}[t,r]}{r} & -\frac{u^{(1,0)}[t,r]}{r} & 0 & 0
\end{array}\right)
\end{eqnarray}
\section{The effective metric coefficients}
Now let us check Section \ref{5.1.2}. To this end, we calculate the substress tensor \(F^{\mu }{}_{\lambda }F^{\lambda \nu }\):

\ \\ \noindent\(\pmb{\text{substress}\text{:=}}
\pmb{\text{substress}=}
\pmb{\text{FullSimplify}[}\\
\pmb{\text{Table}[\text{Sum}[\text{faraday1}[[i,k]]\text{faraday2}[[k,j]],}
\pmb{\{k,1,4\}],\{i,1,4\},\{j,1,4\}]]}\)
\
\\

The nonzero components of the substress are displayed upon entering the lines:

\ \\ \noindent\(\pmb{\text{Do}[\text{If}[\text{UnsameQ}[\text{substress}[[i,j]],0],}\\
\pmb{\text{CellPrint}[ }\\
\pmb{\text{DisplayForm}[}\\
\pmb{\text{RowBox}[}\\
\pmb{\{\text{SubscriptBox}[\text{SuperscriptBox}[\text{``F''},\text{coord}[[i]]],}\\
\pmb{\text{{``}$\lambda ${''}}],\text{SuperscriptBox}[\text{``F''},}\\
\pmb{\text{RowBox}[\{\text{{``}$\lambda ${''}},\text{coord}[[j]]\}]],\text{{``}={''}},}\\
\pmb{\text{substress}[[i,j]]\}]]}\\
\pmb{]\text{  }],\{i,1,4\},\{j,1,4\}]}\)

\pmb{}\\

\(F^t{}_{\lambda }F^{\lambda t}=\frac{u^{(1,0)}[t,r]^2}{r^2}+v^{(1,0)}[t,r]^2\)

\(F^t{}_{\lambda }F^{\lambda r}=-\frac{u^{(0,1)}[t,r] u^{(1,0)}[t,r]}{r^2}-v^{(0,1)}[t,r] v^{(1,0)}[t,r]\)

\(F^r{}_{\lambda }F^{\lambda t}=-\frac{u^{(0,1)}[t,r] u^{(1,0)}[t,r]}{r^2}-v^{(0,1)}[t,r] v^{(1,0)}[t,r]\)

\(F^r{}_{\lambda }F^{\lambda r}=\frac{u^{(0,1)}[t,r]^2}{r^2}+v^{(0,1)}[t,r]^2\)

\(F^{\theta }{}_{\lambda }F^{\lambda \theta }=\frac{u^{(0,1)}[t,r]^2-u^{(1,0)}[t,r]^2}{r^4}\)

\(F^{\theta }{}_{\lambda }F^{\lambda z}=\frac{u^{(0,1)}[t,r] v^{(0,1)}[t,r]-u^{(1,0)}[t,r] v^{(1,0)}[t,r]}{r^2}\)

\(F^z{}_{\lambda }F^{\lambda \theta }=\frac{u^{(0,1)}[t,r] v^{(0,1)}[t,r]-u^{(1,0)}[t,r] v^{(1,0)}[t,r]}{r^2}\)

\(F^z{}_{\lambda }F^{\lambda z}=v^{(0,1)}[t,r]^2-v^{(1,0)}[t,r]^2\)
\
\\

The output above agrees with Equations (\ref{substressTT}) - (\ref{substressZZ}).

To calculate the effective cometric, enter:

\ \\ \noindent\(\pmb{\text{ecometric}\text{:=}\text{ecometric}=\text{cometric}+\Lambda\  \text{substress}}\)
\
\\

Here, $\Lambda $ stands for \(\Lambda _{\pm }\) in Equation (\ref{EHlambda}). The effective metric is calculated by entering:

\ \\ \noindent\(\pmb{\text{emetric}\text{:=}\text{emetric}=\text{Inverse}[\text{ecometric}]}\)
\
\\

Let us multiply the effective metric by a certain conformal factor; this choice of conformal factor considerably simplifies the effective metric coefficients:

\ \\ \noindent\(\pmb{\text{simplifiedemetric}\text{:=}}
\pmb{\text{simplifiedemetric}=}\\
\pmb{\text{FullSimplify}[}
\pmb{-\left(\left(\left(-1+\Lambda\  v^{(0,1)}[t,r]^2\right) \left(r^2+\Lambda\  u^{(1,0)}[t,r]^2\right)-\right.\right.}\\
\pmb{2 \Lambda^2 u^{(0,1)}[t,r] v^{(0,1)}[t,r] u^{(1,0)}[t,r] }
\pmb{v^{(1,0)}[t,r]-}\\
\pmb{\left.\left.r^2 \Lambda\  v^{(1,0)}[t,r]^2+\Lambda\  u^{(0,1)}[t,r]^2 \left(1+\Lambda\  v^{(1,0)}[t,r]^2\right)\right)/r^2\right)}
\pmb{\text{emetric}]}\)
\
\\

The nonzero components of the conformally rescaled effective metric are displayed upon entering the lines:

\ \\ \noindent\(\pmb{\text{Do}[\text{If}[\text{UnsameQ}[\text{simplifiedemetric}[[i,j]],0],}\\
\pmb{\text{CellPrint}[ }
\pmb{\text{DisplayForm}[}\\
\pmb{\text{RowBox}[\{\text{coord}[[i]],\text{coord}[[j]],\text{{``}-comp{''}},}\\
\pmb{\text{{``}={''}},\text{simplifiedemetric}[[i,j]]}
\pmb{\}]]}\\
\pmb{]\text{  }],\{i,1,4\},\{j,1,4\}]}\)
\
\\

\(tt-\text{comp}=1-\frac{\Lambda  u^{(0,1)}[t,r]^2}{r^2}-\Lambda  v^{(0,1)}[t,r]^2\)

\(tr-\text{comp}=-\frac{\Lambda  u^{(0,1)}[t,r] u^{(1,0)}[t,r]}{r^2}-\Lambda  v^{(0,1)}[t,r] v^{(1,0)}[t,r]\)

\(rt-\text{comp}=-\frac{\Lambda  u^{(0,1)}[t,r] u^{(1,0)}[t,r]}{r^2}-\Lambda  v^{(0,1)}[t,r] v^{(1,0)}[t,r]\)

\(rr-\text{comp}=-1-\frac{\Lambda  u^{(1,0)}[t,r]^2}{r^2}-\Lambda  v^{(1,0)}[t,r]^2\)

\(\theta \theta -\text{comp}=r^2 \left(-1+\Lambda  v^{(0,1)}[t,r]^2-\Lambda  v^{(1,0)}[t,r]^2\right)\)

\(\theta z-\text{comp}=-\Lambda  u^{(0,1)}[t,r] v^{(0,1)}[t,r]+\Lambda  u^{(1,0)}[t,r] v^{(1,0)}[t,r]\)

\(z\theta -\text{comp}=-\Lambda  u^{(0,1)}[t,r] v^{(0,1)}[t,r]+\Lambda  u^{(1,0)}[t,r] v^{(1,0)}[t,r]\)

\(zz-\text{comp}=-\frac{r^2-\Lambda  u^{(0,1)}[t,r]^2+\Lambda  u^{(1,0)}[t,r]^2}{r^2}\)
\
\\

This output agrees with Equations (\ref{emetric4.28}) - (\ref{emetric4.33}).
\section{The radial null geodesics}
Now we check Equations (\ref{in2o}) and (\ref{out2o}) from Section \ref{5.1.3}, which concerns  radial null geodesics in the effective geometry. To this end, we will need to let \emph{Mathematica} compute the values \(\Lambda _{\pm
}\) from Equation  (\ref{EHlambda}). The value  \(\Lambda _{+
}\) is designated ``lambdaplus," and  \(\Lambda _{-
}\) is designated ``lambdaminus:"

\ \\ \noindent\(\pmb{\text{lambdaplus}\text{:=}}
\pmb{\text{lambdaplus}=}
\pmb{\left.\left(224\alpha ^2\right)\right/}
\pmb{\left(495+12\text{Fspecialcase}\ \alpha ^2-\right.}\\
\pmb{\text{Sqrt}\left[18225-18360\text{Fspecialcase}\ \alpha ^2+\right.}
\pmb{\left.\left.4624\text{Fspecialcase}^2\alpha ^4+3136\text{Gspecialcase}^2\alpha ^4\right]\right)}\)\\
\pmb{}\\
\ \\ \noindent\(\pmb{\text{lambdaminus}\text{:=}}
\pmb{\text{lambdaminus}=}
\pmb{\left.\left(224\alpha ^2\right)\right/}
\pmb{\left(495+12\text{Fspecialcase}\ \alpha ^2+\right.}\\
\pmb{\text{Sqrt}\left[18225-18360\text{Fspecialcase}\ \alpha ^2+\right.}
\pmb{\left.\left.4624\text{Fspecialcase}^2\alpha ^4+3136\text{Gspecialcase}^2\alpha ^4\right]\right)}\)
\
\\

 The choice of $\pm$ in the calculation of $\Lambda_\pm$ depends on the polarization state of the photon. We call these the $+$ and $-$ polarization states (corresponding to $\Lambda_+$ and $\Lambda_-$, respectively).

Using Equation (\ref{radial2}), we get that for outgoing geodesics in the $+$ polarization state, $dr/dt$ expanded as a series in $\alpha
$ is:

\ \\ \noindent\(\pmb{\text{radA1} = }
\pmb{\text{FullSimplify}[}
\pmb{\text{Series}[}\\
\pmb{(-D[u[t,r],t]D[u[t,r],r]- }
\pmb{r^2D[v[t,r],t]D[v[t,r],r] +}\\
\pmb{\text{Sqrt}[}
\pmb{(D[u[t,r],t]D[u[t,r],r]+ }
\pmb{\left.r^2D[v[t,r],t]D[v[t,r],r] \right)^2 +}\\
\pmb{\left(\frac{r^2}{\text{lambdaplus}}-D[u[t,r],r]^2-r^2D[v[t,r],r]^2\right)}\\
\pmb{\left.\left.\left.\left(\frac{r^2}{\text{lambdaplus}}+D[u[t,r],t]^2+r^2D[v[t,r],t]^2\right)\right]\right)\right/}\\
\pmb{\left(\frac{r^2}{\text{lambdaplus}} + D[u[t,r],t]^2+r^2D[v[t,r],t]^2\right),}\\
\pmb{\left.\{\alpha ,0,3\}],r^2\geq 0\right]}\)

\begin{eqnarray}\label{A.8}1-\frac{1}{45 r^2}14 \left(\left(u^{(0,1)}[t,r]+u^{(1,0)}[t,r]\right)^2+\right.
\left.r^2 \left(v^{(0,1)}[t,r]+v^{(1,0)}[t,r]\right)^2\right) \alpha ^2+O[\alpha ]^4
\end{eqnarray}

For  outgoing geodesics in the $-$ polarization state:

\ \\ \noindent\(\pmb{\text{radA2} = }
\pmb{\text{FullSimplify}[}
\pmb{\text{Series}[}\\
\pmb{(-D[u[t,r],t]D[u[t,r],r]- }
\pmb{r^2D[v[t,r],t]D[v[t,r],r] +}\\
\pmb{\text{Sqrt}[}
\pmb{(D[u[t,r],t]D[u[t,r],r]+ }
\pmb{\left.r^2D[v[t,r],t]D[v[t,r],r] \right)^2 +}\\
\pmb{\left(\frac{r^2}{\text{lambdaminus}}-D[u[t,r],r]^2-r^2D[v[t,r],r]^2\right)}\\
\pmb{\left.\left.\left.\left(\frac{r^2}{\text{lambdaminus}}+D[u[t,r],t]^2+r^2D[v[t,r],t]^2\right)\right]\right)\right/}\\
\pmb{\left(\frac{r^2}{\text{lambdaminus}} + D[u[t,r],t]^2+r^2D[v[t,r],t]^2\right),}\\
\pmb{\left.\{\alpha ,0,3\}],r^2\geq 0\right]}\)

\begin{eqnarray}\label{A.9}1-\frac{1}{45 r^2}8 \left(\left(u^{(0,1)}[t,r]+u^{(1,0)}[t,r]\right)^2+\right.
\left.r^2 \left(v^{(0,1)}[t,r]+v^{(1,0)}[t,r]\right)^2\right) \alpha ^2+O[\alpha ]^4
\end{eqnarray}

Outputs (\ref{A.8}) and (\ref{A.9}) imply Equation (\ref{out2o}).\\
\\
For the ingoing geodesics with $+$ polarization:

\ \\ \noindent\(\pmb{\text{radB1} = }
\pmb{\text{FullSimplify}[}
\pmb{\text{Series}[}\\
\pmb{(-D[u[t,r],t]D[u[t,r],r]- }
\pmb{r^2D[v[t,r],t]D[v[t,r],r] -}\\
\pmb{\text{Sqrt}[}
\pmb{(D[u[t,r],t]D[u[t,r],r]+ }
\pmb{\left.r^2D[v[t,r],t]D[v[t,r],r] \right)^2 +}\\
\pmb{\left(\frac{r^2}{\text{lambdaplus}}-D[u[t,r],r]^2-r^2D[v[t,r],r]^2\right)}\\
\pmb{\left.\left.\left.\left(\frac{r^2}{\text{lambdaplus}}+D[u[t,r],t]^2+r^2D[v[t,r],t]^2\right)\right]\right)\right/}\\
\pmb{\left(\frac{r^2}{\text{lambdaplus}} + D[u[t,r],t]^2+r^2D[v[t,r],t]^2\right),}\\
\pmb{\left.\{\alpha ,0,3\}],r^2\geq 0\right]}\)

\begin{eqnarray}\label{A.10}-1+\frac{1}{45 r^2}14 \left(\left(u^{(0,1)}[t,r]-u^{(1,0)}[t,r]\right)^2+\right.
\left.r^2 \left(v^{(0,1)}[t,r]-v^{(1,0)}[t,r]\right)^2\right) \alpha ^2+O[\alpha ]^4\nonumber\\
\end{eqnarray}

For the ingoing geodesics with $-$ polarization:

\ \\ \noindent\(\pmb{\text{radB2}= }
\pmb{\text{FullSimplify}[}
\pmb{\text{Series}[}\\
\pmb{(-D[u[t,r],t]D[u[t,r],r]- }
\pmb{r^2D[v[t,r],t]D[v[t,r],r] -}\\
\pmb{\text{Sqrt}[}
\pmb{(D[u[t,r],t]D[u[t,r],r]+ }
\pmb{\left.r^2D[v[t,r],t]D[v[t,r],r] \right)^2 +}\\
\pmb{\left(\frac{r^2}{\text{lambdaminus}}-D[u[t,r],r]^2-r^2D[v[t,r],r]^2\right)}\\
\pmb{\left.\left.\left.\left(\frac{r^2}{\text{lambdaminus}}+D[u[t,r],t]^2+r^2D[v[t,r],t]^2\right)\right]\right)\right/}\\
\pmb{\left(\frac{r^2}{\text{lambdaminus}} + D[u[t,r],t]^2+r^2D[v[t,r],t]^2\right),}\\
\pmb{\left.\{\alpha ,0,3\}],r^2\geq 0\right]}\)

\begin{eqnarray}\label{A.11}-1+\frac{1}{45 r^2}8 \left(\left(u^{(0,1)}[t,r]-u^{(1,0)}[t,r]\right)^2+\right.
\left.r^2 \left(v^{(0,1)}[t,r]-v^{(1,0)}[t,r]\right)^2\right) \alpha ^2+O[\alpha ]^4\nonumber\\
\end{eqnarray}

Outputs (\ref{A.10}) and (\ref{A.11}) imply Equation (\ref{in2o}).

\section{The field equations}
Our next task is to derive the nonlinear PDEs (\ref{system4}) which arise from the Euler-Heisenberg field equation (\ref{EHFE}) together with the cylindrical field ansatz presently under   consideration. Since we wish to express these PDEs in terms of the functions $v$ and $\hat u = u/r$, we will go back and re-enter the field tensor, so that \emph{
in the present section, $u[t,r]$ should be read as {``}$\ \hat u[t,r]${''}}. To this end, we will clear out and recompute the relevant variables:

\ \\ \noindent\(\pmb{\text{Clear}[\text{Afield},\text{faraday},\text{faraday1},\text{faraday2},\text{Fspecialcase},}\\
\pmb{\text{maxwell},\text{Gspecialcase},\text{maxwell1},}
\pmb{\text{maxwell2}]}\)

We re-enter the $A$-field as:

\ \\ \noindent\(\pmb{\text{Afield}\text{:=}\text{Afield}=\{0,0,r\ u[t,r],v[t,r]\}}\)
\
\\

Recalculating $F_{\mu\nu}$:

\ \\ \noindent\(\pmb{\text{faraday}\text{:=}}
\pmb{\text{faraday}=}
\pmb{\text{Table}[D[\text{Afield}[[j]],\text{coord}[[i]]]-}\\
\pmb{D[\text{Afield}[[i]],\text{coord}[[j]]],\{i,1,4\},\{j,1,4\}]}\)
\
\\

$F^\mu_{\phantom{\mu}\nu}$:

\ \\ \noindent\(\pmb{\text{faraday1}\text{:=}}
\pmb{\text{faraday1}=}
\pmb{\text{FullSimplify}[}
\pmb{\text{Table}[\text{Sum}[\text{cometric}[[i,k]]\text{faraday}[[k,j]],}\\
\pmb{\{k,1,4\}],\{i,1,4\},\{j,1,4\}]]}\)
\
\\

$F^{\mu\nu}$:

\ \\ \noindent\(\pmb{\text{faraday2}\text{:=}}
\pmb{\text{faraday2}=}
\pmb{\text{FullSimplify}[}
\pmb{\text{Table}[\text{Sum}[\text{cometric}[[k,j]]\text{faraday1}[[i,k]],}\\
\pmb{\{k,1,4\}],\{i,1,4\},\{j,1,4\}]]}\)
\
\\

 $F$:
 
\ \\ \noindent\(\pmb{\text{Fspecialcase}\text{:=}}
\pmb{\text{Fspecialcase}=}
\pmb{\text{Simplify}[\text{Sum}[\text{faraday}[[i,j]]\text{faraday2}[[i,j]],}\\
\pmb{\{i,1,4\},\{j,1,4\}]]}\)
\
\\

$F^*_{\mu\nu}$:

\ \\ \noindent\(\pmb{\text{maxwell}\text{:=}}
\pmb{\text{maxwell}=\text{FullSimplify}[}
\pmb{\text{Table}\left[\frac{1}{2}\text{Sqrt}[-\text{Det}[\text{metric}]]\right.}\\
\pmb{\text{Sum}[\text{Signature}[\{i,j,k,l\}]\text{faraday2}[[i,j]],}
\pmb{\{i,1,4\},\{j,1,4\}],\{k,1,4\},\{l,1,4\}],}
\pmb{\ r\geq 0]}\)
\
\\

$G$:

\ \\ \noindent\(\pmb{\text{Gspecialcase}\text{:=}}
\pmb{\text{Gspecialcase}=\text{Sum}[\text{maxwell}[[i,j]]\text{faraday2}[[i,j]],}\\
\pmb{\{i,1,4\},\{j,1,4\}]}\)
\
\\

${F^*}^\mu_{\phantom{\mu}\nu}$:

\ \\ \noindent\(\pmb{\text{maxwell1}\text{:=}}
\pmb{\text{maxwell1}=}
\pmb{\text{FullSimplify}[}\\
\pmb{\text{Table}[\text{Sum}[\text{cometric}[[i,k]]\text{maxwell}[[k,j]],}
\pmb{\{k,1,4\}],\{i,1,4\},\{j,1,4\}]]}\)
\
\\

${F^*}^{\mu\nu}$:

\ \\ \noindent\(\pmb{\text{maxwell2}\text{:=}}
\pmb{\text{maxwell2}=}
\pmb{\text{FullSimplify}[}\\
\pmb{\text{Table}[\text{Sum}[\text{cometric}[[k,j]]\text{maxwell1}[[i,k]],}
\pmb{\{k,1,4\}],\{i,1,4\},\{j,1,4\}]]}\)
\
\\

Next, we calculate and display the covariant derivative \(\nabla _{\mu }F^{\mu \nu }\) by entering:

\ \\ \noindent\(\pmb{\text{CDfaraday2}\text{:=}}
\pmb{\text{CDfaraday2}=}
\pmb{\text{FullSimplify}[}\\
\pmb{\text{Table}[\text{Sum}[D[\text{faraday2}[[s,o]],\text{coord}[[s]]],}
\pmb{\{s,1,4\}] + }
\pmb{\text{Sum}[\text{affine}[[i,i,j]]\text{faraday2}[[j,o]],}\\
\pmb{\{i,1,4\},\{j,1,4\}] + }
\pmb{\text{Sum}[\text{affine}[[o,k,l]]\text{faraday2}[[k,l]],}
\pmb{\ \{k,1,4\},\{l,1,4\}],\{o,1,4\}]]}\)\\
\pmb{}\\
\ \\ \noindent\(\pmb{\text{MatrixForm}[\text{CDfaraday2}]}\)

\begin{eqnarray}\left(
\begin{array}{c}
 0 \\
 0 \\
 \frac{-u[t,r]+r \left(u^{(0,1)}[t,r]+r \left(u^{(0,2)}[t,r]-u^{(2,0)}[t,r]\right)\right)}{r^3} \\
 \frac{v^{(0,1)}[t,r]}{r}+v^{(0,2)}[t,r]-v^{(2,0)}[t,r]
\end{array}
\right)
\end{eqnarray}

This matches Equation (\ref{lefthandside}). \\
\\
Now we must calculate \(\frac{\alpha ^2}{45}\left(4\nabla _{\mu }\left(F F^{\mu \nu }\right)+7\nabla _{\mu }\left(G F^{*\mu \nu }\right)\right)\), which we shall call ``righthandside."  
To this end, we define \(F F^{\mu \nu }\) and \(G F^{*\mu \nu }\) as  variables ``Ffaraday2" and ``Gmaxwell2" respectively,  then we take
their covariant derivatives ``CDFfaraday2" and ``CDGmaxwell2." Finally, we combine these so as to calculate  {``}righthandside.{''}

We define  \(F F^{\mu \nu }\) by entering:

\ \\ \noindent\(\pmb{\text{Ffaraday2}\text{:=}\text{Ffaraday2}=\text{Fspecialcase}\  \text{faraday2}}\)
\
\\

Taking the covariant derivative \(\nabla_\mu (F F^{\mu \nu })\):

\ \\ \noindent\(\pmb{\text{CDFfaraday2}\text{:=}}
\pmb{\text{CDFfaraday2}=}
\pmb{\text{FullSimplify}[}\\
\pmb{\text{Table}[\text{Sum}[D[\text{Ffaraday2}[[s,o]],\text{coord}[[s]]],}
\pmb{\{s,1,4\}] + }
\pmb{\text{Sum}[\text{affine}[[i,i,j]]\text{Ffaraday2}[[j,o]],}\\
\pmb{\{i,1,4\},\{j,1,4\}] + }
\pmb{\text{Sum}[\text{affine}[[o,k,l]]\text{Ffaraday2}[[k,l]],}
\pmb{\ \{k,1,4\},\{l,1,4\}],\{o,1,4\}]]}\)
\
\\

Defining \(G F^{*\mu \nu }\):

\ \\ \noindent\(\pmb{\text{Gmaxwell2}\text{:=}\text{Gmaxwell2}=\text{Gspecialcase} \ \text{maxwell2}}\)
\
\\

Taking the covariant derivative \(\nabla _{\mu }\left(G F^{*\mu \nu }\right)\):

\ \\ \noindent\(\pmb{\text{CDGmaxwell2}\text{:=}}
\pmb{\text{CDGmaxwell2}=}
\pmb{\text{FullSimplify}[}\\
\pmb{\text{Table}[\text{Sum}[D[\text{Gmaxwell2}[[s,o]],\text{coord}[[s]]],}
\pmb{\{s,1,4\}] + }
\pmb{\text{Sum}[\text{affine}[[i,i,j]]\text{Gmaxwell2}[[j,o]],}\\
\pmb{\{i,1,4\},\{j,1,4\}] + }
\pmb{\text{Sum}[\text{affine}[[o,k,l]]\text{Gmaxwell2}[[k,l]],}
\pmb{\ \{k,1,4\},\{l,1,4\}],\{o,1,4\}]]}\)
\
\\

We compute  \(\frac{\alpha ^2}{45}\left(4\nabla
_{\mu }\left(F F^{\mu \nu }\right)+7\nabla _{\mu }\left(G F^{*\mu \nu }\right)\right)\) by entering:

\ \\ \noindent\(\pmb{\text{righthandside}\text{:=}}
\pmb{\text{righthandside}=}
\pmb{\text{FullSimplify}\left[\frac{4\alpha ^2}{45}\text{CDFfaraday2} + \frac{7\alpha ^2}{45}\text{CDGmaxwell2}\right]}\)
\
\\

We note that Equation (\ref{righthandside}) checks out since:

\ \\ \noindent\(\pmb{\text{righthandside}[[1]]}\)
\
\\

0

and:

\ \\ \noindent\(\pmb{\text{righthandside}[[2]]}\)
\
\\

0

moreover $\mathcal{U}$ is given by:

\ \\ \noindent\(\pmb{\text{righthandside}[[3]]*\left(45 r^5/\left(4\alpha ^2\right)\right)}\)

\begin{eqnarray}-6 u[t,r]^3+2 r u[t,r]^2
\left(-3 u^{(0,1)}[t,r]+3 r u^{(0,2)}[t,r]-r u^{(2,0)}[t,r]\right)+  \nonumber\\
r^2 u[t,r] 
\left(6 u^{(0,1)}[t,r]^2-2 v^{(0,1)}[t,r]^2-5 v^{(1,0)}[t,r]^2-\right. \nonumber\\
2 u^{(1,0)}[t,r] \left(u^{(1,0)}[t,r]+4 r u^{(1,1)}[t,r]\right)+ \nonumber\\
3 r v^{(1,0)}[t,r] v^{(1,1)}[t,r]+ 
4 r u^{(0,1)}[t,r] \left(3 u^{(0,2)}[t,r]-u^{(2,0)}[t,r]\right)+ \nonumber\\
\left.r v^{(0,1)}[t,r] \left(4 v^{(0,2)}[t,r]-7 v^{(2,0)}[t,r]\right)\right)+ \nonumber\\
r^3 \left(6 u^{(0,1)}[t,r]^3+\right.
v^{(0,1)}[t,r] \left(-7 v^{(1,0)}[t,r] \left(u^{(1,0)}[t,r]+2 r u^{(1,1)}[t,r]\right)+\right. \nonumber\\
\left.3 r u^{(1,0)}[t,r] v^{(1,1)}[t,r]\right)+ 
r v^{(0,1)}[t,r]^2 \left(2 u^{(0,2)}[t,r]+5 u^{(2,0)}[t,r]\right)+ \nonumber\\
u^{(0,1)}[t,r]^2 \left(6 r u^{(0,2)}[t,r]-2 r u^{(2,0)}[t,r]\right)+ \nonumber\\
u^{(0,1)}[t,r] \left(2 v^{(0,1)}[t,r]^2-6 u^{(1,0)}[t,r]^2-\right. \nonumber\\
8 r u^{(1,0)}[t,r] u^{(1,1)}[t,r]+ 
v^{(1,0)}[t,r] \left(5 v^{(1,0)}[t,r]+3 r v^{(1,1)}[t,r]\right)+ \nonumber\\
\left.r v^{(0,1)}[t,r] \left(4 v^{(0,2)}[t,r]-7 v^{(2,0)}[t,r]\right)\right)+ \nonumber\\
r \left(-7 v^{(0,2)}[t,r] u^{(1,0)}[t,r] v^{(1,0)}[t,r]+\right. \nonumber\\
u^{(0,2)}[t,r] \left(-2 u^{(1,0)}[t,r]^2+5 v^{(1,0)}[t,r]^2\right)+ \nonumber\\
2 \left(3 u^{(1,0)}[t,r]^2+v^{(1,0)}[t,r]^2\right) u^{(2,0)}[t,r]+ \nonumber
\left.\left.4 u^{(1,0)}[t,r] v^{(1,0)}[t,r] v^{(2,0)}[t,r]\right)\right)\end{eqnarray}

and  $\mathcal{V}$ is given by:

\ \\ \noindent\(\pmb{\text{righthandside}[[4]]*\left(45r^3/\left(4\alpha ^2\right)\right)}\)

\begin{eqnarray}u[t,r]^2  
\left(-2 v^{(0,1)}[t,r]+2 r v^{(0,2)}[t,r]+5 r v^{(2,0)}[t,r]\right)+ \nonumber\\
r u[t,r] 
\left(v^{(1,0)}[t,r] \left(10 u^{(1,0)}[t,r]+3 r u^{(1,1)}[t,r]\right)-\right. \nonumber\\
14 r u^{(1,0)}[t,r] v^{(1,1)}[t,r]+ 
r v^{(0,1)}[t,r] \left(4 u^{(0,2)}[t,r]-7 u^{(2,0)}[t,r]\right)+ \nonumber\\
2 u^{(0,1)}[t,r] \left(2 v^{(0,1)}[t,r]+\right. 
\left.\left.2 r v^{(0,2)}[t,r]+5 r v^{(2,0)}[t,r]\right)\right)+ \nonumber\\
r^2 \left(2 v^{(0,1)}[t,r]^3- \right.  
v^{(0,1)}[t,r]\left(2 u^{(1,0)}[t,r]^2-3 r u^{(1,0)}[t,r] u^{(1,1)}[t,r]+\right. \nonumber\\
\left.2 v^{(1,0)}[t,r] \left(v^{(1,0)}[t,r]+4 r v^{(1,1)}[t,r]\right)\right)+ \nonumber\\
u^{(0,1)}[t,r] \left(3 r v^{(1,0)}[t,r] u^{(1,1)}[t,r]-\right. 
2 u^{(1,0)}[t,r] \left(2 v^{(1,0)}[t,r]+7 r v^{(1,1)}[t,r]\right)+ \nonumber\\
\left.r v^{(0,1)}[t,r] \left(4 u^{(0,2)}[t,r]-7 u^{(2,0)}[t,r]\right)\right)+ \nonumber\\
v^{(0,1)}[t,r]^2 \left(6 r v^{(0,2)}[t,r]-2 r v^{(2,0)}[t,r]\right)+ \nonumber\\
u^{(0,1)}[t,r]^2 
\left(6 v^{(0,1)}[t,r]+2 r v^{(0,2)}[t,r]+5 r v^{(2,0)}[t,r]\right)+ \nonumber\\
r \left(v^{(0,2)}[t,r] \left(5 u^{(1,0)}[t,r]^2-2 v^{(1,0)}[t,r]^2\right)+\right. \nonumber\\
u^{(1,0)}[t,r] v^{(1,0)}[t,r]  
\left(-7 u^{(0,2)}[t,r]+4 u^{(2,0)}[t,r]\right)+ \nonumber\\
\left.2 \left(u^{(1,0)}[t,r]^2+3 v^{(1,0)}[t,r]^2\right) v^{(2,0)}[t,r]\right)\end{eqnarray}

\section{The Maxwellian approximation}
Our next task is to check Equations  (\ref{inMEave}) - (\ref{outMCave}). Introducing constants $U$ and $V$ (not to be confused with the functions $\mathcal{U}$ and $\mathcal{V}$ given above), we have that an elliptically polarized ingoing cylindrical wave is given
by $u[t,r]$ and $v[t,r]$, where:

\ \\ \noindent\(\pmb{u[t,r]\text{:=}}
\pmb{\frac{U r}{\omega }}
\pmb{(\text{BesselJ}[1,\omega  r]\text{Cos}[\omega  t] - \text{BesselY}[1,\omega  r]\text{Sin}[\omega  t])}\)
\
\\

\ \\ \noindent\(\pmb{v[t,r]\text{:=}}
\pmb{\frac{V}{\omega }(\text{BesselJ}[0,\omega  r]\text{Cos}[\omega  t] - }
\pmb{\text{BesselY}[0,\omega  r]\text{Sin}[\omega  t])}\)
\
\\

Note that \emph{the function $u[t,r]$ now reverts  back to denoting $u\ (=A_\theta)$ again, instead of $\hat u$}. 

In order to calculate $dr/dt$ for ingoing radial null geodesics, the equations of Section \ref{5.1.3} require us to calculate $\frac{1}{r^2}\left((\partial_t u - \partial_r u)^2 + r^2 (\partial_t v -\partial_r v)^2\right)$. Thereby we enter the lines:

\ \\ \noindent\(\pmb{\text{FullSimplify}[}\\
\pmb{\frac{1}{r^2}\left((D[u[t,r],t]-D[u[t,r],r])^2 + \right.}
\pmb{\left.\left.r^2(D[v[t,r],t]-D[v[t,r],r])^2\right)\right]}\)

\begin{eqnarray}\label{A.15}U^2 ((\text{BesselJ}[0,r \omega ]+\text{BesselY}[1,r \omega ]) \text{Cos}[t \omega ]+\nonumber\\
(\text{BesselJ}[1,r \omega ]-\text{BesselY}[0,r \omega ]) \text{Sin}[t \omega ])^2+\nonumber\\
V^2 ((-\text{BesselJ}[1,r \omega ]+\text{BesselY}[0,r \omega ]) \text{Cos}[t \omega ]+\nonumber\\
(\text{BesselJ}[0,r \omega ]+\text{BesselY}[1,r \omega ]) \text{Sin}[t \omega ])^2\end{eqnarray}
This output (\ref{A.15}) confirms Equation (\ref{inMEave}).

For the outgoing radial null geodesics, we need $\frac{1}{r^2}\left((\partial_t u + \partial_r u)^2 + r^2 (\partial_t v +\partial_r v)^2\right)$. To this end, we enter:

\ \\ \noindent\(\pmb{\text{FullSimplify}[}\\
\pmb{\frac{1}{r^2}\left((D[u[t,r],t]+D[u[t,r],r])^2 + \right.}\pmb{\left.\left.r^2(D[v[t,r],t]+D[v[t,r],r])^2\right)\right]}\)

\begin{eqnarray}\label{A.16}U^2 ((-\text{BesselJ}[0,r \omega ]+\text{BesselY}[1,r \omega ]) \text{Cos}[t \omega ]+\nonumber\\
(\text{BesselJ}[1,r \omega ]+\text{BesselY}[0,r \omega ]) \text{Sin}[t \omega ])^2+\nonumber\\
V^2 ((\text{BesselJ}[1,r \omega ]+\text{BesselY}[0,r \omega ]) \text{Cos}[t \omega ]+\nonumber\\
(\text{BesselJ}[0,r \omega ]-\text{BesselY}[1,r \omega ]) \text{Sin}[t \omega ])^2\end{eqnarray}

This  confirms  Equation (\ref{outMEave}).
 \\
\\
Henceforth, we restrict to the case of circular polarization, whereby $U = V = A =$ constant.

\ \\ \noindent\(\pmb{U\text{:=}A}\)

\ \\ \noindent\(\pmb{V\text{:=}A}\)
\
\\

We check Equations (5.57) and (5.58) by re-entering:

\ \\ \noindent\(\pmb{\text{FullSimplify}[}\\
\pmb{\frac{1}{r^2}\left((D[u[t,r],t]-D[u[t,r],r])^2 + \right.}
\pmb{\left.\left.r^2(D[v[t,r],t]-D[v[t,r],r])^2\right)\right]}\)

\begin{eqnarray}\label{A.17}A^2 \left(-\frac{4}{\pi  r \omega }+\text{BesselJ}[0,r \omega ]^2+\text{BesselJ}[1,r \omega ]^2+\right.\nonumber\\
\left.\phantom{\frac{4}{\pi  r \omega }}\text{BesselY}[0,r \omega ]^2+\text{BesselY}[1,r \omega ]^2\right)\end{eqnarray}

\ \\ \noindent\(\pmb{\text{FullSimplify}[}\\
\pmb{\frac{1}{r^2}\left((D[u[t,r],t]+D[u[t,r],r])^2 + \right.}
\pmb{\left.\left.r^2(D[v[t,r],t]+D[v[t,r],r])^2\right)\right]}\)

\begin{eqnarray}\label{A.18}A^2 \left(\frac{4}{\pi  r \omega }+\text{BesselJ}[0,r \omega ]^2+\text{BesselJ}[1,r \omega ]^2+\right.\nonumber\\
\left.\phantom{\frac{4}{\pi  r \omega }}\text{BesselY}[0,r \omega ]^2+\text{BesselY}[1,r \omega ]^2\right)\end{eqnarray}
The outputs (\ref{A.17}) and (\ref{A.18}) confirm Equations (\ref{inMCave}) and (\ref{outMCave}).\\

\chapter{Some additional remarks}
\label{Appendix:Key2}

In Section \ref{proofmaintheorem}, we studied effective geometries corresponding to  branch I of our exact solution $B + 8\alpha^2(E^2B - B^3)/45 = k/r$ (\ref{bIMP}), with certain restrictions on the constant $E$. The purpose of this appendix is to take a further look at the effective geometries of our exact solution. We do so only  briefly.   This appendix is intended to be read as a continuation of Section \ref{proofmaintheorem}; we still assume that $k>0$, and the $(+)$ and $(-)$ polarization states have the same meaning as in Section \ref{proofmaintheorem}. 

\begin{thm}For $E=0$, in the effective geometry of branch I corresponding to the $(+)$ polarization state, the outgoing radial null  geodesics, issuing from any point where the effective geometry is defined,  are never trapped.
\end{thm}

\begin{proof}
For the $(+)$ polarization state, with $E=0$, Equation (\ref{4.73}) gives:
\begin{eqnarray}\label{B1}
\frac{1}{\Lambda} &=& \frac{495 + 24\alpha^2B^2  + \Big| 135 - 136\alpha^2 B^2\Big|}{224\alpha^2}.
\end{eqnarray}We note that there is a particular radius $r_1$ such that for  $r\geq r_1$, we have $B^2 \leq 135/(136\alpha^2)$, and  for $r_s\leq r <r_1$, we have $135/(136\alpha^2)<B^2\leq B_s^2$. Since $E=0$ in the present case,  we have $B_s=15/(8\alpha^2)$.

So for $r\geq r_1$ Equation (\ref{B1}) gives:
\begin{eqnarray}
\frac{1}{\Lambda}\Big|_{r\geq r_1} &=&\frac{45}{16\alpha^2} - \frac{1}{2}B^2.
\end{eqnarray}Note that $1/\Lambda|_{r\geq r_1}>0$.

Using  Equation (\ref{4.79}), we get that:
\begin{eqnarray}
\frac{dr}{dt}\Big|_{\textrm{out},\ r\geq r_1}&=&\frac{-EB+\sqrt{\left(\frac{1}{\Lambda}\Big|_{r\geq r_1}\right)^2 -\left(\frac{1}{\Lambda}\Big|_{r\geq r_1}\right) \left(B^2-E^2\right)}}{E^2 +\frac{1}{\Lambda}\Big|_{r\geq r_1} }\nonumber\\
\nonumber\\
&=&\sqrt{1-B^2 \Lambda\Big|_{r\geq r_1}}.
\end{eqnarray}

For $r\geq r_1$, the value of $B^2$ never exceeds $135/(136\alpha^2)\ (<B_s^2)$, and we observe that this fact implies  $dr/dt\Big|_{\textrm{out},\ r\geq r_1}=\sqrt{1-B^2  \Lambda\Big|_{r\geq r_1}}>0$.

In the region where $r_s\leq r <r_1$, we have:
\begin{eqnarray}
\frac{1}{\Lambda}\Big|_{r_s\leq r<r_1} &=& \frac{45}{28\alpha^2}+\frac{5}{7}B^2.
\end{eqnarray} So:
\begin{eqnarray}
\frac{dr}{dt}\Big|_{\textrm{out},\ r_s\leq r<r_1}&=&\sqrt{1-B^2 \Lambda\Big|_{r_s\leq r<r_1}}.
\end{eqnarray}For $r_s\leq r< r_1$, the value of $B^2$ never exceeds $B_s^2$, and consequently $dr/dt\Big|_{\textrm{out},\ r_s<r\leq r_1}=\sqrt{1-B^2  \Lambda\Big|_{r_s<r\leq r_1}}>0$.

\end{proof}

\begin{thm}
For $E^2\leq \frac{45}{34\alpha^2}$, in the effective geometry of branch I  corresponding to the $(-)$ polarization state, either  $dr/dt_{\textrm{out},\ r=r_s}$ is zero (when $E>0$) or  $dr/dt_{\textrm{in,}\ r=r_s}$ is zero (if $E<0$). If $E=0$, then both   $dr/dt_{\textrm{out},\ r=r_s}$ and $dr/dt_{\textrm{in,}\ r=r_s}$ are zero.\end{thm}
\begin{proof}

Since  $E^2\leq \frac{45}{34\alpha^2}$, we have that  for the $(-)$ polarization state,   at $r=r_s$:
\begin{eqnarray}\label{B.9}\frac{1}{\Lambda}\Big|_{r=r_s} &=&B_s^2.
\end{eqnarray}
Using  (\ref{4.79}) and (\ref{B.9}), we get:
\begin{eqnarray}\label{5.88}
\frac{dr}{dt}\Big|_\textrm{out, $r=r_s$}&=& \frac{B_s\left(|E| - E\right)}{E^2 + B_s^2},
\end{eqnarray}and:
\begin{eqnarray}
\frac{dr}{dt}\Big|_{\textrm{in, $r=r_s$}} = -\frac{B_s\left(|E| + E\right)}{E^2 + B_s^2}.
\end{eqnarray}
If $E>0$, then   $dr/dt\Big|_\textrm{out, $r=r_s$}=0$, and $dr/dt\Big|_\textrm{in, $r=r_s$}<0$. 
  On the other hand, if $E<0$, then $dr/dt\Big|_\textrm{out, $r=r_s$}>0$, and $dr/dt\Big|_\textrm{in, $r=r_s$}=0$. If $E=0$, then $dr/dt\Big|_\textrm{out, $r=r_s$}=dr/dt\Big|_\textrm{in, $r=r_s$}=0$.
\end{proof}

Now let us say something about  branch II. In this branch, $B$ is an increasing function of $r$; it starts with the value of $B_s$ at $r=r_s$, and increases towards the asymptotic value $\sqrt{3}B_s$ at $r=\infty$.

\begin{thm}For $E\neq0$, in branch II, at infinity, the radial null geodesics propagate in only one direction; towards the axis if $E>0$, away from the axis if $E<0$.
\end{thm}

\begin{proof}
In the limit $r\rightarrow \infty$, we have $B=\sqrt{3} B_s$, and Equation  (\ref{4.73}) gives: 
\begin{eqnarray}
\frac{1}{\Lambda}\Big|_{r=\infty}&=&\frac{630 \mp 630}{224\alpha^2}\nonumber\\
\nonumber\\
&=&0 \textrm{ or } \frac{45}{8\alpha^2} \textrm{ (depending on the polarization)}.
\end{eqnarray}

Equation (\ref{4.79}) gives, for the radial null geodesics (both ``ingoing" and ``outgoing"):
\begin{eqnarray}\label{4.81}
\frac{dr}{dt}\Big|_{r=\infty}&=& \frac{-E\sqrt{E^2+\frac{45}{8\alpha^2}} \mp\sqrt{\frac{1}{\Lambda^2} -\frac{45}{\Lambda\alpha^2}}}{E^2 + \frac{1}{\Lambda}}\nonumber\\
\nonumber\\
&=&\frac{-E\sqrt{E^2+\frac{45}{8\alpha^2}}}{E^2 + \frac{1}{\Lambda}}\nonumber\\
\nonumber\\
&=&-\frac{\sqrt{E^2 + \frac{45}{8\alpha^2}}}{E}\textrm{ or }  -\frac{E}{\sqrt{E^2 + \frac{45}{8\alpha^2}}} \textrm{ (depending on the polarization)}.
\end{eqnarray} Both polarization states  travel inwards towards  the axis if $E>0$; outwards to $r=\infty$ if $E<0$. Note that this is consistent with  Conjecture \ref{conjecture1}. We also note that, at infinity, the coordinate speeds of the two polarization states are reciprocal to one another and  one polarization state propagates superluminally. 

\end{proof}

In branch III, the absolute value $|B|$ is an decreasing function of $r$;  at   $r=0$ we have $|B|\rightarrow\infty$, and as $r\rightarrow \infty$, we have $|B|\rightarrow \sqrt{3}B_s$. At $r=\infty$, the effective geometry corresponding to branch III is  similar the effective geometry of branch II.

\end{document}